\documentclass[reqno,twoside,a4paper,11pt]{article}
\usepackage[english]{babel}
\usepackage{color}
\usepackage{tocloft}
\usepackage{enumitem}

\usepackage{titlesec}

\usepackage{moresize}

\titleformat{\part}
  {\sf\huge\bfseries}{\thepart}{1em}{}

\usepackage[
  hmarginratio={1:1},     
  vmarginratio={1:1},     
  textwidth=16cm,        
  textheight=24cm,
  heightrounded,          
]{geometry}

\usepackage{amsmath,amsthm,amssymb,accents}
\usepackage{mathrsfs, mathtools}
\usepackage{dsfont}

\usepackage[all]{xy}
\input xy 
\xyoption{all}
\xyoption{2cell} 
\xyoption{v2}
\SelectTips{cm}{11} 

\usepackage{tikz-cd}

\usepackage[hidelinks]{hyperref}

\usepackage[numbers,sort&compress]{natbib}

\usepackage[mathscr]{euscript}

\newcommand{\ev}{\mathrm{ev}}
\newcommand{\av}{\mathrm{av}}
\newcommand{\pr}{\mathrm{pr}}
\newcommand{\rank}{\mathrm{rk}}
\newcommand{\inv}{\mathrm{inv}}

\newcommand{\ori}{\mathrm{or}}
\newcommand{\cl}{\mathrm{cl}}

\newcommand{\supp}{\mathrm{supp}}
\newcommand{\gr}{\mathrm{gr}}

\newcommand{\rmH}{\mathrm{H}}
\newcommand{\rmQ}{\mathrm{Q}}
\newcommand{\rmK}{\mathrm{K}}

\newcommand{\rmD}{\mathrm{D}}

\newcommand{\rmI}{\mathrm{I}}

\newcommand{\dd}{\mathrm{d}}
\newcommand{\End}{\mathrm{End}}
\newcommand{\trans}{\mathrm{t}}
\newcommand{\dR}{\mathrm{dR}}
\newcommand{\qm}{\mathrm{qu}}

\newcommand{\rmU}{\mathrm{U}}
\newcommand{\rmC}{\mathrm{C}}
\newcommand{\rmM}{\mathrm{M}}

\newcommand{\rmL}{\mathrm{L}}
\newcommand{\rmS}{\mathrm{S}}
\newcommand{\rmF}{\mathrm{F}}
\newcommand{\rmt}{\mathrm{t}}
\newcommand{\tr}{\mathrm{tr}}
\newcommand{\rms}{\mathrm{s}}
\newcommand{\intern}{\mathrm{int}}
\newcommand{\pfaff}{\mathrm{pfaff}}
\newcommand{\vol}{\mathrm{vol}}

\newcommand{\scG}{\mathscr{G}}

\newcommand{\scH}{\mathscr{H}}
\newcommand{\scT}{\mathscr{T}}

\newcommand{\scC}{\mathscr{C}}

\newcommand{\scD}{\mathscr{D}}
\newcommand{\scF}{\mathscr{F}}

\newcommand{\CG}{\mathcal{G}}
\newcommand{\CI}{\mathcal{I}}
\newcommand{\CH}{\mathcal{H}}
\newcommand{\CA}{\mathcal{A}}
\newcommand{\CB}{\mathcal{B}}

\newcommand{\CL}{\mathcal{L}}

\newcommand{\CS}{\mathcal{S}}

\newcommand{\CO}{\mathcal{O}}
\newcommand{\CV}{\mathcal{V}}

\newcommand{\sfG}{\mathsf{G}}
\newcommand{\sfH}{\mathsf{H}}
\newcommand{\sfP}{\mathsf{P}}

\newcommand{\sfN}{\mathsf{N}}
\newcommand{\sfu}{\mathsf{u}}

\newcommand{\sfm}{\mathsf{m}}
\newcommand{\sfR}{\mathsf{R}}
\newcommand{\sfU}{\mathsf{U}}
\newcommand{\sfE}{\mathsf{E}}

\newcommand{\sfS}{\mathsf{S}}
\newcommand{\sfA}{\mathsf{A}}
\newcommand{\hol}{\mathsf{hol}}

\newcommand{\sfB}{\mathsf{B}}

\newcommand{\sfc}{\mathsf{c}}
\newcommand{\sft}{\mathsf{t}}
\newcommand{\Aut}{\mathsf{Aut}}
\newcommand{\Aff}{\mathsf{Aff}}

\newcommand{\fru}{\mathfrak{u}}

\newcommand{\FC}{\mathbb{C}}
\newcommand{\FR}{\mathbb{R}}
\newcommand{\bbK}{\mathbb{K}}
\newcommand{\RZ}{\mathbb{Z}}
\newcommand{\NN}{\mathbb{N}}
\newcommand{\ZZ}{\mathbb{Z}_2}

\newcommand{\TT}{\mathbb{T}}
\newcommand{\PP}{\mathbb{P}}
\newcommand{\One}{\mathds{1}}

\newcommand{\HLBdl}{\mathscr{HL}\hspace{-0.02cm}\mathscr{B}\mathrm{dl}}

\newcommand{\BGrb}{\mathscr{B}\hspace{-.025cm}\mathscr{G}\mathrm{rb}}

\newcommand{\HVBdl}{\mathscr{HVB}\mathrm{dl}}
\newcommand{\QVBdl}{\mathscr{QVB}\mathrm{dl}}
\newcommand{\Mfd}{\mathscr{M}\mathrm{fd}}
\newcommand{\TopGrp}{\mathscr{T}\mathrm{op}\mathscr{G}\mathrm{rp}}
\newcommand{\TopAb}{\mathscr{T}\mathrm{op}\mathscr{A}\mathrm{b}}
\newcommand{\GapSys}{\mathscr{G}\mathrm{ap}\mathscr{S}\mathrm{ys}}
\newcommand{\TP}{\mathscr{TP}}
\newcommand{\RTP}{\mathscr{RTP}}

\newcommand{\arisom}{\overset{\cong}{\longrightarrow}}

\newcommand{\dslash}{/\hspace{-0.1cm}/}

\newcommand{\iu}{\mathrm{i}}

\newcommand{\qen}{\hfill$\triangleleft$}

\newcommand{\<}{\langle}
\renewcommand{\>}{\rangle}

\newcommand{\brtimes}{{\times}}

\newenvironment{myenumerate}{\begin{enumerate}[itemsep=-0.1cm, leftmargin=*, topsep=0cm, label=(\arabic*)]}{\end{enumerate}}

\newcommand{\textint}{{\textstyle{\int} \hspace{-0.1cm}}}

\newtheoremstyle{thm} 								
{0.25cm}   					
{0.2cm}   	 				
{\itshape} 	
{}         				
{\bfseries}				
{}        				
{0.2cm} 					
{\thmname{#1}~\thmnumber{#2}\thmnote{ (#3)}}%

\newtheoremstyle{rmk} 								
{0.25cm}   					
{0.2cm}   	 				
{} 	
{}         				
{\bfseries}				
{}        				
{0.2cm} 					
{}         				

\theoremstyle{thm}
\newtheorem{theorem}[equation]{Theorem}

\newtheorem{corollary}[equation]{Corollary}

\newtheorem{lemma}[equation]{Lemma}
\newtheorem{proposition}[equation]{Proposition}
\newtheorem{definition}[equation]{Definition}

\theoremstyle{rmk}
\newtheorem{example}[equation]{Example}
\newtheorem{remark}[equation]{Remark}

\numberwithin{equation}{section}

\begin{document}

\begin{flushright}
\small
Hamburger Beitr\"age zur Mathematik Nr. 713\\
ZMP--HH/17--30\\
EMPG--17--22
\end{flushright}

\vskip 1cm

%
%


\begin{center}
	\textbf{\LARGE{Topological Insulators and the Kane-Mele Invariant: \\[0.3cm] 
            Obstruction and Localisation Theory}}\\
	
	\vspace{1.5cm}
	
	{\Large Severin Bunk}

\vspace{0.2cm}

\noindent{\em Fachbereich Mathematik, Bereich Algebra und Zahlentheorie\\ Universit\"{a}t Hamburg\\ Bundesstra\ss e 55, 20146 Hamburg, Germany}\\
Email: \ {\tt Severin.Bunk@uni-hamburg.de}

\vspace{0.75cm}

{\Large Richard J. Szabo}

\vspace{0.2cm}

\noindent{\em Department of Mathematics, Heriot-Watt
  University\\
Colin Maclaurin Building, Riccarton, Edinburgh EH14 4AS, UK\\ 
Maxwell Institute for Mathematical Sciences, Edinburgh, UK\\
The Higgs Centre for Theoretical Physics, Edinburgh, UK}\\
Email: \ {\tt
  R.J.Szabo@hw.ac.uk}

\vspace{2cm}

\begin{abstract}
\noindent
We present homotopy theoretic and geometric interpretations of the
Kane-Mele invariant for gapped fermionic quantum systems in three
dimensions with time-reversal symmetry.
We show that the invariant is related to a certain $4$-equivalence which lends it an interpretation as an obstruction to a block decomposition of the sewing matrix up to nonequivariant homotopy.
We prove a Mayer-Vietoris Theorem for manifolds with $\ZZ$-actions
which intertwines Real and $\ZZ$-equivariant de~Rham cohomology
groups, and apply it to derive a new localisation formula for the
Kane-Mele invariant. This provides a unified cohomological explanation
for the equivalence between the discrete Pfaffian formula and the
known local geometric computations of the index for periodic lattice systems.
We build on the relation between the Kane-Mele invariant and the theory of bundle
gerbes with $\ZZ$-actions to obtain geometric refinements
of this obstruction and localisation technique. In the preliminary part we review the Freed-Moore theory of band insulators on Galilean
spacetimes with emphasis on geometric constructions, and present a bottom-up approach to time-reversal symmetric topological phases.
\end{abstract}

\end{center}

\newpage

\tableofcontents

\newpage

\part{Overview}

In this partly expository paper we shall derive some novel topological and geometric interpretations
of the Kane-Mele invariant of time-reversal symmetric topological
insulators in three dimensions, and in particular offer unified geometric explanations for the equivalence between discrete and integral representations of the invariant by deriving new
localisation formulas for it over invariant periodic crystal
momenta. We set the background and motivation behind our analysis in
Section~\ref{sect:Introduction}, and also briefly outline the
structure of this paper.
In Section~\ref{sect:Definitions and Results} we provide an extended overview of the main definitions and results of this work.
Our intent in this first part is to enable the reader to access the main results of
the present paper without recourse to the technical details which are provided in the later sections.

\section{Introduction}
\label{sect:Introduction}

The relevance of topology to condensed matter physics was first
realised in theoretical investigations of the integer quantum
Hall effect, whereby the quantised Hall conductance is determined mathematically
by the first Chern class of a line bundle over the $2$-dimensional Brillouin torus in momentum space~\cite{TKNN} and
represents a topological invariant measuring the obstruction to
defining Bloch wavefunctions of the system globally.
Only recently has the full scope of applications of mathematical tools to condensed matter physics started to emerge and to receive considerable attention.
The appreciation and impact of the field of research whose foundations
were laid down in~\cite{TKNN}, currently referred to as \emph{topological phases of matter}, is reflected in the award of the 2016
Nobel Prize in physics.
The general classification of topological phases and their physical
features in electron band theory is at present largely an open question of widespread interest in
both physics and mathematics, and many different approaches and
techniques are currently under investigation; a sample partial list of
references is
e.g.~\cite{AiKongZheng,Freed-Moore--TEM,Freed-Hopkins:Reflection_positivity,Freed:SRE_and_Invertible_TFTs,KWZ:Braided_fusion_and_top_orders,BCR:Noncommutative_framework_for_TIs,Mathai2017,Prodan-Schulz-Baldes:Bulk_and_bdry_invariants}.

One of the most prominent topological phases of quantum matter are
topological insulators, see e.g.~\cite{Hasan:2010xy,Qi:2011zya} for
reviews. Loosely speaking, these are materials with an energy gap in
their bulk electronic band structure. The integer quantum Hall state
is a topological insulator without additional symmetries which is
topologically protected by its Hall conductance. More generally, a
topological insulator is a gapped quantum system which gives rise to
invariants which are topologically protected by the presence of
additional symmetries, such as time-reversal; their boundary may not
have the symmetries of the bulk, which can lead to conducting surface
states of electrons closing the gap between bulk energy bands.
The features of band insulators most promising for applications, for
instance in materials science and quantum computing, become apparent when these systems are considered not individually but altogether.
In particular, it is an important question whether at a continuous transition between two given band insulators the spectral gap of their Hamiltonians, i.e. their band gap, closes.
Such transitions are easily realised by joining two band insulators
along their boundary surface.
Due to electron displacement caused by differing energy levels, a continuous transitional region forms near their interface.
If the band gap closes, a thin conducting layer emerges at the
interface allowing electrons occupying valence bands to jump into the
empty conduction bands.
Many qualitative physical properties of a band insulator,
such as this conductance property, depend only on the existence of a band gap and on how it splits the spectrum of the Hamiltonian.
Consequently these properties are stable under small deformations or
perturbations of the systems, such as those given by
introducing impurities or defects in the material.
Stability features of this type are particularly relevant to practical
uses since at least small impurities or defects are usually unavoidable in
real-world materials, even under laboratory conditions. In this paper
we deal only with band theory for systems of noninteracting fermions on
infinite crystals.

A big advancement in the field of topological phases was Kitaev's
classification of topological insulators in terms of topological K-theory in his famous periodic table~\cite{Kitaev:Periodic_table}, with the K-theory generated by Bloch bundles over the Brillouin zone.
This classification was subsequently refined and generalised in the
Freed-Moore theory~\cite{Freed-Moore--TEM} to include additional
symmetries, such as time-reversal and particle-hole (charge conjugation) symmetries, by
using twisted equivariant K-theory, where the equivariance accounts
for chiral (sublattice) symmetries while the twisting accounts for lattice symmetries and the crystallographic group of the Bravais lattice. At the same time this theory develops a mathematical foundation for gapped topological phases which emerge from a crystal in any space dimension.

While the classification of topological phases is a powerful tool,
some observable properties of band insulators are governed by
weaker topological invariants, such as the $\RZ$-valued Chern class giving the transverse Hall conductivity.
In the presence of additional $\ZZ$-symmetries, such as time-reversal
which makes the first Chern class of the Bloch bundle vanish, an
interesting invariant was discovered
in~\cite{Kane-Mele:Z2_top_order_and_QSHE,Fu-Kane:Time-reversal_polarisation,FKM}
and is usually refered to as the \emph{Kane-Mele invariant}; it is a
$\ZZ$-valued index indicating the presence or absence of a topological
insulating phase. The Kane-Mele invariant is defined by a discrete
Pfaffian formula which measures whether there exists a set of
compatible Kramers pairs over the whole Brillouin torus, and is interpreted as an obstruction to defining a global smooth and periodic time-reversal symmetric frame for the Bloch bundle, i.e. to defining global time-reversal symmetric Bloch wavefunctions for the occupied valence states. This invariant has attracted much attention, especially in two and three space dimensions.
In the seminal $2$-dimensional example of the quantum spin Hall
effect~\cite{spinHall}, it is related to the existence of
spin-oriented edge currents, such that each spin component has nonzero
conductance while the total current is zero, which are linked to topological invariants of Bloch bundles over the Brillouin zone.
Kane-Mele invariants for semimetals have also been discussed in~\cite{ThiangSatoGomi2017}.
A variety of mathematical approaches to the Kane-Mele invariant are known: they
range from algebraic topology
perspectives~\cite{Freed-Moore--TEM,Monaco-Tauber:Berry_WZW_FKM} over
index theory~\cite{Ertem,Li:Bott-Kitaev_table_and_index} and $C^*$-algebras~\cite{MathaiThiang2015a,MathaiThiang2015b} to geometric approaches~\cite{Gawedzki--2DFKM_and_WZW,Gawedzki:BGrbs_for_TIs,Gawedzki:2017whj}.
An overview relating various different treatments of the Kane-Mele
invariant is provided by~\cite{KLWK}.

In this paper we concentrate on the Kane-Mele invariant in three dimensions.
Starting from the local formula for the Kane-Mele
invariant in terms of an integral of a $3$-form over the Brillouin
torus, we present homotopy theoretic and higher geometric interpretations of the invariant.
These new perspectives shed further light on both the geometric meaning and the physical content of
the Kane-Mele invariant as certain obstructions.
We moreover develop equivariant localisation techniques for the
$3$-dimensional integral formulas for the Kane-Mele invariant that allow one to
compute it from data in all lower dimensions, eventually arriving at
new localisation formulas for the Kane-Mele invariant over the
discrete time-reversal invariant periodic momenta (or fixed points) in the Brillouin torus. This naturally explains, from a unified and purely geometric perspective, the equivalence between the integral invariant and the original discrete Pfaffian formula for the Kane-Mele invariant~\cite{WQZ}.
Along the way, we develop new techniques for computing equivariant de
Rham cohomology and bundle gerbe holonomy in the presence of
$\ZZ$-actions. Our considerations in fact apply more generally to a large class of
$3$-manifolds $X$ with $\ZZ$-actions and so are applicable to the
classification of protected topological phases of other time-reversal
symmetric gapped electron systems, described by a family of
Hamiltonians parameterised by $X$, which are not necessarily related
to lattice models for
topological insulators. Our main examples involve Bloch electrons, in
periodic quantum systems for
which the Brillouin zone is a torus, and free fermions in continuous
systems such as topological superconductors for which the
first Brillouin zone is a sphere, regarded as the momentum space of the continuous limit of a Bravais lattice. As such, some of our mathematical developments are of independent interest.

This paper is organised into three parts which can be read
independently, but which follow a mutual line of investigation. In the
remainder of this first part we shall provide an extended summary of
our main definitions and results, particularly those alluded to above,
in a more informal physics-oriented language which we hope is palatable to a wider audience, defering all technical details to the remainder of the paper.
In Part~\ref{part:TPMs} we give an exposition of the Freed-Moore theory of topological
phases of matter described by band insulators~\cite{Freed-Moore--TEM} as
tailored to our needs; this part
does not contain original results unless explicitly stated, but we elucidate several geometric
constructions and we deviate at points from the
line of argument in~\cite{Freed-Moore--TEM} to streamline the presentation. The main intent here is
to present a bottom-up approach to time-reversal symmetric topological phases, endeavouring to provide a
satisfying amount of detail and self-containment in our presentation
relevant to this paper. In Part~\ref{part:KM and localisation} we develop
three new perspectives on the $3$-dimensional Kane-Mele invariant
based on homotopy theory, $\ZZ$-equivariant de Rham cohomology, and the theory of bundle
gerbes, and we propose interpretations of the Kane-Mele invariant
based on all three approaches.

\subsection*{Acknowledgments}

We thank Alan Carey, Lukas M\"uller and Alexander Schenkel for helpful discussions.
This work was supported in part by the COST Action MP1405 ``Quantum
Structure of Spacetime''. 
S.B. is partially supported by the RTG~1670 ``Mathematics Inspired by String Theory and Quantum Field Theory''.
The work of R.J.S.\ is supported in part by the STFC Consolidated Grant
ST/P000363/1 ``Particle Theory at the Higgs Centre''. 

\section{Main definitions and results}
\label{sect:Definitions and Results}

\subsection{Crystals in Galilean spacetimes and quantum symmetries}

In Part~\ref{part:TPMs} we build up the general framework for topological
insulators from the basic symmetry principles of Galilean relativity
and quantum theory. The microscopic properties of condensed matter systems are described
by nonrelativistic quantum mechanics; in this paper we are interested
in systems of noninteracting electrons in an atomic crystal. The
description of symmetries in this setting is the topic of
Section~\ref{sect:symmetries of Galiliean spacetimes and QM}.

The background geometry is a \emph{Galilean spacetime}.
This can be thought of as an affine space $M$ over a real vector space
$W$, together with a choice of space translations, i.e. a
codimension one subspace $V \subset W$, and measures of Euclidean space
distances and angles as well as time intervals; these are given by choosing inner products $g_V$ and $g_{W/V}$ on $V$ and $W/V$, respectively.
The orbits of $V$ are the `time slices', i.e. all events which occur
simultaneously in $M$.

If we would also like to specify reference frames at rest in $M$
relative to an observer, we need to single out the time translations in $W$.
This is achieved by choosing a splitting $W \cong U \oplus V$ of the short exact sequence $V \to W \to W/V$.
Such a choice of splitting is called a \emph{time direction} in $M$.
A time direction induces a notion of rest frames, but it does not
determine in which way time flows on $M$.
A \emph{time orientation} is additional data given by the choice of an orientation on~$U$.

We denote the data of a Galilean spacetime with a time direction by $(M,W,V,g_V, g_{W/V}, U) = (M, \Gamma, U)$.
The \emph{Galilean transformations} $\Aut(M,\Gamma,U)$ of a Galilean spacetime with a time direction are those affine transformations of $M$ which preserve Euclidean space distances, time intervals and the splitting. There is a short exact sequence
\begin{equation}
	\xymatrix{
		1 \ar@{->}[r] & W \ar@{->}[r] & \Aut(M,\Gamma,U)
                \ar@{->}[r]^-{\pi} & \Aut(W,\Gamma,U) \ar@{->}[r] & 1\
                ,
	}
\end{equation}
where $\Aut(W,\Gamma,U)$ is the group of automorphisms of $W$ of the form
\begin{equation}
	g =
	\begin{pmatrix}
		\pm\, 1 & 0 \\
		 0 & A
	\end{pmatrix}
\end{equation}
for $A$ an orthogonal transformation of $(V, g_V)$.
The group $\Aut(W,\Gamma,U)$ carries two $\ZZ$-gradings, $t,p \colon \Aut(W,\Gamma,U) \to \ZZ$, where $t(g)$ is the sign of the first block of $g$ while $p(g) = \det(A)$.
In other words, $t$ measures whether $g$ is time-reversing and $p$ measures whether $g$ reverses the space orientation.
Any symmetry group which acts on $M$ preserving the Galilean structure will inherit these two $\ZZ$-gradings from $\Aut(W,\Gamma,U)$.

If $(M,\Gamma,U)$ is a Galilean spacetime with a time direction, a
\emph{crystal} in $(M, \Gamma, U)$ consists of a subset $C \subset M$,
the worldlines of the sites of the crystal, such that the subset $\Pi \subset V$ of those spatial translations that preserve the subset $C$ forms a lattice in $V$ of full rank.
For a crystal $(C, \Pi)$, we introduce the \emph{group of crystal symmetries} $\sfG(C)$ consisting of all Galilean transformations in $\Aut(M,\Gamma,U)$ that preserve the subset $C \subset M$.
It naturally gives rise to the \emph{crystallographic group} $\sfG(C)/U$ and the \emph{magnetic point group} $\widehat{\sfP}(C)=\sfG(C)/(U\oplus\Pi)$ of $(C,\Pi)$.

Let us now proceed to treat the quantum mechanics of systems defined
classically on Galilean spacetimes. 
One of the main results of~\cite{Freed-Moore--TEM} is the realisation of how much of the framework and properties of topological insulators follows entirely from the careful and rigorous combination of the classical notion of symmetry on Galilean spacetimes with the notion of symmetry on quantum systems.
In quantum theory states are the elements of a projective Hilbert space $\PP\scH$, and it is this projective space that symmetries act on.
We write $\Aut_\qm(\PP\scH)$ for the automorphisms of $\PP\scH$ which
leave transition probabilities invariant, and $\Aut_\qm(\scH)$ for the
group of unitary and antiunitary operators on $\scH$.
By a famous theorem of Wigner~\cite{Wigner--Group_theory}, 
there is a short exact sequence
\begin{equation}
\label{eq:Wigner intro}
	\xymatrixcolsep{1cm}
	\xymatrix{
		1 \ar@{->}[r] & \sfU(1) \ar@{->}[r]^-{(-) \cdot \One}
                & \Aut_{\qm}(\scH) \ar@{->}[r]^-{q} &
                \Aut_{\qm}(\PP\scH) \ar@{->}[r] & 1\ ,
	}
\end{equation}
where $\sfU(1)$ is included into $\Aut_{\qm}(\scH)$ as scalar multiples of the identity.
The group $\Aut_{\qm}(\PP\scH)$ is naturally $\ZZ$-graded by setting $\phi_\scH([T]) = 1$ if $[T] \in \Aut_{\qm}(\PP\scH)$ is the class of a unitary operator $T$ and $\phi_\scH([T]) = -1$ if $T$ is antiunitary.

If a quantum system comes endowed with a specified Hamiltonian $H$ on $\scH$, we would like to restrict to the subgroup of symmetries $\Aut_{\qm}(\scH,H) \subset \Aut_{\qm}(\scH)$ which projectively commute with the Hamiltonian.
That is, we demand that there is a map $\sfc \colon \Aut_{\qm}(\scH,H) \times \scH \to \sfU(1)$ such that $T\, H \psi = \sfc(T, \psi)\, H\, T \psi$ for all $\psi \in \scH$ and $T \in \Aut_{\qm}(\scH,H)$.
By refining the treatment in~\cite{Freed-Moore--TEM}, we show
in Section~\ref{sect:QM_symmetries} directly from this setup that the only consistent choices of such maps are given by continuous $\ZZ$-gradings $\sfc \colon \Aut_{\qm}(\scH,H) \to \ZZ$.
Therefore $\Aut_{\qm}(\scH,H)$ and $\Aut_{\qm}(\PP\scH,H) =
\Aut_{\qm}(\scH,H)/\sfU(1)$ become $\ZZ$-bigraded topological groups,
with gradings given by $\phi_\scH$ and $\sfc$, and
$\Aut_{\qm}(\scH,H)$ consists of those unitary and antiunitary
operators on $\scH$ that either commute or anticommute with the
Hamiltonian $H$. 
Symmetries $T$ with $\sfc(T)=-1$ are often called \emph{chiral symmetries}.
As a consequence, we have 

\begin{proposition}
Let $(\scH,H)$ be a quantum system.
\begin{myenumerate}
\item
If $(\scH,H)$ admits a chiral symmetry $T \in \Aut_{\qm}(\scH,H)$ then
the spectrum $\sigma(H)\subset\FR$ of $H$ is symmetric, i.e. $\sigma(H) = - \sigma(H)$.
\item
A symmetry $T$ is \emph{time-reversing} if and only if $\sft(T) = \sfc(T)\, \phi_\scH(T) = -1$.
\item
If $\sigma(H)$ is not symmetric (for example if $H$ is bounded from
below but not from above), then there are no chiral symmetries $T$ in $\Aut_{\qm}(\scH,H)$.
However, there can still be antiunitary time-reversing symmetries.
\end{myenumerate}
\end{proposition}

Any group of symmetries $\sfG$ acts on the quantum system $(\scH,H)$ via a group homomorphism $\PP(\rho) \colon \sfG \to \Aut_\qm(\PP\scH, H)$.
By pulling back the short exact sequence~\eqref{eq:Wigner intro} along
$\PP(\rho)$, the group $\sfG$ inherits a $\sfU(1)$-extension $\sfG_\qm
\to \sfG$ as well as a $\ZZ$-bigrading and a continuous homomorphism $\rho_\qm \colon \sfG_\qm \to \Aut_\qm(\scH,H)$ compatible with both gradings.

\begin{definition}
A \emph{quantum extension} of a $\ZZ$-bigraded topological group $(\sfG, \phi_\sfG, \sfc_\sfG) $ is a triple $(\sfG_\qm, \iota_\sfG, q_\sfG)$ giving rise to a short exact sequence of topological groups
\begin{equation}
	\xymatrix{
		1 \ar@{->}[r] & \sfU(1) \ar@{->}[r]^-{\iota_\sfG} &
                \sfG_\qm \ar@{->}[r]^-{q_\sfG} & \sfG \ar@{->}[r] & 1
	}
\end{equation}
satisfying the relation $g\, \iota_\sfG(\lambda) =
\iota_\sfG(\lambda)^{\phi_\sfG(g)}\, g$ for all $g \in \sfG_\qm$ and $\lambda\in\sfU(1)$.
The extension $\sfG_\qm$ naturally becomes a $\ZZ$-bigraded
topological group with the pullbacks along $q_\sfG$ of the gradings on~$\sfG$.
\end{definition}

This allows us to give a precise definition of a gapped quantum system with a fixed action of symmetries.

\begin{definition}
Let $(\sfG, \phi_\sfG, \sfc_\sfG)$ be a $\ZZ$-bigraded topological group and let $(\sfG_\qm, \iota_\sfG, q_\sfG)$ be a quantum extension of $\sfG$.
A \emph{gapped quantum system with symmetry class $(\sfG_\qm,
  \iota_\sfG, q_\sfG)$} is a gapped quantum system $(\scH,H)$ together
with a homomorphism $\rho_\qm : \sfG_\qm \to\Aut_\qm(\scH,H)$ of
$\ZZ$-bigraded topological groups such that
$\rho_\qm\circ\iota_\sfG(\lambda)=\lambda\,\One$ for all $\lambda\in\sfU(1)$.
\end{definition}

The gap in the energy spectrum $\sigma(H)$ between the ground state and the first excited state represents the \emph{Fermi energy}, which can always be shifted to zero. 

Bringing these two notions of symmetries together is the topic of Section~\ref{sect:Gapped topological phases}.
Let $C \subset M$ be a crystal in a Galilean spacetime $(M,\Gamma, U)$ with a time direction.
We consider symmetry groups $\sfH$ acting on $M$ in a way that preserves $C$, i.e. via certain group homomorphisms $\sfH \to \sfG(C)$.
Among other technical conditions we assume that $U \oplus \Pi$ is contained as a normal subgroup in the image of $\sfH$.
The quantum Hilbert space of wavefunctions is associated to the space manifold $M_\rms \coloneqq M/U$ underlying the Galilean spacetime.
More precisely $\scH = \rmL^2(M_\rms, V_\intern)$ where $V_\intern$ is a vector space of internal degrees of freedom such as spin.
Thus we divide by time translations to obtain the space symmetry group $\sfG$, which sits in a commutative diagram
\begin{equation}
	\xymatrix{
		1 \ar@{->}[r] & \Pi \ar@{->}[r]^-{j} \ar@{=}[d] & \sfG \ar@{->}[r]^-{\pi} \ar@{->}[d]^-{\gamma} & \sfG'' \ar@{->}[r] \ar@{->}[d]^-{\gamma''} & 1\\
		1 \ar@{->}[r] & \Pi \ar@{->}[r] & \sfG(C)/U \ar@{->}[r] & \widehat{\sfP}(C) \ar@{->}[r] & 1
	}
\end{equation}
Any action of $\sfG$ on $\scH$ as quantum symmetries induces a $\sfU(1)$ extension $\sfG_\qm \to \sfG$ and a group homomorphism $\rho_\qm \colon \sfG_\qm \to \Aut_\qm(\scH,H)$.
The Hamiltonian $H$ is usually of Schr\"{o}dinger type, e.g. $H=-\frac{\hbar^2}{2m}\, \triangle+{\rm W}_\Pi$ for a single electron of mass $m$, where $\triangle$ is the Laplace operator on Euclidean space and ${\rm W}_\Pi$ is a periodic potential energy function invariant under lattice translations by $\Pi\subset V$. If $H$ is gapped we thus arrive at a gapped quantum system with symmetry class $(\sfG_\qm, \iota_\sfG, q_\sfG)$.

\subsection{Band theory, topological phases and time-reversal symmetry}

The standard Bloch theory of crystals in condensed matter physics exploits the translational symmetry of a crystal to describe electronic
states in terms of their crystal momenta which are defined in a
periodic Brillouin zone.
Its connection to the present formalism comes about as follows.
We assume that the quantum extension $(\sfG_\qm, \iota_\sfG,
q_\sfG)$ is trivial over over the Bravais lattice of the crystal $\Pi \subset \sfG$ (which is the case for
Schr\"odinger wavefunctions on $M_\rms$).
In this case the representation $\rho_\qm$ restricts to a continuous unitary representation of $\Pi \subset \sfG_{\qm}$.
A continuous unitary representation of the discrete abelian group $\Pi$ on a quantum Hilbert space $\scH$ has a spectral decomposition in terms of a projection-valued measure on the Pontryagin dual group $X_{\Pi}$ of $\Pi$.
The space $X_\Pi$ is the
\emph{Brillouin torus} in condensed matter physics and it parameterises the periodic crystal momenta.

The projection-valued measure decomposes $\scH$ into a sheaf of
Hilbert spaces over $X_\Pi$. We assume that this sheaf arises as the space of $\rmL^2$-sections of a smooth bundle of Hilbert spaces $E \to X_\Pi$ in the sense that there is a unitary isomorphism $\rmU \colon \scH \to \rmL^2(X_\Pi, E)$ which intertwines the representation $\rho_\qm$ of $\sfG_\qm$ on $\scH$ and the representation of $\sfG_\qm$ on $\rmL^2(X_\Pi, E)$ induced by the adjoint action of $\sfG_\qm$ on its normal subgroup $\Pi$.
The Hamiltonian $H$ is an endomorphism of this bundle which acts fibrewise on $E$ since it commutes with $\rho_{\qm|\Pi}$. Its restrictions to the fibres $E_{|\lambda}$ over $\lambda\in X_\Pi$ are
assumed to be gapped, and these are the \emph{Bloch Hamiltonians}. 

We write $E^\pm$ for the subbundles of $E$ consisting of the
eigenspaces for positive and negative energy eigenvalues of the Bloch
Hamiltonians, respectively. In an insulator, the energy levels above
the Fermi energy are the empty \emph{conduction bands} (holes),
whereas the energy levels below the Fermi energy are the filled
\emph{valence bands} (particles). Since only a finite number of
occupied levels typically exist, we can assume that $E^-$ has finite
rank $\rank(E^-)$; the vector bundle $E^- \rightarrow X_\Pi$ is called the \emph{Bloch bundle}. The number of empty states can however be infinite. 

\begin{definition}
Let $(M,\Gamma,U)$ be a Galilean spacetime with a time direction.
A \emph{band insulator} in $(M,\Gamma,U)$ consists of the following data:
\begin{myenumerate}
	\item A crystal $(C,\Pi)$ in $(M,\Gamma,U)$.
	
	\item A Galilean symmetry group $\sfH$ preserving the crystal and defining a space symmetry group $\sfG$ with quotient group $\sfG'' \coloneqq \sfG/\Pi$.
	
	\item A quantum extension $(\sfG''_\qm,\iota_{\sfG''}, q_{\sfG''})$ of $(\sfG'', \phi_{\sfG''}, \sfc_{\sfG''})$ which pulls back to a quantum extension $(\sfG_\qm, \iota_\sfG, q_\sfG)$ of $\sfG$ that is trivial over $\Pi \subset \sfG$.
	
	\item A gapped quantum system $(\scH,H)$ with symmetry class $(\sfG_\qm,\iota_\sfG, q_\sfG)$.
\end{myenumerate}
A band insulator is of \emph{type I} if $E^+$ has infinite rank and of \emph{type F} if $E^+$ has finite rank.\\
\end{definition}

Thus far we have mostly described the action of lattice
translations. What becomes of the rest of the symmetry group $\sfG$?
In Section~\ref{sect:Decomp_reps} we observe that there is an action
$R:X_\Pi\times \sfG''\to X_\Pi$ of $\sfG''$ on the Brillouin torus
$X_\Pi$. However, this action does not lift to $E$ because the
representation $\rho_\qm$ is not defined on $\sfG''$. The obstruction
is a Hermitean \emph{twisting line bundle} $\hat{L} \to X_\Pi \times
\sfG^{\prime \prime}$, whose construction we give explicitly and which is multiplicative in a sense made precise in Proposition~\ref{st:Brillouin_gerbe}.
We content ourselves here in saying that there is a bundle isomorphism 
\begin{equation}
	\hat{\mu}_{|(\lambda, g_1'', g_2''\,)} \colon \hat{L}_{|(\lambda, g''_1\,)}^{[\phi_\sfG(g_2''\,)]} \otimes \hat{L}_{|(R_{g_1''} (\lambda), g''_2\,)} \arisom \hat{L}_{|(\lambda, g''_1\, g_2''\,)}
\end{equation}
for all $\lambda\in X_\Pi$ and $g_1'', g_2'' \in \sfG''$; the exponent of the first factor indicates that the complex conjugate or dual line bundle is used instead whenever $\phi_\sfG(g_2''\,) = -1$.
The bundle $E$ on the Brillouin torus is not $\sfG''$-equivariant in
general, but rather there is a bundle isomorphism
\begin{equation}
\label{eq:alpha intro}
	\hat{\alpha}_{|(\lambda,g''\,)} \colon \big( \hat{L}_{|(\lambda,g''\,)} \otimes E_{|\lambda} \big)^{[\phi_\sfG(g''\,)]} \arisom E_{|R_{g''}(\lambda)}
\end{equation}
for all $\lambda\in X_\Pi$ and $g'' \in \sfG''$, which is compatible with $\hat{\mu}$.
The vector bundle $E = E^+ \oplus E^-$ is naturally $\ZZ$-graded.
Hermitean vector bundles of this type are called \emph{$\hat{L}$-twisted $\sfG''$-equivariant $\ZZ$-graded Hermitean vector bundles} on $X_\Pi$.

Recalling from Section~\ref{sect:Introduction} that we are interested in deformation classes, we introduce an equivalence
relation on band insulators with fixed symmetry class $(\sfG_\qm,
\iota_\sfG, q_\sfG)$ by saying that two topological insulators are equivalent whenever they are linked
by a string of unitary transformations of Hilbert spaces respecting the
symmetry, together with continuous deformations of the Hamiltonian $H$
that do not close the energy gap in $\sigma(H)$ and which are
compatible with the symmetry.
Let $\TP(\sfG_\qm, \iota_\sfG, q_\sfG)$ denote the set of equivalence classes.
The direct sum of Hilbert spaces turns this into a commutative monoid, and we let $\RTP(\sfG_\qm, \iota_\sfG, q_\sfG)$ denote its Grothendieck completion.

\begin{definition}
A \emph{topological phase} with symmetry class $(\sfG_\qm, \iota_\sfG,
q_\sfG)$ is an element of the commutative monoid $\TP(\sfG_\qm, \iota_\sfG, q_\sfG)$.
Elements of the abelian group $\RTP(\sfG_\qm, \iota_\sfG, q_\sfG)$ are \emph{reduced topological phases} with symmetry class $(\sfG_\qm, \iota_\sfG, q_\sfG)$.
\end{definition}

In Section~\ref{sect:Brillouin_gerbe_and_KT_classification} we
describe $\RTP(\sfG_\qm, \iota_\sfG, q_\sfG)$ geometrically using the Freed-Moore extension of Kitaev's K-theory classification of topological phases~\cite{Kitaev:Periodic_table,Freed-Moore--TEM}.
Defining the \emph{$\hat{L}$-twisted $\sfG''$-equivariant
  K-theory group $\rmK^{\hat{L}, \sfG''}(X_\Pi)$} of $X_\Pi$ to be the
Grothendieck completion of the commutative monoid of isomorphism
classes of twisted equivariant vector bundles on $X_\Pi$ from above,\footnote{A $C^*$-algebraic perspective of this twisted K-theory is described in~\cite{Thiang2014}.} one obtains

\begin{theorem}
\label{st:KT classification intro}
Let $\RTP_\rmF(\sfG_\qm, \iota_\sfG, q_\sfG)$ denote the Grothendieck
group of reduced topological phases obtained from band insulators of type $\rmF$.
Sending a topological phase of this type to the K-theory class of the underlying Hermitean vector bundles $E$ induces an isomorphism
\begin{equation}
	\RTP_\rmF\big(\sfG_\qm,\iota_\sfG, q_\sfG\big) \arisom
        \rmK^{\hat{L},\sfG''}\big(X_\Pi\big)\ .
\end{equation}
\end{theorem}

In Section~\ref{sect:Bloch_and_Berry_explicit} we specialise to the explicit Hilbert space $\scH = \rmL^2(M_\rms, V_\intern)$ of wavefunctions described earlier.
We work out the details of how to construct the Hilbert bundle $E$ over the Brillouin torus and the $\hat{L}$-twisted $\sfG''$-action on $E$.
This is contained in Theorem~\ref{st:Equivariance_of_E_and_PT}, together with the canonical construction of the \emph{Berry
  connection} $\nabla^{E^-}$ on the Bloch bundle from the projection
$E \to E^-$, which is equivalent to its usual definition as the
canonical Grassmann connection on $E^-$. State vectors
$\psi_{|\lambda}\in E_{|\lambda}$ for $\lambda\in X_\Pi$ are the
quasiperiodic \emph{Bloch wavefunctions} on space $M_\rms$ with $\rho_\qm(\xi)(\psi_{|\lambda}) = \lambda(\xi)\, \psi_{|\lambda}$ for $\xi\in\Pi$.

Sections~\ref{sect:splittings_and_equivariance}
and~\ref{sect:pure_time_reversal_in_TEM} are then devoted to specialising to one of the simplest cases with additional symmetries:
crystalline systems with time-reversal symmetry.
This is in essence already contained in the general treatment
of~\cite{Freed-Moore--TEM}, but we take a different route and
explicitly provide all details.
The key new input here is the assumption that the projection map
$\sfG_\qm \to \sfG''$ admits a global section $\hat{s} \colon \sfG''
\to \sfG_\qm$, which we do not assume to be a group homomorphism.
We show in Lemma~\ref{st:sections 2-cocycles and trivialisation} that such sections induce trivialisations of the twisting line bundle $\hat{L}$, and that the obstruction to $\hat{s}$ being a group homomorphism is a $2$-cocycle on $\sfG''$ with values in $\Pi \times \sfU(1)$.
Plugging such a trivialisation into~\eqref{eq:alpha intro} leaves a bundle isomorphism
\begin{equation}
	\hat{\alpha}_{|(\lambda,g''\,)} \colon E_{|\lambda}^{[\phi_\sfG(g''\,)]} \arisom E_{|R_{g''}(\lambda)}
\end{equation}
for all $\lambda\in X_\Pi$ and $g'' \in \sfG''$. Iterating
these maps, the $2$-cocycle on $\sfG''$ enters through the trivialisation of $\hat{L}$ so that $E$ still fails to be $\sfG''$-equivariant: the deviation from an equivariant structure is now controlled by a similar $2$-cocycle on $\sfG''$.
The details can be found in Proposition~\ref{st:Sections_of_Gqm_and_Bloch_bdl}.

In Section~\ref{sect:pure_time_reversal_in_TEM} we fix $\sfG'' = \ZZ$ and $M = \FR^{d+1}$, i.e. we work on the affine spacetime in $d$ space dimensions.
Therefore the Bravais lattice is $\Pi = \RZ^d$, the Brillouin torus is
$X_\Pi \cong \sfU(1)^d \cong
\TT^d$, and $-1 \in \ZZ$ acts trivially on $\FR^d$ and on the
$d$-torus $\TT^d$ by inversion $\tau:=R_{-1}:\lambda \mapsto
\lambda^{-1}$. Many of our considerations apply equal well when the
Brillouin zone in momentum space is the $d$-sphere $S^d$, regarded as the one-point compactification of the
Pontryagin dual of the continuous limit of $\RZ^d$, i.e. $\FR^d$. These continuum models are effective long wavelength theories for the lattice systems above and their quantum systems involve low energy effective Hamiltonians
for the corresponding topological phases; they
will serve as further examples illustrating aspects of our general
formalism.

The full symmetry group is $\sfG=\RZ^d\times\ZZ$, and assuming $\phi_\sfG(-1)=-1$ and $\sfc_\sfG(-1)=1$, the lifts of $-1$ to the quantum extensions $\sfG_\qm = \big( \RZ^d \times \sfU(1) \big)\times (\ZZ)_\qm$ are then time-reversing, where $(\ZZ)_\qm$ is a quantum extension of $\ZZ$;
the latter extensions are classified by whether any lift of $-1 \in
\ZZ$ squares to plus or minus the identity by Lemma~\ref{st:Quantum_extensions_of_ZZ}.
Consequently there is only one nontrivial quantum extension of $\ZZ$.
In this case, any choice of a lift of $-1 \in \ZZ$ to $(\ZZ)_\qm$ determines a section $\hat{s} \colon \sfG'' \to \sfG_\qm$ and we work out the resulting $2$-cocycle.
Defining a \emph{Real} (respectively \emph{Quaternionic}) vector bundle on a space with a
$\ZZ$-action $\tau$ to be a Hermitean vector bundle with an antilinear
lift $\hat\tau$
of the involution to the total space that squares to plus (respectively minus) the identity (see Definition~\ref{def:Quaternionic_vector_bundle}), we derive

\begin{proposition}
\label{st:Bloch bdls for time-reversal intro}
Let $\sfG_\qm = \big(\RZ^d \times  \sfU(1)\big) \times (\ZZ)_\qm$.
\begin{myenumerate}
	\item
	If $(\ZZ)_\qm$ is the trivial quantum extension of $\ZZ$, then
        $(E^-,\nabla^{E^-})$ has the structure of a Real vector bundle
        with connection over $X_\Pi$.
	
	\item
	If $(\ZZ)_\qm$ is the nontrivial quantum extension of $\ZZ$,
        then $E^-$ has the structure of a Quaternionic vector bundle
        over $X_\Pi$.
	The Berry connection $\nabla^{E^-}$ is generally not
        compatible with this structure, but there always exist
        connections $\nabla$ on $E^-$ such that $(E^-,\nabla)$ is a
        Quaternionic vector bundle with connection on $X_\Pi$.
\end{myenumerate}
\end{proposition}

As we are primarily interested in electronic systems in three dimensions, from here on we fix the symmetry group $\sfG = \RZ^3 \times  \ZZ$ as
well as the nontrivial quantum extension of $\ZZ$; this corresponds to
topological insulators of type AII in the Altland-Zirnbauer
classification~\cite{AltlandZirnbauer}.

\subsection{The Kane-Mele invariant and obstruction theory}

Bloch bundles of class AII topological insulators have several
powerful special properties with deep physical implications, whose mathematical formulations are
the subject of Part~\ref{part:KM and localisation} together with their relation to topological phases.
We write $X^\tau$ for the set of fixed points of an involutive space
$(X, \tau)$, i.e. a space $X$ with a $\ZZ$-action $\tau:X\to X$.
If $X^\tau \neq \varnothing$, the fibres of a Quaternionic vector
bundle on $X^\tau$ are quaternionic vector spaces and hence of even complex dimension.
In the application to topological insulators this means that, over the
time-reversal invariant periodic momenta, we can always find a basis of
\emph{Kramers pairs} consisting of band states and their images
under the lift $\hat{\tau}$ of the involution to the Bloch bundle over
the Brillouin torus.
In~\cite{Freed-Moore--TEM,Gawedzki:BGrbs_for_TIs} it is shown that any
Quaternionic vector bundle on $\TT^3$ has trivial first Chern class,
and we may invoke a result of
Panati~\cite{Panati:Triviality_of_Bloch_and_Dirac} which gives

\begin{proposition}
Let $X$ be a connected manifold of dimension $d \leq 3$, and let $m \in 2\,\NN$ be an even positive integer.
Then the first Chern class induces an isomorphism from the pointed set
of isomorphism classes of Hermitean vector bundles of rank $m$ on $X$
to the second integer cohomology group $\rmH^2(X,\RZ)$ regarded as a pointed set.
\end{proposition}

The key to this statement is showing that there exists a
$4$-equivalence of classifying spaces $\sfB\sfU(m) \to \sfB\sfU(1) \simeq K(\RZ,2)$, where $ K(\RZ,2)$ is an Eilenberg-MacLane space.
In other words, there exists a canonical isomorphism $\rmH^2(X,\RZ)
\cong [X, \sfB\sfU(1)]$ between the second integer cohomology group of $X$ and homotopy classes of maps from $X$ to $\sfB\sfU(1)$.
A $p$-equivalence from a space $Y$ to a space $Z$ is a continuous map
$f \colon Y \to Z$ which induces isomorphisms on the $k$-th homotopy
groups for $k = 0,1, \ldots, p$ and a surjection for $k=p$.
The statement then follows from the property of $p$-equivalences~\cite[Theorem 6.31]{Switzer:AT} that for every CW-complex $X$ of dimension $d < p$, composition with $f$ induces an isomorphism on homotopy classes of maps
\begin{equation}
	f_* \colon [X,Y] \arisom [X,Z]\ .
\end{equation}
We provide the full proof in Appendix~\ref{app:proof_of_1st_Chern_knows_everything}, where we show that the map establishing the $4$-equivalence can be made more explicit than in~\cite{Panati:Triviality_of_Bloch_and_Dirac}. Writing $\QVBdl(\TT^3,\tau)$ for the category of Quaternionic vector bundles on $(\TT^3,\tau)$, we then have

\begin{proposition}
Any Quaternionic vector bundle $(E^-, \hat{\tau})\in \QVBdl(\TT^3,\tau)$ has an underlying Hermitean vector bundle $E^-$ which is trivialisable.
\end{proposition}

Let $(X,\tau)$ be a $3$-dimensional compact involutive manifold and let $(E^-, \hat{\tau}) \in \QVBdl_U(X,\tau)$, where $\QVBdl_U(X,\tau)$ is the category of Quaternionic vector bundles on $(X,\tau)$ whose underlying vector bundle is trivialisable.
Assume that $\tau$ has fixed points so that $\rank(E^-) = m$ is even.
Any smooth trivialisation $\varphi$ of $E^- \to X$ induces a smooth map
$w_\varphi \colon X \to \sfU(m)$, sometimes called a \emph{sewing
  matrix} (see e.g.~\cite{Gawedzki--2DFKM_and_WZW}), given by the
commutative diagram
\begin{equation}
\label{eq:w_varphi intro}
	\xymatrixrowsep{1.2cm}
	\xymatrixcolsep{1.2cm}
	\xymatrix{
		E^- \ar@{->}[r]^-{\hat{\tau}^{-1}} \ar@{->}[d]_-{\varphi}	&	\overline{\tau^* E^-} \ar@{->}[d]^-{\overline{\tau^*\varphi}}\\
		X \times  \FC^m \ar@{-->}[r]_-{w_\varphi \circ \theta}	&	\overline{X \times  \FC^m}
	}
\end{equation}
where $\theta$ denotes complex conjugation.
The map $w_\varphi$ is $\ZZ$-equivariant with respect to the
involution $\tau$ on $X$ and the involution $g
\mapsto -g^\trans$ on the unitary group $\sfU(m)$ sending a matrix $g
\in \sfU(m)$ to minus its transpose. The elements of the sewing matrix
can be written in an orthonormal basis of global sections $\{\psi^a\}_{a=1,\dots,m}$ as 
\begin{equation}
w_\varphi^{ab}(\lambda)= h_{E^-}\big(\psi_{|\tau(\lambda)}^a ,\hat \tau(\psi_{|\lambda}^b)\big)
\end{equation}
where $h_{E^-}$ denotes the Hermitean metric on $E^-$. This matrix is
skewadjoint at the fixed points of $\tau$. On the other hand, a
trivial topological phase is characterised by the existence of a
global frame of smooth Bloch wavefunctions such that
$\psi^a_{|\tau(\lambda)}=\sum_{b=1}^m\, w_0^{ab}\,
\hat\tau(\psi^b_{|\lambda})$ for the unitary skewadjoint matrix
$w_0\in\sfU(m)$ given by
\begin{equation}
w_0=\begin{pmatrix}
0 & \One \\
-\One & 0
\end{pmatrix} \ ,
\end{equation}
 which decomposes
the sewing matrix globally into the block offdiagonal form $w_\varphi=w_0$.

Let $H_{\sfU(m)} \in \Omega^3(\sfU(m))$ denote the canonical $3$-form on $\sfU(m)$.

\begin{definition}
Let $(X, \tau)$ be a closed and oriented involutive manifold of dimension $d = 3$ with orientation-reversing involution $\tau$ such that $X^\tau \neq \varnothing$.
If $(E^-, \hat{\tau}) \in \QVBdl_U(X, \tau)$ and $\varphi$ is any trivialisation of $E^-$, the \emph{Kane-Mele invariant} of $(E^-, \hat{\tau})$ is
\begin{equation}
	\rmK\rmM(E^-, \hat{\tau}) \coloneqq \exp \bigg( \pi\, \iu\,
        \int_{X}\, w_\varphi^*H_{\sfU(m)} \bigg)\ \in\ \ZZ\ .
\end{equation}
Given a time-reversal symmetric band insulator on a $3$-dimensional
Bravais lattice with Bloch bundle $(E^-, \hat{\tau}) \in \QVBdl(\TT^3, \tau)$,
its \emph{Kane-Mele invariant} is the mod~$2$ index $\rmK\rmM(E^-, \hat{\tau}) \in \ZZ$.
\end{definition}

This perspective on the Kane-Mele invariant, which defines it via the Wess-Zumino-Witten action functional of the map $w_\varphi:X\to\sfU(m)$, was introduced in~\cite{WQZ,Gawedzki--2DFKM_and_WZW}; the language of Quaternionic vector bundles is particularly emphasised in~\cite{DeNittis-Gomi:Classification_of_quaternionic_Bloch_bdls,DeNittis-Gomi:FKMM_in_low_dimension}.
It follows from Theorem~\ref{st:KT classification intro} and Proposition~\ref{st:Bloch bdls for time-reversal intro} that reduced topological phases with time-reversal symmetry as in part~(2) of Proposition~\ref{st:Bloch bdls for time-reversal intro} are in bijection with the Quaternionic K-theory group $\rmK\rmQ(\TT^3, \tau)$ of $(\TT^3, \tau)$.
The abelian group $\rmK\rmQ(\TT^3, \tau)\cong\RZ\oplus\ZZ^4$ is the Grothendieck completion of the commutative monoid of isomorphism classes of Quaternionic vector bundles on $(\TT^3, \tau)$ with the direct sum.
The Kane-Mele invariant contains information about this group:

\begin{proposition}
The Kane-Mele invariant induces a homomorphism of abelian groups
\begin{equation}
	\rmK\rmM \colon \rmK\rmQ(\TT^3, \tau) \longrightarrow \ZZ\ .
\end{equation}
\end{proposition}

Recalling Theorem~\ref{st:KT classification intro}, we thus infer that
the Kane-Mele invariant knows part of whether a given band insulator is in a nontrivial reduced topological phase.
It is nonetheless still interesting to understand what it knows and how much.
That is, which part of the information about the topological phase of a band insulator does it detect?

To address these questions, we collect and investigate several different perspectives on the Kane-Mele invariant in Part~\ref{part:KM and localisation} of this paper.
In Section~\ref{sect:KM as CS}, we recall and 
generalise the perspective
of~\cite{Freed-Moore--TEM,Gawedzki--2DFKM_and_WZW} on the Kane-Mele invariant via
the Chern-Simons action functional through

\begin{proposition}
If $(E^-, \hat{\tau}) \in \QVBdl_U(X, \tau)$, the Kane-Mele invariant $\rmK\rmM(E^-, \hat{\tau})$ coincides with the Chern-Simons invariant of any Quaternionic connection on $E^-$.
\end{proposition}

When $X=\TT^3$, this result realises the Kane-Mele invariant as a magnetoelectric polarisation~\cite{WQZ}, which is experimentally measurable.

Section~\ref{sect:KM as obstruction to factorisation} then contains a novel approach to the Kane-Mele invariant via homotopy theory.
Two main observations enter this description.
Firstly, we show that there exists a $4$-equivalence $\sfS\sfU(m) \to
K(\RZ,3)$ in Proposition~\ref{st:SUm to KZ3 4-equiv} which implies

\begin{proposition}
\label{prop:H3XSUm}
If $X$ is a CW-complex of dimension $d \leq 3$ and $m \geq 2$,
there are isomorphisms
\begin{equation}
	\rmH^3(X,\RZ) \cong [X, K(\RZ,3)] \cong [X, \sfS\sfU(m)]\ .
\end{equation}
\end{proposition}

Secondly, a recuring theme throughout Part~\ref{part:KM and localisation} is the interplay of integration with orientation-reversal and pullbacks.
For example, if $(X, \tau)$ is a $d$-dimensional connected closed and orientable involutive
manifold with orientation-reversing involution $\tau$, and $\omega \in
\Omega^d(X)$ is a top-form with integer periods on $X$, then
$\int_X \, \omega \in 2\, \RZ$ if and only if there exists $\eta \in
\Omega^d(X)$ such that $\int_X\, \eta \in \RZ$ and $[\omega] = [\eta] - \tau^*[\eta]$.
Combining this simple observation with the $4$-equivalence
$\sfS\sfU(m) \to K(\RZ,3)$ and Proposition~\ref{prop:H3XSUm} leads to

\begin{theorem}
\label{st:splitting ho intro}
Let $(X, \tau)$ be an orientable connected and closed $3$-dimensional
CW-complex with orientation-reversing involution $\tau$, and
$X^\tau \neq \varnothing$. Let $(E^-, \hat{\tau}) \in \QVBdl_U(X, \tau)$ with $\rank(E^-) = m = 2n$.
The following statements are equivalent:
\begin{myenumerate}
	\item The Kane-Mele invariant of $(E^-, \hat{\tau})$ is trivial: $\rmK\rmM(E^-, \hat{\tau}) = 1$.
	
	\item There exists a trivialisation $\varphi$ of $E^-$ such
          that the induced Quaternionic structure on $X \times \FC^m$
          is homotopic to a split Quaternionic structure on $X \times
          (\FC^n \oplus \FC^n)$, i.e. there is a homotopy of $\sfU(m)$-valued maps
	\begin{equation}
		w_\varphi \simeq
		\begin{pmatrix}
			0 & g
			\\
			- (g \circ \tau)^\trans & 0
		\end{pmatrix}
	\end{equation}
	for some map $g \colon X \to \sfS\sfU(n)$.
	
	\item There exists a trivialisation $\varphi'$ of $E^-$ such
          that the induced Quaternionic structure on $X \times \FC^m$
          is homotopic to the trivial Quaternionic structure on $X
          \times (\FC^n \oplus \FC^n)$, i.e. there is a homotopy of $\sfU(m)$-valued maps
	\begin{equation}
	\label{eq:Quaternionic triv up to ho intro}
		w_{\varphi'} \simeq
		\begin{pmatrix}
			0 & \One
			\\
			- \One & 0
		\end{pmatrix}\ .
	\end{equation}
\end{myenumerate}
\end{theorem}

It may seem surprising at first that the homotopies here do not respect the time-reversal symmetry: they are not necessarily $\ZZ$-equivariant.
We interpret this fact as a concrete answer to the questions posed above (see Remark~\ref{rmk:homotopy interpretation of KM}):
the $3$-dimensional Kane-Mele invariant is by no means complete, neither at the level of Quaternionic K-theory nor at the level of a classification of Quaternionic vector bundles.
All it detects is whether there exists a global frame of $E^-$ which is time-reversal symmetric up to (nonequivariant) homotopy in the sense of~\eqref{eq:Quaternionic triv up to ho intro}.
However, if a topological phase has nontrivial Kane-Mele invariant then it is certainly a nontrivial phase.

Another description of the Kane-Mele invariant is presented in
Section~\ref{sect:holonomies_and_localisation} in terms of higher
geometry, specifically Jandl gerbes.
A Jandl gerbe consists of a bundle gerbe $\CG$ with
connection on an involutive manifold $(X,\tau)$, together with a
$1$-isomorphism $(A,\alpha)$ and a $2$-isomorphism $\psi$ fitting into
the diagram
\small
\begin{equation}
\begin{tikzcd}[row sep=1cm, column sep=1.5cm]
	|[alias=A]| \tau^*(\tau^*\CG) \ar[r, "{\tau^*(A,\alpha)}"] \ar[d, equal] & \tau^*\CG^*
	\\
	\CG \ar[ur, "{(A,\alpha)^*}"', ""{name=U,inner sep=1pt,above}] & 
	\arrow[Rightarrow, from=A, to=U, "\psi"]
\end{tikzcd}
\end{equation}
\normalsize
such that $\psi$ squares to the identity in a certain sense (see Definition~\ref{def:Real BGrb}).

Combining the Dixmier-Douady classification of bundle gerbes on $X$ in
terms of the third integer cohomology group $\rmH^3(X,\RZ)$ from
Theorem~\ref{st:Classification of BGrbs} with the $4$-equivalence
$\sfS\sfU(m) \to K(\RZ,3)$ and
Proposition~\ref{prop:H3XSUm} implies that maps $X \to \sfS\sfU(m)$ classify bundle gerbes on $X$ in dimensions $d \leq 3$.
This insight allows us to provide geometric versions of the homotopy theoretic statements above:
we denote the Dixmier-Douady class of a bundle gerbe $\CG$ on $X$ by
$\rmD(\CG) \in \rmH^3(X,\RZ)$ and the basic bundle gerbe on
$\sfS\sfU(m)$ by $\CG_{\rm basic}$~\cite{Gawedzki:BGrbs_for_TIs,Murray-Stevenson:Basi_BGrb_on_unitary_groups,Waldorf--Thesis}, which is a canonical bundle gerbe with connection on $\sfS\sfU(m)$ whose Dixmier-Douady class generates $\rmH^3(\sfS\sfU(m), \RZ) \cong \RZ$.
Denoting the pairing between cohomology and homology classes by
$\<-,-\>$, we then have

\begin{theorem}
Let $(X, \tau)$ be a closed connected and oriented involutive $3$-manifold with
orien{-}tation-reversing involution $\tau$, let $\CG $ be a bundle gerbe on $X$, and let $(E^-, \hat{\tau}) \in \QVBdl_U(X, \tau)$.
\begin{myenumerate}
	\item There exists a bundle gerbe $\CH $ on $X$ and an isomorphism $\CG \cong \tau^{*}\CH \otimes \CH^*$ if and only if $\exp \big( \pi\, \iu\, \< \rmD(\CG), [X^3]\>\big) = 1$.
	
	\item $\rmK\rmM(E^-, \hat{\tau}) = 1$ if and only if for any
          trivialisation $\varphi$ of $E^-$, after deforming the sewing
          matrix $w_\varphi$ to an $\sfS\sfU(m)$-valued map, there exists a
          bundle gerbe $\CH $ on $X$ with $w_\varphi^*\CG_{\rm basic} \cong \tau^{*}\CH \otimes \CH^*$.
\end{myenumerate}
\end{theorem}

Since the basic bundle gerbe $\CG_{\rm basic}$ is
\emph{multiplicative}~\cite{Waldorf:Multiplicative_BGrbs}, i.e. it
carries a $2$-categorical lift of the group operations on
$\sfS\sfU(m)$, the second part of this theorem indeed rephrases
statement~(2) of Theorem~\ref{st:splitting ho intro} in a geometric
framework and provides yet another interpretation of the Kane-Mele invariant in three dimensions:
the failure of the homotopy in Theorem~\ref{st:splitting ho intro} to
respect the time-reversal symmetry is translated to the fact that the
isomorphism $w_\varphi^*\CG_{\rm basic} \cong \tau^{*}\CH \otimes
\CH^*$ is unrelated to the $\ZZ$-action, i.e. there is no
relation to Jandl structures.
This is the manifestation in the geometric formalism of the
incompleteness of the Kane-Mele invariant, but evidently the invariant nevertheless contains deep information about the underlying geometry. In particular, this identifies the
Kane-Mele invariant as a $\ZZ$-valued version of the Dixmier-Douady
invariant of the bundle gerbe $w_\varphi^*\CG_{\rm basic}$ on $X$,
suggesting that there may be an underlying purely cohomological definition of the invariant.

\subsection{Localisation formulas}

In Sections~\ref{sect:Localisation in equivar HdR}
and~\ref{sect:holonomies_and_localisation} we consider the
localisation of the Kane-Mele invariant with respect to the
$\ZZ$-action on $X$.
That the Kane-Mele invariant may be computed entirely from data
associated to the fixed point set $X^\tau$, or the time-reversal
invariant momenta, is not new and several discrete Pfaffian formulas are known already involving the sewing matrix, see for instance~\cite{Gawedzki:BGrbs_for_TIs,Freed-Moore--TEM,Kane-Mele:Z2_top_order_and_QSHE,FKM,WQZ}.
In this setting the Kane-Mele invariant counts the parity of Majorana
zero modes in the form of Dirac cones, which are conical singularities
produced by surface states at the fixed points over the $\tau$-invariant
codimension one submanifolds of $X$; geometrically,
these conical singularities correspond to nonzero transition
amplitudes between sections of
the Bloch bundle $E^-$ occuring in Kramers pairs. If a material has an odd
number of Dirac cones, then it is in a nontrivial topological
insulating phase; in this case all but one of the Majorana zero modes can pair together to form a composite boson. Techniques used to derive these discrete formulas are drawn from algebraic topology or explicit but very cumbersome computations with matrices. 

In Section~\ref{sect:Localisation in equivar HdR} we provide a new localisation principle for the $3$-dimensional Kane-Mele invariant which is rooted entirely in differential geometry and homological algebra.
The key new ingredient is a Mayer-Vietoris Theorem for $\ZZ$-equivariant and Real differential forms on involutive manifolds.
By $\ZZ$-equivariant or Real differential forms we mean forms $\omega \in \Omega^p(X)$ with $\tau^* \omega = \pm\, \omega$, respectively.
We introduce the notation $\omega \in \Omega^p_\pm(M)$ and denote the respective quotients by $\dd \Omega^{p-1}_\pm(X)$ as $\rmH^p_{\dR\, \pm}(X)$.
For general Lie group actions, this version of equivariant de~Rham
cohomology is already contained in the seminal Chevalley-Eilenberg
treatise on the cohomology of Lie algebras~\cite{CE--Cohomology_of_Lie_groups_and_Lie_algebras}.
We provide the definitions of the groups $\rmH^\bullet_{\dR\, \pm}(X)$ in detail and work out some of their properties in Appendix~\ref{app:HdR and group actions}.
The Mayer-Vietoris sequence for involutive spaces then allows us to develop a localisation technique for de~Rham cohomology in the presence of $\ZZ$-actions.
The usual localisation techniques (see
e.g.~\cite{Guillemin-Sternberg:SuSy_and_equivar_HdR,Szabo:2000fs}) for
equivariant cohomology do not apply to $\ZZ$-actions, as these
treatments strongly depend on the group actions generating fundamental
vector fields; indeed, it is known (see e.g.~\cite{Hekmati:2016ikh})
that the Borel equivariant cohomology (with local coefficients) in this case is rationally isomorphic to the (anti)invariant cohomology classes in the $\ZZ$-module $\rmH^\bullet_{\dR\, \pm}(X)$.

The Mayer-Vietoris long exact sequence in ordinary de~Rham cohomology is induced by a decomposition of a manifold into two open subsets $X = V \cup V'$.
If there is an involution $\tau$ on $X$, then specifying the subset $V$
automatically yields a second open subset $\tau(V)
\subset X$ as required for the Mayer-Vietoris Theorem.
If we are considering either $\ZZ$-equivariant or Real forms $\omega$ on $X$, then $\omega_{|\tau(V)}$ is determined by $\omega_{|V}$.
This would seem to imply that whenever we are interested in forms with
such symmetry properties, we can entirely drop one of the two subsets
that usually occur in the Mayer-Vietoris sequence. This intuition is
indeed correct and we have

\begin{theorem}
If $(X^d,\tau^d)$ is a $d$-dimensional involutive manifold, and $V^d
\subset X^d$ is an open subset such that $X^d = V^d \cup \tau^d(V^d)$ and $U^d = V^d \cap \tau^d(V^d)$,
there is a long exact sequence of cohomology groups
\small
\begin{equation}
\begin{tikzcd}[column sep=2cm, row sep=1cm]
	\cdots \ar[r] & \rmH^{p-1}_{\dR\, \mp}(X^d) \ar[r, "\iota_d^*"] & \rmH_\dR^{p-1}(V^d)\ar[r, "(1 \pm \tau^{d\,*})\,\jmath_d^*"] \ar[d, phantom, ""{coordinate, name=Z}] & \rmH^{p-1}_{\dR\, \pm}(U^d)
	\ar[dll, "\Delta^{p-1}"' pos=1, rounded corners, to path={-- ([xshift=2ex]\tikztostart.east) |- (Z) \tikztonodes -| ([xshift=-2ex]\tikztotarget.west) -- (\tikztotarget)}]
	\\
	 & \rmH^p_{\dR\, \mp}(X^d) \ar[r, "\iota_d^*"] & \rmH_\dR^p(V^d)\ar[r, "(1 \pm \tau^{d\,*})\,\jmath_d^*"] & \rmH^p_{\dR\, \pm}(U^d) \ar[r] & \cdots
\end{tikzcd}
\end{equation}
\normalsize
where $\imath_d \colon V^d \hookrightarrow X^d$ and $\jmath_d \colon U^d \hookrightarrow V^d$ are inclusions of open subsets.
\end{theorem}

This long exact sequence interpolates between Real and $\ZZ$-equivariant cohomology.
We proceed by showing that the theorem still applies if we choose
$V^d$ to be a thickening of a fundamental domain $F^d$ for the
$\ZZ$-action and replace the intersection $U^d$ by $X^{d-1}
\coloneqq \partial F^d$, using the invariance of $\ZZ$-equivariant de
Rham cohomology under equivariant homotopy.
In Section~\ref{sect:MV for involutive spaces} we thereby establish

\begin{proposition}
Let $(X^d,\tau^d)$ be a closed and oriented involutive $d$-dimensional
manifold with a fundamental domain $F^d$ and define $(X^{d-1},\tau^{d-1}) = (\partial F^d, \tau^d_{|\partial F^d})$.
Let $U^d$ be a small thickening of $X^{d-1}$ as made precise in Proposition~\ref{st:HdR SES for fundamental domain}.
If $[\omega] \in \rmH^d_{\dR\, \mp}(X^d)$ is in $\ker(\imath_d^*)$
for $V^d = F^d \cup U^d$, there exists a preimage $[\rho] \in
\rmH^{d-1}_{\dR\, \pm}(X^{d-1})$ of $[\omega]$ under the connecting
homomorphism.
Any such preimage satisfies
\begin{equation}
	\int_{X^d}\, \omega= \int_{X^{d-1}}\, \rho \ .
\end{equation}
\end{proposition}

This implies that any Real or $\ZZ$-equivariant form of top degree
whose restriction to the thickening $V^d$ is exact has a preimage
under the connecting homomorphism up to exact forms, and the integral of the original form over the whole manifold $X^d$ coincides with the integral of any such preimage over the boundary of the fundamental domain $X^{d-1}$.
The localisation technique elaborated on in
Section~\ref{sect:Localisation via MV} then applies to situations
where the preimage again satisfies the same assumptions on the
involutive manifold $(X^{d-1},\tau^{d-1})$ so that we can iterate the choice of a preimage under the connecting homomorphism.
An important application is when $d=3$ and $(X^3, \tau^3)$ is the
$3$-dimensional Brillouin torus of a time-reversal symmetric
topological insulator, and we wish to compute the integral of the Real
form $\omega=w_\varphi^*H_{\sfU(m)}$; in this context a fundamental
domain is often referred to as an \emph{effective
  Brillouin zone}. In this way our localisation technique reproduces
not only the discrete Kane-Mele invariant of the topological
insulator, but also the expression of the $3$-dimensional invariant
(the ``strong'' Kane-Mele invariant) as a product of $2$-dimensional
invariants (the ``weak'' Kane-Mele invariants) defined on the
$\tau^3$-invariant invariant $2$-tori in $X^3$~\cite{FKM}, and the
geometric expression as a product of Berry phases (holonomies of the
trace of the Berry connection) around the invariant circles in
$X^3$~\cite{Fu-Kane:Time-reversal_polarisation}. We summarise this
unified perspective on the local geometric formulations of the
Kane-Mele invariant in Section~\ref{sect:unifiedKM}.

In Section~\ref{sect:KM from BGrb hol} we show that the combination of
a Jandl gerbe on a $3$-manifold $X^3$ with orientation-reversing involution
also yields a localisation principle, which follows from the
$2$-category theory of bundle gerbes. This is in some sense a
geometric refinement of the localisation in de~Rham cohomology: the differential forms used there can be understood as the curvatures of the higher bundles occurring here.
While the combination of Jandl gerbes with orientation-reversing involutions in two dimensions has already been studied to quite some extent (see for instance~\cite{Waldorf--Thesis,SSW:Unoriented_WZW_and_hol_of_BGrbs,GSW:BGrbs_for_orientifold_Sigma_models}), here we consider holonomies of Jandl gerbes on surfaces with orientation-preserving involutions.
These surfaces arise as the fundamental domains of
orientation-reversing involutions on $3$-manifolds, such as the
$3$-dimensional Brillouin torus of a time-reversal symmetric
topological insulator.
The holonomy of a bundle gerbe enters through a higher bundle version of
Stokes' Theorem: setting $X^3 = \TT^3$ with involution given by
inversion and $F^3$ an effective Brillouin zone with $X^2 \coloneqq \partial F^3$, we have 
\begin{equation}
\label{eq:first localisation steps intro}
	\rmK\rmM\big(E^-,\hat\tau^3\big)= \exp \bigg( \pi\, \iu\, \int_{X^3}\, w_\varphi^*H_{\sfU(m)} \bigg)
	= \exp \bigg( 2 \pi\, \iu\, \int_{F^3}\, w_\varphi^*H_{\sfU(m)} \bigg)
	= \hol\big( w_\varphi^* \CG_{\rm basic}, X^2 \big)\ .
\end{equation}

Let $F^2$ be a fundamental domain for $\tau^3_{|X^2}$ and set $X^1 \coloneqq \partial F^2$.
We show that this holonomy can be reduced to the holonomy of an equivariant line bundle with connection $\sfR K(A,T)$ on $X^1$.
This line bundle is not canonically determined by the Jandl gerbe but
depends on the choice of a trivialisation $T$ of it over $X^2$; it arises in a completely analogous way to how we obtained the sewing matrix $w_\varphi$ from a trivialisation of the Bloch bundle $E^-$ in~\eqref{eq:w_varphi intro}.
In analogy to the localisation technique for differential forms, where each iteration leads to a choice of a new form of lower degree, so can we iterate the construction of the equivariant line bundle $\sfR K(A,T)$ to obtain a $\sfU(1)$-valued function $f_{(A,T,j)}$ on $X^1$ which satisfies $f_{(A,T,j)} \circ \tau^1 = f_{(A,T,j)}^{-1}$ where $\tau^1 \coloneqq \tau^3_{|X^1}$, i.e. $f_{(A,T,j)}$ is a Real function.
The function $f_{(A,T,j)}$ is again obtained by choosing a trivialisation $j$ of $\sfR K(A,T)$ over $X^1$ and then employing a definition as in~\eqref{eq:w_varphi intro}.
This yields the localisation of the $3$-dimensional Kane-Mele
invariant~\eqref{eq:first localisation steps intro} in terms of bundle
gerbes, which is a consequence of

\begin{theorem}
Let $(X^2, \tau^2)$ be a closed and oriented $2$-dimensional
involutive manifold that can be made the top componentwise involutive
triple of an oriented $2$-dimensional filtered involutive manifold
$\big\{(X^d, \tau^d, F^d)\big\}_{d=0,1,2}$ (see Definition~\ref{def:filtered involutive manifold}).
Let $(\CG, (A,\alpha), \psi)$ be a Jandl gerbe on $X^2$.
\begin{myenumerate}
	\item The holonomy of $\CG$ around $X^2$ squares to $1$:
	\begin{equation}
		\hol(\CG, X^2) \ \in \ \ZZ\ .
	\end{equation}
	
	\item The holonomy localises completely as
	\begin{equation}
	\label{eq:Real BGrb holonomy reduction intro}
		\hol\big(\CG, X^2\big)
		= \hol \big(\sfR K(A,T), X^1 \big)
		= \prod_{\lambda \in X^0}\, f_{(A,T,j)}(\lambda)\ ,
	\end{equation}
	where $X^0 \subset X^2$ is the set of fixed points of the involution $\tau^2$, which is independent of the extension of $(X^2, \tau^2)$ to an oriented filtered involutive manifold.
	The expression~\eqref{eq:Real BGrb holonomy reduction intro} is independent of the choice of Jandl structure and all trivialisations.
\end{myenumerate}
\end{theorem}

Our representation of the Kane-Mele invariant as a bundle gerbe holonomy again expresses it in terms of $2$-dimensional weak Kane-Mele invariants over $X^2$~\cite{Fu-Kane:Time-reversal_polarisation,Gawedzki:2017whj}.
The intermediate expression of the Kane-Mele invariant~\eqref{eq:first
  localisation steps intro} as the holonomy of an equivariant line
bundle on $X^2$ is equivalent to the formulation of the discrete
Pfaffian formula as a geometric obstruction which expresses it as a
kind of Berry phase~\cite{Fu-Kane:Time-reversal_polarisation} (see
also~\cite{KLWK}); the relation between this Berry phase formula and
the continuous representation as a Wess-Zumino-Witten amplitude is
derived in~\cite{Monaco-Tauber:Berry_WZW_FKM}. However, these
representations are only derived for the Kane-Mele invariant in two
dimensions, whereas our computation of the Kane-Mele invariant stems
from purely $3$-dimensional considerations. We elaborate on and summarise
these local expressions for the discrete invariant within our
formalism in Section~\ref{sect:unifiedKM}.

\newpage

\part{Topological Insulators}
\label{part:TPMs}

In this part we will review the definition of a {band
  insulator} from~\cite{Freed-Moore--TEM}, which provides a formal and
general mathematical framework for such quantum phases of matter suitable for a classification in terms of K-theory.
We will then restrict our attention to the special case of a lattice
system with additional time-reversal symmetry, and show explicitly how the general formalism boils down to a more familiar description of the physics.

In Section~\ref{sect:symmetries of Galiliean spacetimes and QM} we start by introducing symmetries of Galilean relativity and quantum mechanics. We describe the geometric background data, recalling the notions of Galilean spacetimes and transformations, together with crystals, and then recall the corresponding description of symmetries in gapped quantum systems.
These two formalisms for symmetries are then combined in Section~\ref{sect:Gapped topological phases} to yield the notion of a band insulator, as well as the group of reduced topological phases modelled by systems of this kind.
We subsequently specialise to the case where the only symmetry in addition to lattice translations is time-reversal symmetry, to arrive at time-reversal symmetric topological insulators.

Throughout we follow closely the treatment of~\cite{Freed-Moore--TEM}, where a
complete and detailed account can be found. However, we deviate
from~\cite{Freed-Moore--TEM} at certain key points in order to specialise
the analysis in a way more suitable to our applications later on in
this paper, and to adapt slightly different nomenclature and notation better tailored to our needs; for example, we first investigate symmetries and their
properties in general quantum systems, deriving their fundamental
$\ZZ$-gradings without reference to any underlying Galilean spacetime. We
further primarily concentrate on
the specific cases we are most interested in and treat these in
some detail. 

\section{Symmetries in Galilean relativity and quantum mechanics}
\label{sect:symmetries of Galiliean spacetimes and QM}

\subsection{Galilean spacetimes and symmetries}
\label{sect:Galilean_spacetimes_and_symmetries}

The entity we start from is a {Galilean spacetime}.
Such a spacetime has an underlying manifold $M$ with a simply action of a finite-dimensional real vector space $W$.
In other words, $M$ is an affine space over $W$.

\begin{definition}[Galilean spacetime, Euclidean space]
\label{def:Galilean_structure_and_spacetime}
A \emph{Galilean structure} on an affine space $(M,W)$ consists of a choice of a subspace $V \subset W$ of codimension one together with a positive definite inner product $g_V$ on $V$, and a positive definite inner product $g_{W/V}$ on the line $W/V$.
We abbreviate these data as $\Gamma = (W, V, g_V, g_{W/V})$.
A \emph{Galilean spacetime} $(M,\Gamma)$ is an affine space with a Galilean structure.\\
A \emph{Euclidean structure} on an affine space $(N,V)$ is a positive definite inner product on $V$.
An affine space with a Euclidean structure is called a \emph{Euclidean space}.
\end{definition}

The action of the distinguished subspace $V \subset W$ generates the \emph{space foliation} of $M$.
Its leaves are submanifolds of codimension one, diffeomorphic to $V$
via the choice of an origin on the leaf, and Euclidean since $V$ is
endowed with a positive definite inner product.

\begin{definition}[Time direction, time orientation]
A \emph{time direction} on a Galilean spacetime $(M,\Gamma)$ is a splitting of the short exact sequence $ 0 \rightarrow V \rightarrow W \rightarrow W/V \rightarrow 0$, i.e. an inclusion $W/V \hookrightarrow W$ which is right inverse to the projection $W \to W/V$.
The image of such a splitting will usually be denoted as a $1$-dimensional subspace $U \subset W$, so that there is a canonical isomorphism $W \cong U \oplus V$.\\
A \emph{time orientation} is a choice of orientation on $U$.
\end{definition}

A time direction yields a foliation of $M$ into $1$-dimensional submanifolds diffeomorphic to $U$ which is transversal to the space foliation.
It does not, however, endow the timelike leaves with an orientation.
That is, a time direction does not specify a time orientation, which is extra data.
Together with the inner product on $U$, a time orientation is the same as an isometric isomorphism $U \arisom \FR$.

An \emph{affine transformation} of $M$ is a diffeomorphism $f \colon M \arisom M$ such that, for every $w \in W$ and $x_0 \in M$, setting $\pi(f)(w) = f(x_0+w) - f(x_0)$ defines a map $\pi(f) \in \Aut(W)$ independently of the choice of $x_0$.
Thus, denoting the group of all such diffeomorphisms by $\Aff(M,W)$, there is a short exact sequence of topological groups
\begin{equation}
	\xymatrix{
		1 \ar@{->}[r] & W \ar@{->}[r] & \Aff(M,W)
                \ar@{->}[r]^-{\pi} & \Aut(W) \ar@{->}[r] & 1\ .
	}%
\end{equation}

\begin{definition}[Galilean transformation]
\label{def:Galilean_trafos}
The group $\Aut(M,\Gamma)$ of \emph{Galilean transformations} of a Galilean spacetime $(M,\Gamma)$ consists of those affine transformations $f\in\Aff(M,W)$ such that $\pi(f)$ preserves $V \subset W$ acting orthogonally on $(V, g_V)$, and such that the quotient map $\pi(f)/V$ acts orthogonally on $(W/V, g_{W/V})$.
\end{definition}

Again we find a short exact sequence
\begin{equation}
	\xymatrix{
		1 \ar@{->}[r] & W \ar@{->}[r] & \Aut(M,\Gamma)
                \ar@{->}[r]^-{\pi} & \Aut(W,\Gamma) \ar@{->}[r] & 1\ ,
	}
\end{equation}
where $\Aut(W,\Gamma) \subset \Aut(W)$ is the subgroup of linear automorphisms of $W$ that preserve the additional data as in Definition~\ref{def:Galilean_trafos}.

The leaves of the space foliation are orientable, so that we can distinguish Galilean transformations which preserve or reverse the orientation of the leaves.
This property depends solely on $\pi(f)$.
Moreover, the action of $\pi(f)/V$ can either reverse or preserve orientation on $W/V$.\footnote{Deciding this requires only {orientability} and not an actual choice of orientation.}
Therefore, there are surjective group homomorphisms
\begin{equation}
	t,\, p \colon \Aut(W,\Gamma) \longrightarrow \ZZ \qquad
        \mbox{and} \qquad t
        \circ \pi,\ p \circ \pi  \colon \Aut(M,\Gamma) \longrightarrow
        \ZZ\ ,
\end{equation}
whose value is $-1$ precisely on those transformations reversing time or space orientation, respectively.
The letters $p$ and $t$ are indicative: $p$ is for `parity', while $t$ is for `time'.
In the following we will often abbreviate $t \circ \pi$ by $t$ and $p \circ \pi$ by $p$.

We call a $\ZZ\times\ZZ$-grading a $\ZZ$-bigrading, and we denote by $\TopGrp_{\ZZ^2}$ the category of $\ZZ$-bigraded topological groups with continuous group homomorphisms which preserve both $\ZZ$-gradings.
Thus $(\Aut(M,\Gamma),t,p)$ and $(\Aut(W,\Gamma),t,p)$ are $\ZZ$-bigraded topological groups.

It will be desirable to act with a group $\sfG$ on a Galilean spacetime in a way which is compatible with the Galilean structure.
Such group actions factor through a group homomorphism $\gamma \colon \sfG \rightarrow \Aut(M,\Gamma)$.
The group elements which act as pure translations on $M$ are given by $\gamma^{-1}(W \cap \gamma(\sfG))$.
However, such group elements might still act nontrivially on some internal degrees of freedom.
For instance, we could consider $\sfG$ acting on $M \brtimes V_\intern$ instead of $M$, where $V_\intern$ is a vector space of internal degrees of freedom (for example, spin).
We hence assume that to each $g \in \sfG$ with $\gamma(g) \in W$ there is always an element in $\sfG$ which acts trivially on the internal degrees of freedom and gives the same translation.
That is, we demand the existence of a group homomorphism $j \colon W \cap \gamma(\sfG) \hookrightarrow \sfG$, and that this is an inclusion of a normal subgroup.

\begin{definition}[Galilean symmetry group]
A \emph{Galilean symmetry group} or \emph{Galilean group action} on a
Galilean spacetime $(M,\Gamma)$ is a triple $(\sfG,\gamma,j)$, where $\gamma \colon \sfG \rightarrow \Aut(M,\Gamma)$ is a homomorphism of topological groups, and $j \colon W \cap \gamma(\sfG) \hookrightarrow \sfG$ is an inclusion of a normal subgroup which is compatible with the gradings on $\Aut(M,\Gamma)$ and the pullbacks of these gradings along $\gamma$ to~$\sfG$.
\end{definition}

We could have equivalently started from a $\ZZ$-bigraded group $(\sfG,\phi_\sfG, \sfc_\sfG)$ instead of a generic topological group $\sfG$, requiring $\gamma$ and $j$ to be morphisms in $\TopGrp_{\ZZ^2}$.
There is no generic constraint on the induced gradings on $\sfG$, i.e. each of them might be trivial.
In particular, translations in $\Aut(M,\Gamma)$ have both gradings trivial, as they preserve time and space orientations.

If the Galilean spacetime $(M,\Gamma)$ admits a time direction, so
that $W = U \oplus V$, the linear automorphisms $L\in\Aut(W,\Gamma)$ which preserve the time direction have the block form
\begin{equation}
	L =
	\begin{pmatrix}
		\pm\, 1 & 0 \\
		 0 & A
	\end{pmatrix} \ ,
\end{equation}
where $A$ is an orthogonal transformation of $V$.
We write $\Aut(M,\Gamma,U)$ for this subgroup of $\Aut(M,\Gamma)$. In particular, even though the time {direction} is preserved, these transformations may still reverse its orientation.

\subsection{Crystals and Bravais lattices}

Intuitively, a crystal is a lattice of points in space with certain periodicities.
In an affine space, however, there is no specified origin and so a space crystal cannot be a lattice in the mathematical sense.
Nevertheless, we can describe a crystal as an orbit of a point under a lattice of space translations $\Pi \subset V$, called a \emph{Bravais lattice}.
We then have to describe how each point of the space crystal moves in time.

\begin{definition}[Crystal in Galilean spacetime]
\label{def:crystal}
Let $(M,\Gamma,U)$ be a Galilean spacetime with a time direction.
A \emph{crystal} in $M$ is a pair $(C,\Pi)$ consisting of a subset $C \subset M$ such that the subset $\Pi \subset V$ of those spatial translations that preserve $C$ form a lattice of full rank in $V$.
\end{definition}

A crystal in a Galilean spacetime is of course extra structure, so
that the symmetries which preserve this data are generally smaller than the original Galilean symmetry group $\Aut(M,\Gamma,U)$.

\begin{definition}[Crystal symmetries]
\label{def:Crystal_symmetries}
For a crystal $(C,\Pi)$ in a Galilean spacetime $(M,\Gamma,U)$, denote by $\sfG(C) \subset \Aut(M,\Gamma,U)$ the subgroup of Galilean transformations which preserve $C$ as a set.
\begin{myenumerate}
	\item
	The time translations $U \subset W$ are naturally contained in $\sfG(C)$ by the definition of a crystal.
	The quotient group $\sfG(C)/U$ is called the \emph{crystallographic group} of $(C,\Pi)$.
	\item
	The lattice translations $\Pi$ are naturally contained in $\sfG(C)/U$, and the quotient group $\widehat{\sfP}(C)$ in the short exact sequence of groups
	\begin{equation}
	\xymatrix{
		1 \ar@{->}[r] & \Pi \ar@{->}[r] & \sfG(C)/U \ar@{->}[r] & \widehat{\sfP}(C) \ar@{->}[r] & 1
	}
	\end{equation}
	is called the \emph{magnetic point group} of the crystal.
	That is,
	\begin{equation}
		\widehat{\sfP}(C) \cong \sfG(C)/(U \oplus \Pi)\ .
	\end{equation}
	\item
	Since $\widehat{\sfP}(C) \subset \Aut(W,\Gamma)/U$, it carries the restriction of $t \colon \Aut(W,\Gamma) \rightarrow \ZZ$.
	We call the time-preserving component of $\widehat{\sfP}(C)$
        the \emph{point group} $\sfP(C)$ of the crystal $(C,\Pi)$.
\end{myenumerate}
\end{definition}

\subsection{Quantum symmetries}
\label{sect:QM_symmetries}

The preceding definitions were aimed at creating a general flat background setting which underlies nonrelativistic quantum systems.
We would now like to consider quantum mechanics on such spacetimes.
We start with some recollections about symmetries in quantum theory.
It is crucial here to strictly adhere to the axiom of quantum
mechanics that, given a quantum system whose Hilbert space is
$(\scH,\langle-,-\rangle_{\scH})$, the space of \emph{states} of that
system is not $\scH$ but $\PP\scH$, the projective space of $\scH$,
since one can only measure \emph{transition probabilities} rather than transition amplitudes.
We assume that $\dim(\scH) > 1$ and that $\scH$ is separable, but it may be finite-dimensional.

Transition amplitudes are not well-defined on states, but the resulting probabilities are.
They correspond to a symmetric function
\begin{equation}
	\mathrm{prob} \colon \PP\scH \times \PP\scH \longrightarrow
        [0,1] \ ,\quad ([\psi], [\phi]) \longmapsto \frac{| \< \psi,
          \phi \>_\scH |^2}{\<\psi,\psi\>_\scH\, \<\phi,\phi\>_\scH} \ ,
\end{equation}
and they
form the only measurable object in the kinematical data of a quantum system.

\begin{definition}[Projective quantum symmetries]
Given a quantum system with underlying Hilbert space $\scH$, its \emph{projective quantum symmetries} comprise the group
\begin{equation}
	\Aut_{\qm}(\PP\scH) = \big\{ f \in C(\PP\scH,\PP\scH)\, \big| \, f^* \mathrm{prob} = \mathrm{prob} \big\}\ ,
\end{equation}
where $C(\PP\scH,\PP\scH)$ denotes the topological space of continuous functions from $\PP\scH$ to itself.\footnote{As a mapping space of topological spaces it carries the compact-open topology.}
\end{definition}

Even though this defines the possible transformations of a quantum system based on $\scH$, it is desirable to specify these transformations in a more familiar language.
This has been achieved in Wigner's Theorem~\cite{vNW--Spektren_Drehelektron,Wigner--Group_theory} (see also~\cite{Freed--On_Wigner} for two geometric proofs).
A continuous map $T \colon \scH \rightarrow \scH$ is called
\emph{antilinear} if it satisfies $T(\lambda\, \psi + \phi) =
\overline{\lambda}\ T\psi + T\phi$ for all $\lambda \in \FC$ and
$\psi, \phi \in \scH$, and \emph{antiunitary} if it is antilinear and
additionally satisfies
\begin{equation}
\< T\psi, T\phi\>_\scH = \overline{\<\psi, \phi\>_\scH}
\end{equation}
for all $\psi, \phi \in \scH$.
We define $\Aut_{\qm}(\scH)$ to be the group of continuous maps $T \colon \scH \rightarrow \scH$ which are either unitary or antiunitary.

\begin{theorem}[Wigner]
If $\dim(\scH) > 1$ there is a short exact sequence of topological groups
\begin{equation}
\label{diag:Wigner_extension}
	\xymatrixcolsep{1cm}
	\xymatrix{
		1 \ar@{->}[r] & \sfU(1) \ar@{->}[r]^-{(-) \cdot \One} & \Aut_{\qm}(\scH) \ar@{->}[r]^-{q} & \Aut_{\qm}(\PP\scH) \ar@{->}[r] & 1\ ,
	}
\end{equation}
where the inclusion of $\sfU(1)$ into $\Aut_{\qm}(\scH)$ is via $\lambda \mapsto \lambda\, \One$.
\end{theorem}

The two components of $\Aut_{\qm}(\scH)$ form the parts of a
$\ZZ$-grading $\phi_\scH:\Aut_{\qm}(\scH)\to\ZZ$, defined as $\phi_\scH(T) = 1$ precisely when $T$ is unitary.
It is important to note that
\begin{equation}
	T \circ \lambda\, \One = \lambda^{\phi_\scH(T)}\, \One \circ T
\end{equation}
for all $\lambda \in \sfU(1)$.
Thus the short exact sequence in Wigner's Theorem is not a central extension.

If the Hilbert space $\scH$ is infinite-dimensional, as will be the case in later applications, then both components of homogeneous degree are contractible.
This is essentially Kuiper's Theorem, and it applies in both the norm topology and the compact-open topology on the spaces of continuous (anti)linear operators on $\scH$.

The dynamics of the quantum system is generated by a \emph{Hamiltonian}, i.e. a linear, possibly unbounded, selfadjoint operator $H \colon \scD(H) \rightarrow \scH$, where $\scD(H) \subset \scH$ is the domain of $H$.\footnote{Usually $\scD(H)$ will be required to be dense in $\scH$.}
A choice of a Hamiltonian singles out a special subgroup of quantum symmetries, namely those which are symmetries of the Hamiltonian \footnote{In writing this review section, we realised that one can derive the properties of $\sfc$ from scratch without any further assumptions; thus, we briefly deviate from the presentation in~\cite{Freed-Moore--TEM}.}.
For $T \in \Aut_{\qm}(\scH)$, compatibility with $H$ then means that $T$ preserves $\scD(H)$ and that the induced transformations of $T\, H$ and $H\, T$ on $\PP\scH$ coincide.
That is
\begin{equation}
	[T \, H \psi] = [H \, T \psi] \ \in \ \PP\scH
\end{equation}
for all $\psi \in \scD(H) \backslash \ker(H)$. This implies that for
each $T$ and $\psi$ with $H \psi \neq 0$ there is a nonzero complex number $\sfc(T,\psi) \in \FC^\times:=\FC\setminus\{0\}$ such that
\begin{equation}
	T\, H \psi = \sfc(T,\psi)\, H\,T \psi\ .
\end{equation}
If moreover $H \phi \neq 0$, we can compute
\begin{equation}
\begin{aligned}
	\< T\phi, T\,H \psi\>_\scH &= \sfc(T,\psi)\, \< T \phi, H\,T \psi \>_\scH\\*[4pt]
		&= \sfc(T,\psi)\, \< H\,T \phi, T \psi \>_\scH\\*[4pt]
		&= \sfc(T,\psi)\, \overline{\sfc(T,\phi)}{}^{-1}\, \< T\,H \phi, T \psi \>_\scH\ .
\end{aligned}
\end{equation}
Depending on $\phi_{\scH}(T)$, the left-hand side is either equal to $\< \phi, H \psi\>_\scH$ or $\overline{\< \phi, H \psi\>_\scH}$.
The right-hand side, as a consequence of the (anti)unitarity of $T$ and the selfadjointness of $H$, either reads as $\sfc(T,\psi)\, \overline{\sfc(T,\phi)}{}^{-1}\, \< \phi, H \psi \>_\scH$ or $\sfc(T,\psi)\, \overline{\sfc(T,\phi)}{}^{-1}\, \overline{\< \phi, H \psi \>_\scH}$, respectively.
Since $\scD(H)$ is dense in $\scH$, we deduce that
\begin{equation}
	\sfc(T,\psi) = \overline{\sfc(T,\phi)}
\end{equation}
for all $\psi, \phi \in \scD(H) \backslash \ker(H)$. In particular, the case $\phi = \psi$ then shows that $\sfc(T,\psi) \in \FR$, and hence
\begin{equation}
	\sfc(T,\psi) = \sfc(T,\phi) \ \in \ \FR 
\end{equation}
for all $\psi, \phi \in \scD(H) \backslash \ker(H)$. Thus we have $ T\,H \psi = \sfc(T)\, H\, T \psi$ for all $\psi\in\scD(H)$ such that $H \psi \neq 0$, where we define $\sfc(T) = \sfc(T,\psi)$ for any $\psi \in \scD(H) \backslash \ker(H)$.
If the system has nontrivial quantum dynamics, such $\psi$ always exist.
Using the reality of $\sfc(T)$ it is straightforward to deduce
\begin{equation}
	\sfc(T)\, \sfc(S) = \sfc(T\,S) \qquad \mbox{and} \qquad \sfc(T^{-1}) = \sfc(T)^{-1}\ .
\end{equation}

For a physical symmetry $T$ of the Hamiltonian $H$ one has to demand that $T\,
H$ and $H\, T$ cannot be distinguished in any possible transition probability.
That is, for every $\psi, \phi \in \scD(H)$ we demand
\begin{equation}
	| \< \phi, H\,T \psi \>_\scH |^2 = | \< \phi, T\,H \psi\>_\scH |^2 = \sfc(T)^2\, | \< \phi, H\,T \psi \>_\scH |^2\ ,
\end{equation}
leading to the restriction $\sfc(T) \in \ZZ$.

The map $T \mapsto \sfc(T)$ depends continuously on $T$:
Let $T_{(-)} \colon [0,1] \rightarrow \Aut_\qm(\scH)$, $t \mapsto T_t$ be a continuous family of unitary or antiunitary transformations which are compatible with $H$.
Then the function $t \mapsto \< T_t\, \phi, H \psi \>_\scH$ is also continuous for any $\phi, \psi \in \scD(H)$.
Using the selfadjointness of $H$ and the definition of $\sfc(T)$, we obtain
\begin{equation}
	\< T_t\, \phi, H \psi \>_\scH = \sfc(T_t)^{-1}\, \< T_t\, H \phi, \psi \>_\scH = \sfc(T_t)\, \< T_t\, H \phi, \psi \>_\scH\ .
\end{equation}
Therefore the function $t \mapsto \sfc(T_t)\, \< T_t \, H \phi,\psi \>_\scH$ is continuous.
By the continuity of $T_{(-)}$, the second factor in this function is
continuous, whence the full function is continuous for all
$\phi,\psi\in\scH$ if and only if $t \mapsto \sfc(T_t)$ is continuous.

\begin{definition}[Quantum symmetries]
\label{def:Aut_qm}
Given a Hilbert space $\scH$ together with a (densely defined) Hamiltonian $H \colon \scD(H) \rightarrow \scH$, the \emph{group of quantum symmetries} of $(\scH,H)$ is
\begin{equation}
	\Aut_{\qm}(\scH,H) := \big\{ T \in \Aut_{\qm}(\scH)\, \big| \, T \scD(H) \subset \scD(H),\, T\,H = \pm\, H\,T \big\}\ .
\end{equation}
This comes endowed with the continuous group homomorphism $\sfc \colon \Aut_{\qm}(\scH,H) \rightarrow \ZZ$, and together with the $\ZZ$-grading $\phi_\scH$ inherited from $\Aut_{\qm}(\scH)$, the quantum symmetries of $(\scH,H)$ form a $\ZZ$-bigraded topological group $(\Aut_{\qm}(\scH,H), \phi_\scH, \sfc) \in \TopGrp_{\ZZ^2}$.
\end{definition}

Elements $T$ of $\Aut_\qm(\scH,H)$ therefore either commute or anticommute with $H$, depending on $\sfc(T) = \pm\, 1$, and act either unitarily or antiunitarily on $\scH$, depending on $\phi_\scH(T) = \pm\, 1$.

\begin{remark}
Quantum symmetries which anticommute with the Hamiltonian are often called \emph{chiral symmetries} of $(\scH,H)$.
Examples are rotations about a given axis by angle $\pi$ and certain
Hamiltonians built from the corresponding angular momentum operators,
since on an eigenspace of $J$ we have $R(\pi)\, J\, R(\pi)^{-1} =
(-1)^{2l}\, J$, where $R(\alpha)$ denotes the rotation about the given
axis by an angle $\alpha\in[0,2\pi)$, and $J$ is the angular momentum operator with $l \in \frac{1}{2}\, \NN_0$ the eigenvalue of $J$ on the respective subspace.
See~\cite{Bhattacharya-Kleinert:Chiral_quantum_symmetries} for an
account of chiral symmetries in quantum mechanics together with some
explicit examples.
\qen
\end{remark}

The existence of chiral symmetries $T$ of a Hamiltonian $H$ has far-reaching consequences.
If $T \in \Aut_{\qm}(\scH,H)$ and $\psi$ is an eigenvector of $H$ with
energy eigenvalue $\varepsilon_\psi \in \FR$, then
\begin{equation}
	H \, T \psi = \sfc(T)^{-1}\, T \, H \psi = \sfc(T)\, \varepsilon_\psi\, T \psi\ .
\end{equation}
Moreover, one has
\begin{equation}
\begin{aligned}
	T\, \exp\Big(\, \frac{\iu}{\hbar}\, t \, H \Big) &= \sum_{n=0}^{\infty}\, \frac{t^n}{\hbar^n\, n!}\, T\ \iu^n\, H^n\\*[4pt]
		&= \sum_{n=0}^{\infty}\, \frac{t^n}{\hbar^n\, n!}\, \phi_\scH(T)^n\, \sfc(T)^n\, \iu^n\, H^n\ T\\*[4pt]
		&= \exp\Big(\, \frac{\iu}{\hbar}\, \sfc(T)\, \phi_\scH(T)\, t\, H \Big)\ T\ .
\end{aligned}
\end{equation}
Let $\sfc(T)\, \phi_\scH(T) = \colon \sft(T) = \pm\, 1$, and let $U_{H,t} =
\exp( \frac{\iu}{\hbar}\, t\, H )$ denote the time evolution operator
of the quantum system $(\scH,H)$. Then
\begin{equation}
	T\, U_{H,t} = U_{H, \, \sft(T)\, t}\, T\ .
\end{equation}
From these simple facts we easily deduce
\begin{proposition}
\label{st:time-reversal_and_spectrum_of_H}
Let $(\scH,H)$ be a quantum system.
\begin{myenumerate}
\item
If $(\scH,H)$ admits a chiral symmetry $T \in \Aut_{\qm}(\scH,H)$, then the spectrum of $H$ is symmetric in $\FR$, i.e. $\sigma(H) = - \sigma(H)$.
\item
A symmetry $T$ is \emph{time-reversing} if and only if $\sft(T) = \sfc(T)\, \phi_\scH(T) = -1$.
\item
If $\sigma(H)$ is not symmetric (for example, if $H$ is bounded from
below but not from above), there are no chiral symmetries $T$ in $\Aut_{\qm}(\scH,H)$.
However, there can still be antiunitary time-reversing symmetries.
\end{myenumerate}
\end{proposition}

If $T \in \Aut_{\qm}(\scH,H)$ then also $\lambda\, T \in \Aut_{\qm}(\scH,H)$ for every $\lambda \in \sfU(1)$.
Hence we have another short exact sequence
\begin{equation}
		\xymatrixcolsep{1cm}
	\xymatrix{
			1 \ar@{->}[r] & \sfU(1)
                        \ar@{->}[r]^-{(-)\cdot \One} & \Aut_{\qm}(\scH,H)
                        \ar@{->}[r]^-{q} & \Aut_{\qm}(\PP\scH,H)
                        \ar@{->}[r] & 1\ .
		}
\end{equation}
This can be taken as the definition of the projective symmetry group $\Aut_{\qm}(\PP\scH,H)$:
it is the image of the composition $\Aut_{\qm}(\scH,H) \hookrightarrow
\Aut_{\qm}(\scH) \xrightarrow{ \ q \ } \Aut_{\qm}(\PP\scH)$.

\subsection{Representations and gapped Hamiltonians}
\label{sect:Reps_and_gapped_Hamiltonians}

Often the symmetry group relevant to a system will not be precisely the quantum symmetry group $\Aut_{\qm}(\PP\scH,H)$.
For example, if we consider a quantum system $(\scH,H)$ associated with a Galilean spacetime $(M,\Gamma)$, the symmetries of the system have to be both quantum symmetries of $(\scH,H)$ and Galilean symmetries of $(M,\Gamma)$.

Any group action $\rho \colon \sfG \rightarrow \Aut_\qm(\PP\scH,H)$ on $\PP\scH$ allows us to pull back the $\sfU(1)$-fibration $\Aut_\qm(\scH,H) \rightarrow \Aut_\qm(\PP\scH,H)$ to $\sfG$.
In more concrete terms, Wigner's Theorem implies that any projective quantum symmetry transformation will have a $\sfU(1)$-orbit of lifts to either unitary or antiunitary operators on $\scH$, and we attach this $\sfU(1)$-fibre to the transformation.
In this way, we obtain a commutative diagram of topological groups
\begin{equation}
\label{eq:quantum rep diagram}
	\xymatrixcolsep{1cm}
	\xymatrixrowsep{1cm}
	\xymatrix{
		1 \ar@{->}[r] & \sfU(1) \ar@{->}[r]^-{\iota_\sfG} \ar@{=}[d] & \sfG_\qm \ar@{->}[r]^-{q_\sfG} \ar@{->}[d]^-{\rho_\qm} & \sfG \ar@{->}[d]^-{\rho} \ar@{->}[r] & 1\\
		1 \ar@{->}[r] & \sfU(1) \ar@{->}[r]^-{(-)\cdot \One} & \Aut_{\qm}(\scH,H) \ar@{->}[r]^-{q} & \Aut_{\qm}(\PP\scH,H) \ar@{->}[r] & 1
	}
\end{equation}
Via $\rho$ and $\rho_\qm$, the groups $\sfG$ and $\sfG_\qm$ inherit $\ZZ$-gradings as the pullbacks of $\phi_\scH$ and $\sfc$.
For notational convenience, we denote these grading morphisms on both groups by $\phi_\sfG$ and $\sfc_\sfG$.
By the commutativity of diagram~\eqref{eq:quantum rep diagram} we have 
$g\, \iota_\sfG(\lambda) = \iota_\sfG(\lambda)^{\phi_\sfG(g)}\, g$
for all $g \in \sfG_\qm$ and $\lambda\in\sfU(1)$.

\begin{definition}[Quantum extension]
A \emph{quantum extension} of $(\sfG, \phi_\sfG, \sfc_\sfG) \in \TopGrp_{\ZZ^2}$ is a triple $(\sfG_\qm, \iota_\sfG, q_\sfG)$ giving rise to a short exact sequence of topological groups
\begin{equation}
	\xymatrix{
		1 \ar@{->}[r] & \sfU(1) \ar@{->}[r]^-{\iota_\sfG} & \sfG_\qm \ar@{->}[r]^-{q_\sfG} & \sfG \ar@{->}[r] & 1
	}
\end{equation}
which satisfies the relation
\begin{equation}
	g\, \iota_\sfG(\lambda) = \iota_\sfG(\lambda)^{\phi_\sfG(g)}\,
        g \ ,
\end{equation}
for all $g \in \sfG_\qm$ and $\lambda\in\sfU(1)$. The extension
$\sfG_\qm$ naturally becomes a $\ZZ$-bigraded topological group with the pullbacks along $q_\sfG$ of the gradings on $\sfG$.
\end{definition}

Thus, in a representation of $\sfG_\qm$ on a quantum Hilbert space $\scH$, the grading $\phi_\sfG$ measures whether a symmetry transformation acts linearly or antilinearly on $\scH$.
The second grading $\sfc_\sfG$ is also naturally related to properties of the representation if $\scH$ is a $\ZZ$-graded Hilbert space.
By this we mean that there is a decomposition $\scH = \scH^+ \oplus \scH^-$ of $\scH$ into two orthogonal subspaces, but we do not assume that these subspaces are isomorphic to each other.
This splitting assumption on $\scH$ can be understood as follows.
We are interested in Hamiltonians for which there exists a distinguished energy $\varepsilon_{\scH,H} \in \FR$ with $\varepsilon_{\scH,H} \notin \sigma(H)$, called the \emph{Fermi energy} of $(\scH,H)$.
This is equivalent to saying that $H - \varepsilon_{\scH,H}\, \One$ has a {bounded} inverse $(H - \varepsilon_{\scH,H}\, \One)^{-1} \in \CB(\scH)$, where we denote the bounded linear operators on $\scH$ by $\CB(\scH)$.
We can always perform a shift in energy and consider $H' = H -
\varepsilon_{\scH,H} \, \One$ instead of $H$ as the Hamiltonian, and $\varepsilon_{\scH,H'} = 0$.
When we consider continuous families $\{(H_t, \varepsilon_{\scH,H_t})\}_{t \in [0,1]}$ of such systems, with $t \mapsto \varepsilon_{\scH,H_t}$ a continuous function, we can shift the energy $t$-dependently in a consistent continuous way.
Thus we can assume in general that $\varepsilon_{\scH,H} = 0$ in the first place.
For details on continuous families of Hilbert spaces and (unbounded)
operators, see~\cite[Appendix~D]{Freed-Moore--TEM}.
Continuous deformations of Hamiltonians play a key role in the theory
and classification of topological insulators that we consider later on.

\begin{definition}[Gapped quantum system]
A \emph{gapped quantum system} is a pair $(\scH,H)$ consisting of a Hilbert space $\scH$ and a selfadjoint operator $H \colon \scD(H) \rightarrow \scH$ such that $0 \notin \sigma(H)$, i.e. $H^{-1} \in \CB(\scH)$ exists and is bounded.
\end{definition}

In any gapped quantum system the Hilbert space $\scH$ is naturally $\ZZ$-graded:
we obtain a decomposition $\scH = \scH^+ \oplus \scH^-$ from the gap
in the spectrum $\sigma(H)$ as follows.
Let $\chi_{\FR_\pm}: \FR \rightarrow \{0,1\}$ denote the characteristic function of the subset $\FR_\pm \subset \FR$, where $\FR_+$ and $\FR_-$ denote the sets of positive and negative real numbers, respectively.
Using the functional calculus for selfadjoint operators on a Hilbert
space, we obtain two orthogonal spectral projection operators $P_\pm \coloneqq \chi_{\FR_\pm}(H)$ on $\scH$.
Their images yield a decomposition $\scH = P_+ \scH \oplus P_- \scH \eqqcolon \scH^+ \oplus \scH^-$.

\begin{definition}[Symmetry class]
\label{def:Gapped_system_with_symmetry_class}
Let $(\sfG_\qm, \iota_\sfG, q_\sfG)$ be a quantum extension of a
$\ZZ$-bigraded topological group $(\sfG, \phi_\sfG, \sfc_\sfG)$.
A \emph{gapped quantum system with symmetry class $(\sfG_\qm,
  \iota_\sfG, q_\sfG)$} is a gapped quantum system $(\scH,H)$ with a
homomorphism $\rho_\qm \in \TopGrp_{\ZZ^2}(\sfG_\qm ,
\Aut_\qm(\scH,H))$ of $\ZZ$-bigraded topological groups such that
$\rho_\qm\circ\iota_\sfG(\lambda)=\lambda\, \One$ for all
$\lambda\in\sfU(1)$, thus inducing a commutative diagram of the form \eqref{eq:quantum rep diagram}.
\end{definition}

A morphism of $\ZZ$-bigraded topological groups preserves both gradings.
Therefore $\rho_\qm$ automatically satisfies $\rho_\qm(g)\, H = \sfc_\sfG(g)\, H\, \rho_\qm(g)$ for all $g \in \sfG$.
By Definition~\ref{def:Aut_qm}, any quantum symmetry $T$ either
commutes or anticommutes with the Hamiltonian $H$, depending on
$\sfc(T) = \pm \, 1$.
Those $T$ with $\sfc(T) = -1$ reflect the spectrum of $H$; in particular, they act as odd operators on the $\ZZ$-graded Hilbert space $\scH = \scH^+ \oplus \scH^-$ since operators (anti)commute with $H$ if and only if they (anti)commute with all spectral projections of $H$.
Consequently, $\rho_\qm(g)$ is an even (respectively odd) operator
precisely if $\sfc_\sfG(g)$ is $+1$ (respectively $-1$); it is unitary if and only if $\phi_\sfG(g) = 1$, otherwise it is antiunitary.

For the special case where $H = H_\gr \coloneqq \One_{\scH^+} \oplus (- \One_{\scH^-})$ is the grading operator of $\scH$, gapped quantum systems with symmetry class $(\sfG_\qm, \iota_\sfG, q_\sfG)$ are precisely representations of $\sfG_\qm$ on $\ZZ$-graded Hilbert spaces.

\begin{definition}[Quantum representation]
Let $(\sfG, \phi_\sfG, \sfc_\sfG) \in \TopGrp_{\ZZ^2}$.
A \emph{quantum representation} of $\sfG$ on a (possibly
finite-dimensional) $\ZZ$-graded Hilbert space $\scH$ consists of a
quantum extension $(\sfG_\qm, \iota_\sfG, q_\sfG)$ of $(\sfG,
\phi_\sfG, \sfc_\sfG)$ together with a homomorphism of $\ZZ$-bigraded
topological groups $\rho_\qm \in \TopGrp_{\ZZ^2}(\sfG_\qm,
\Aut_\qm(\scH,H_\gr))$ such that
$\rho_\qm\circ\iota_\sfG(\lambda)=\lambda\, \One$ for all
$\lambda\in\sfU(1)$.
\end{definition}

\subsection{Irreducible quantum representations}
\label{sect:Decomp_reps}

Before proceeding to apply the constructions of this section to our
main physical setup, we need some analytic and geometric features of decompositions of quantum representations into irreducible representations.
First, recall that continuous unitary representations of a locally compact topological group $\sfG'$ decompose into irreducible representations (see for instance~\cite{Folland--Harmonic_analysis}).
If $\sfG'$ is abelian, these are $1$-dimensional and correspond to
characters of the group, i.e. their isomorphism classes are parameterised by the \emph{Pontryagin dual group}
\begin{equation}
	X_{\sfG'} \coloneqq \TopAb(\sfG', \sfU(1))\ ,
\end{equation}
where $\TopAb$ is the category of topological abelian groups.
The set $X_{\sfG'}$ is naturally a topological abelian group itself.

Let us now amend this to the case of quantum representations of $\sfG'$ where $(\sfG',\phi_{\sfG'}, \sfc_{\sfG'})$ is a locally compact $\ZZ$-bigraded topological abelian group.
We let $\TopAb_{\ZZ^2} \subset \TopGrp_{\ZZ^2}$ denote the full subcategory of $\ZZ$-bigraded topological abelian groups.
For a fixed quantum extension $(\sfG'_\qm, \iota_{\sfG'}, q_{\sfG'})$ of $\sfG'$, let $\hat{X}_\qm$ be the set of isomorphism classes of irreducible quantum representations of $\sfG'$ based on the extension $\sfG'_\qm$.

The case of interest to us later on will be when $\sfG' = \Pi$ is the lattice of space translations of a crystal $(C,\Pi)$ in a Galilean spacetime, and $\sfG'_\qm \cong \Pi \times \sfU(1)$ as $\ZZ$-bigraded topological groups.
Note that in that case we have $\phi_{\sfG'}(g'\,) = 1$ for all $g' \in \sfG'$, and assuming that the Hamiltonian is invariant under lattice translations (cf.~Section~\ref{sect:Gapped topological phases}), we then also have $\sfc_{\sfG'}(g') = 1$ for all $g' \in \sfG'$.

\begin{lemma}
\label{st:QM_extensions_and_Pontryagin_duals}
Let $(\sfG',\phi_{\sfG'}, \sfc_{\sfG'}) \in \TopAb_{\ZZ^2}$, and let $(\sfG'_\qm, \iota_{\sfG'}, q_{\sfG'})$ be a quantum extension of $\sfG'$ with $\phi_{\sfG'}(g'\,) = 1$ for all $g' \in \sfG'$.
Then, $\sfG'_{\qm}$ is abelian, and the morphism $\sfG'_\qm \to \sfG'$ is a central extension by $\sfU(1)$.
If in addition $\sfc_{\sfG'}$ is trivial, there is an inclusion
\begin{equation}
	\hat{X}_{\sfG'_\qm} = \big\{ \hat{\lambda} \in X_{\sfG'_\qm}\, \big| \, \hat{\lambda}(z \xi) = z \cdot \hat{\lambda}(\xi)\ \forall\ \xi \in \sfG'_\qm,\, z \in \sfU(1) \big\} \subset X_{\sfG'_\qm}\ ,
\end{equation}
and we endow $\hat{X}_{\sfG'_\qm}$ with the subset topology.
If the central extension $\sfG'_\qm \to \sfG'$ is trivial, i.e.~$\sfG'_\qm \cong \sfG' \times \sfU(1)$ as $\ZZ$-bigraded topological groups, then there is a homeomorphism
\begin{equation}
	\hat{X}_{\sfG'_\qm} \cong X_{\sfG'}\ .
\end{equation}
\end{lemma}

\begin{proof}
The inclusion is immediate from the fact that $\phi_{\sfG'}$ controls commutativity of general group elements with elements of the image of $\sfU(1)$.\\
The second claim is seen from
\begin{equation}
\hat{X}_{\sfG'_\qm} = \big\{ \lambda \in X_{\sfG' \times \sfU(1)}\, \big| \, \lambda(\xi, z) = z \cdot \lambda(\xi,1)\ \text{for all}\ \xi \in \sfG',\, z \in \sfU(1) \big\} \cong X_{\sfG'}\ ,
\end{equation}
where the assumption that the quantum extension is trivialised enters crucially.
\end{proof}

\begin{remark}
\label{rmk:G'_qu has trivial gradings}
We will from now on assume that $\sfG'_\qm$ has trivial $\ZZ$-bigrading $(\phi_{\sfG'}, \sfc_{\sfG'})$.
This will in particular carry over to the cases below, where $\sfG' \subset \sfG$ is a subgroup of a $\ZZ$-graded topological group.
\qen
\end{remark}

Heuristically, the decomposition of a given (continuous unitary) representation of a compact topological abelian group on a Hilbert space $\scH$ into its irreducible representations assigns a Hilbert space to each point of its Pontryagin dual; the fibre over an isomorphism class of irreducible representations is the sum of all the $1$-dimensional subspaces where the group acts via representations in that isomorphism class.
However, the ranks of the fibres may vary and so in general we do not obtain a Hilbert bundle in this way.
Moreover, we would like to consider the case $\sfG' = \Pi$, which is locally compact but not compact.
For representations of locally compact abelian groups on infinite-dimensional Hilbert spaces this statement has to be weakened to the existence of certain projection-valued measures.
The precise statement can be found in~\cite[Theorem 4.45]{Folland--Harmonic_analysis} and is stated here as

\begin{theorem}
\label{st:loc cpt reps and PrVMs}
Let $\sfG'$ be a locally compact abelian group and $\rho$ a continuous unitary representation of $\sfG'$ on a separable Hilbert space $\scH$. Let $\scF\scT
\colon \rmL^1(\sfG'\,) \rightarrow C_0(X_{\sfG'})$ denote the Fourier transform on $\sfG'$, where $\rmL^1(\sfG'\,)$ is the space of integrable functions on $\sfG'$ with respect to a fixed Haar measure.%
\footnote{This becomes a commutative Banach $*$-algebra when endowed with the convolution product, and $X_{\sfG'}$ can be identified with the Gelfand spectrum of $\rmL^1(\sfG'\,)$~\cite{Folland--Harmonic_analysis}.}
Then there exists a unique regular $\scH$-projection-valued measure $P_{\rho}$ on $X_{\sfG'}$ such that
\begin{equation}
\rho(\xi) = \int_{X_{\sfG'}}\, \lambda(\xi)\ \dd P_{\rho}(\lambda) \qquad \text{and} \qquad
	\rho(f) = \int_{X_{\sfG'}}\, \scF\scT(f)(\lambda^{-1})\ \dd
        P_{\rho}(\lambda) \ ,
\end{equation}
for all $\xi \in \sfG'$ and $f \in \rmL^1(\sfG'\,)$. 
\end{theorem}

\begin{remark}
\begin{myenumerate}
	\item
	Theorem~\ref{st:loc cpt reps and PrVMs} is not directly adaptable to $\ZZ$-bigraded topological abelian groups because projection-valued measures do not implement antilinearity.
	Thus for us it applies to $\ZZ$-bigraded groups $(\sfG',\phi_{\sfG'}, \sfc_{\sfG'})$ with trivial grading $\phi_{\sfG'}$ only.
	
	\item
	If for a group $\sfG'$ with $\phi_{\sfG'} = 1$ (so that its quantum extensions are central extensions) we consider a quantum extension $(\sfG'_\qm, \iota_{\sfG'}, q_{\sfG'})$ instead, we can apply Theorem~\ref{st:loc cpt reps and PrVMs} to $\sfG'_\qm$.
	This yields a projection-valued measure on $X_{\sfG'_\qm}$ decomposing the representation of $\sfG'_\qm$ on $\scH$.
	However, for a quantum representation the image of $\sfU(1)$ in $\sfG'_\qm$ acts by multiplication, so that the measure will localise on the irreducible quantum representations $\hat{X}_{\sfG'_\qm} \subset X_{\sfG'_\qm}$ of $\sfG'_\qm$.
	
	\item
	In the situation where $\sfG'_\qm \cong \sfG' \times \sfU(1)$, we can therefore regard the measure as living on $X_{\sfG'}$ by Lemma~\ref{st:QM_extensions_and_Pontryagin_duals}.
	\qen
\end{myenumerate}
\end{remark}

A projection-valued measure $P$ on the topological space $\hat
X_{\sfG_\qm'}$ induces a sheaf $\CS$ of Hilbert spaces on $\hat
X_{\sfG_\qm'}$ via $\CO \mapsto \CS(\CO) = P(\CO) \scH$ for $\CO$ an
open subset of $\hat X_{\sfG_\qm'}$.
For an inclusion $\imath_{\CV, \CO} \colon \CO \hookrightarrow \CV$ of
open subsets of $\hat X_{\sfG_\qm'}$, the restriction map of the sheaf $\CS$ is the projection $\CS(\imath_{\CV, \CO}) = \pr_{P(\CO)\scH}$, since $P(\CO)\scH$ sits inside $P(\CV)\scH$ as a closed subspace.
The sheaf $\rmL^\infty$ of essentially bounded functions on $\hat{X}_{\sfG'_\qm}$ acts on $\CS$ via\footnote{A function $\widehat{f}$ on $\CO \subset \hat{X}_{\sfG'_\qm}$ is \emph{essentially bounded} if $\int_\CO\, | \widehat{f}\, (\lambda)|^2\ \< \psi,\,\dd P(\lambda) \psi \>_\scH$ exists for all $\psi \in \scH$.}
\begin{equation}
	\rmL^\infty(\CO) \times \CS(\CO) \longrightarrow \CS(\CO) \ , \quad
	(\widehat{f}, \psi) \longmapsto \int_\CO\, \widehat{f}\,(\lambda)\ \dd P(\lambda) \psi\ .
\end{equation}
This suggests that $\CS$ might actually be of the form of a sheaf of sections of a Hilbert bundle on $\hat{X}_{\sfG'_\qm}$.
However, in general we can at best only state this as a requirement
rather than deriving it as a consequence, which is indeed what we
shall do in Section~\ref{sect:Band_insulators} below.

To treat Galilean and
quantum symmetries simultaneously, consider the situation of two short exact
sequences in $\TopGrp_{\ZZ^2}$ given by
\begin{equation}
	\xymatrixrowsep{0.5cm}
	\xymatrix{
	 & & 1 \ar@{->}[d] & \\
	 & & \sfU(1) \ar@{->}[d] & \\
	 & & \sfG_\qm \ar@{->}[d] & \\
	1 \ar@{->}[r] & \sfG' \ar@{->}[r] & \sfG \ar@{->}[r] \ar@{->}[d] & \sfG'' \ar@{->}[r] & 1\\
	 & & 1 &
	}
\end{equation}
where the horizontal sequence is an inclusion of a normal abelian
subgroup and the vertical sequence is a quantum extension of $\sfG$.
This induces a quantum extension $\sfG'_\qm \rightarrow \sfG'$ of $\sfG'$ via pullback.
As pointed out in Remark~\ref{rmk:G'_qu has trivial gradings}, we restrict to the cases where $\sfG'_\qm$ has trivial $\ZZ$-bigrading.
In particular, it is a topological abelian group.

\begin{lemma}
\label{st:G'' acts on XG'}
There is a right action of $\sfG''$ on $\hat{X}_{\sfG'_\qm}$ induced
by the adjoint action of $\sfG_\qm$ on $\sfG'_\qm$, and a right action of
$\sfG''$ on $X_{\sfG'}$ induced by the adjoint action of $\sfG$ on
$\sfG'$. Therefore any isomorphism $\sfG'_\qm \cong \sfG' \times \sfU(1)$ induces an isomorphism $\hat{X}_{\sfG'_\qm} \cong X_{\sfG'}$ of topological spaces with a right $\sfG''$-action.
\end{lemma}

\begin{proof}
Since the extension $\sfG'_\qm$ is constructed as the pullback of $\sfG_\qm \rightarrow \sfG$ along the inclusion of $\sfG'$ into $\sfG$, we have $\sfG'_\qm = \ker(\sfG_\qm \rightarrow \sfG \rightarrow \sfG''\,)$, so that $\sfG'_\qm \subset \sfG_\qm$ is a normal subgroup.
Let $\hat{\lambda} \in \hat{X}_{\sfG'_\qm}$, $\hat{\xi} \in \sfG'_\qm$, and $g'' \in \sfG''$.
Choose a preimage $g \in \sfG$ of $g''$, and a further preimage $\hat{g} \in \sfG_\qm$ of $g$.
Set
\begin{equation}
\label{eq:Brillouin_action}
	R_{g''} \hat{\lambda} = \hat{\lambda}^{\phi_\sfG(g''\,)} \circ
        \alpha(\hat{g})\ ,
\end{equation}
where $\alpha(\hat{g})$ is the inner automorphism of $\sfG_\qm$ generated by $\hat{g}$.
Explicitly
\begin{equation}
\label{eq:Brillouin_right_action_explicit}
	\big( R_{g''} \hat{\lambda} \big) (\hat{\xi}\,) =
        \hat{\lambda} \big( \hat{g}\, \hat{\xi}\, \hat{g}^{-1}
        \big)^{\phi_\sfG(g''\,)}\ ,
\end{equation}
which is well-defined because $\sfG'_\qm $ is a normal subgroup of $\sfG_\qm$.
Any other lift of $g''$ to $\sfG_\qm$ will differ from $\hat{g}$ by multiplication from the right by an element of $\sfG'_\qm$, which will cancel in the argument of $\hat{\lambda}$ because $\sfG'_\qm$ is abelian.
The action of $\sfG''$ on $X_{\sfG'}$ is obtained in a completely
analogous way, and the statement about isomorphism then follows from Lemma~\ref{st:QM_extensions_and_Pontryagin_duals}.
\end{proof}

The sign in the exponent of $R_{g''} \hat{\lambda}$ cannot be chosen freely:
using a minus sign instead would yield $R_{g''} \hat{\lambda}$ which
satisfies $\big(R_{g''} \hat{\lambda} \big)(\hat{\xi}\, z) = z^{-1}\cdot \big(R_{g''}
\hat{\lambda} \big)(\hat{\xi}\,)$ for $z \in \sfU(1)$, and hence would take us out
of the space of quantum representations of $\sfG'$.

Let $\hat{X}_{\sfG'_\qm}\dslash \sfG''$ denote the action groupoid of the right $\sfG''$-action on $\hat{X}_{\sfG'_\qm}$ constructed above, and let $\hat{\scG}_\bullet = \sfN(\hat{X}_{\sfG'_\qm}\dslash \sfG''\,)_\bullet$ denote its simplicial nerve.
Explicitly
\begin{equation}
	\hat{\scG}_n = \hat{X}_{\sfG'_\qm} \times {\sfG''}{}^n
\end{equation}
for all $n\in\NN_0$, with face maps $d_i \colon \hat{\scG}_n \to \hat{\scG}_{n-1}$ for $0 \leq i \leq n$ given by
\begin{equation}
	d_i (\hat{\lambda}, g''_1, \ldots, g''_n\,) =%
	\begin{cases}
		\ (R_{g''_1} \hat{\lambda}, g''_2, \ldots, g''_n\,) \ , &
                i = 0\ ,\\
		\ (\hat{\lambda}, g''_1, \ldots, g''_{i}\, g''_{i+1},
                \ldots, g_n''\,) \ , & 0 < i < n\ , \\
		\ (\hat{\lambda}, g''_1, \ldots, g''_{n-1}) \ , & i = n\
                ,
	\end{cases}
\end{equation}
while the degeneracy maps $s_i \colon \hat{\scG}_{n-1} \to
\hat{\scG}_n$, for $0 \leq i \leq n{-}1$, insert the group unit at the
$i$-th position; we write $\pr_i:\hat\scG_n\to\hat\scG_1$ for the
projection onto the $i$-th factor of $\sfG''{}^n$ for $1\leq i\leq n$.
Analogously we define $\scG_\bullet = \sfN(X_{\sfG'}\dslash
\sfG''\,)_\bullet$ to be the nerve of the action groupoid of the right action of $\sfG''$ on $X_{\sfG'}$.
The isomorphism $\hat{X}_{\sfG'_\qm} \cong X_{\sfG'}$ from
Lemma~\ref{st:QM_extensions_and_Pontryagin_duals} together with the
statement on equivariance from Lemma~\ref{st:G'' acts on XG'} induce
an isomorphism of simplicial spaces $\hat{\scG}_\bullet \cong \scG_\bullet$, provided that we have fixed an isomorphism $\sfG'_\qm \cong \sfG' \times \sfU(1)$.

For $\epsilon = \pm\, 1$ and $E$ a complex vector bundle on a topological space $X$ we set
\begin{equation}
	E^{[\epsilon]} =%
	\begin{cases}
		\ E \ , & \epsilon = 1\ ,\\
		\ \overline{E} \ , & \epsilon = -1\ .
	\end{cases}
\end{equation}
Similarly, for a morphism $\alpha \colon E \to F$ of complex vector bundles we set $\alpha^{[1]} = \alpha$ and $\alpha^{[-1]} = \overline{\alpha} \colon \overline{E} \to \overline{F}$.

\begin{proposition}
\label{st:Brillouin_gerbe}
There is a Hermitean line bundle $\hat{L} \rightarrow \hat{\scG}_1$, together with a unitary isomorphism
\begin{equation}
	\hat{\mu} \colon d_2^*\hat{L}^{[\phi_\sfG \circ \pr_2]} \otimes d_0^*\hat{L} \arisom d_1^*\hat{L}
\end{equation}
over $\hat{\scG}_2$ satisfying an associativity condition over $\hat{\scG}_3$.
Thus the group action of $\sfG''$ induces a $\ZZ$-graded simplicial Hermitean line bundle, or a $\ZZ$-graded central groupoid extension, on $\hat{\scG}_\bullet$. \footnote{This is also called a bundle gerbe
  over the groupoid $\hat{X}_{\sfG'_\qm}\dslash \sfG''$, see Section~\ref{sect:BGrbs}.}
\end{proposition}

\begin{proof}
Consider the principal $\sfG'_\qm$-bundle $\hat{X}_{\sfG'_\qm} \times \sfG_\qm \rightarrow \hat{X}_{\sfG'_\qm} \times \sfG''$.
Define
\begin{equation}
\label{eq:Brillouin_gerbe_qm}
	\hat{L} = \big( \hat{X}_{\sfG'_\qm} \times \sfG_\qm \big)
        \times \FC \, \big/ \,
	 \big( (\hat{\lambda},\, \hat{g}),\, z \big)
	\sim \big( (\hat{\lambda},\, \hat{g}\, \hat{\xi}\,),\,
        \hat{\lambda} \big( \hat{g}\, \hat{\xi}\, \hat{g}^{-1}
        \big)^{-\phi_\sfG(g''\,)}\, z \big) \ ,
\end{equation}
where $\hat{g} \in \sfG_\qm$ is an element in the fibre of $\sfG_\qm \rightarrow \sfG''$ over $g''$, $\hat{\lambda} \in \hat{X}_{\sfG'_\qm}$ is a quantum representation, $\hat{\xi}$ is an element of $\sfG'_\qm$, and $z \in \FC$.
Since the action on $\FC$ is purely by elements of $\sfU(1)$, the bundle is Hermitean.
Next we define
\begin{equation}
	\hat{\mu}_{|(\hat{\lambda},\, g'',\, h''\,)} \big( \big[ (\hat{\lambda},\, \hat{g}),\, z \big] \otimes \big[ (R_{g''} \hat{\lambda},\, \hat{h}),\, z'\, \big] \big)
	= \big[ (\hat{\lambda},\, \hat{g}\, \hat{h}),\,
        \theta_{h''}(z)\,z'\, \big]\ ,
\end{equation}
where $h'' \in \sfG''$ with lift $\hat{h} \in \sfG_\qm$ and we set $\theta_{h''}(z) \coloneqq z$ if $\phi_{\sfG}(h''\,) = 1$ and $\theta_{h''}(z) \coloneqq \overline{z}$ if $\phi_{\sfG}(h''\,) = -1$.
For any $\hat\xi,\hat\zeta\in\sfG_\qm$ we check compatibility with the equivalence relation by calculating
\addtocounter{equation}{1}
\begin{align*}
	&\hat{\mu}_{|(\hat{\lambda},\, g'',\, h''\,)} \big( \big[ (\hat{\lambda},\, \hat{g}\, \hat{\xi}\,),\, \hat{\lambda} \big( \hat{g}\, \hat{\xi}\, \hat{g}^{-1} \big)^{-\phi_\sfG(g''\,)}\, z \big]
	\otimes \big[ (R_{g''} \hat{\lambda},\, \hat{h}\, \hat{\zeta}\,),\, \big( (R_{g''} \hat{\lambda}) \big( \hat{h}\, \hat{\zeta}\, \hat{h}^{-1} \big) \big)^{-\phi_\sfG(h''\,)}\, z'\, \big] \big)\\*
	&\qquad\qquad= \big[ (\hat{\lambda},\, \hat{g}\, \hat{\xi}\,
          \hat{h}\, \hat{\zeta}\,),\, \hat{\lambda} \big( \hat{g}\,
          \hat{\xi}\, \hat{g}^{-1} \big)^{-\phi_{\sfG}(g'' \, h''\,)}\,
	\theta_{h''}(z)\, \hat{\lambda} \big( \hat{g}\, \hat{h}\, \hat{\zeta}\, \hat{h}^{-1}\, \hat{g}^{-1} \big)^{-\phi_\sfG(g''\, h''\,)}\, z' \,\big]\\[4pt]
	&\qquad\qquad= \big[ (\hat{\lambda},\, \hat{g}\, \hat{h}),\, \hat{\lambda}\big( \hat{g}\, \hat{h}\, ( \hat{h}^{-1}\, \hat{\xi}\, \hat{h}\, \hat{\zeta}\,)\, \hat{h}^{-1}\, \hat{g}^{-1} \big)^{\phi_{\sfG}(g''\, h''\,)}\
	\hat{\lambda} \big( \hat{g}\, \hat{\xi}\, \hat{g}^{-1} \big)^{-\phi_{\sfG}(g''\, h''\,)} \tag{\theequation}\\*
	&\hspace{7cm} \times\, \hat{\lambda} \big( \hat{g}\, \hat{h}\,
          \hat{\zeta}\, \hat{h}^{-1}\, \hat{g}^{-1}
          \big)^{-\phi_\sfG(g'' \, h''\,)}\, \theta_{h''}(z)\, z'\, \big]\\[4pt]
	&\qquad\qquad= \big[ (\hat{\lambda},\, \hat{g}\, \hat{h}),\, \theta_{h''}(z)\,z'\, \big]\\[4pt]
	&\qquad\qquad= \hat{\mu}_{|(\hat{\lambda},\, g'',\, h''\,)}
          \big( \big[ (\hat{\lambda},\, \hat{g}),\, z \big] \otimes
          \big[ (R_{g''} \hat{\lambda},\, \hat{h}),\, z'\, \big]
          \big)\ .
\end{align*}
Associativity of $\hat{\mu}$ follows immediately from associativity of the multiplication in the groups $\sfG''$ and $\sfG_\qm$.
\end{proof}

It is shown in~\cite[Theorem~9.42]{Freed-Moore--TEM} that the sheaf
$\CS$ of
Hilbert spaces induced on $\hat{X}_{\sfG'_\qm}$ by a decomposition of
a quantum representation $\rho_\qm:\sfG_\qm \rightarrow
\Aut_\qm(\scH,H)$ automatically carries an $\hat{L}$-\emph{twisted}
$\sfG''$-equivariant structure with respect to the action of $\sfG''$
on $\hat{X}_{\sfG'_\qm}$ from above and the line bundle
$\hat{L}$. When we use the isomorphism $\hat{X}_{\sfG'_\qm} \cong
X_{\sfG'}$ induced by a splitting $\sfG'_\qm \cong \sfG' \times
\sfU(1)$, we denote the pullback of $\CS$ to $X_{\sfG'}$ by $\CS$ as well.
We may further pull back $\hat{L}$ to $X_{\sfG'}$ in this case.

\section{Gapped topological phases}
\label{sect:Gapped topological phases}

\subsection{Band insulators and their Bloch bundles}
\label{sect:Band_insulators}

A crystalline condensed matter system consists of a space lattice of
atoms which generate a periodic Coulomb potential that electrons in the material are subjected to.
Thus in order to describe such systems, we use a crystal $(C,\Pi)$ in
a Galilean spacetime with a time direction $(M,\Gamma,U)$ (see Definition~\ref{def:crystal}).
Such a configuration admits Galilean symmetry groups $(\sfH,\delta,i)$
preserving the crystal, i.e. such that $\delta \colon \sfH \rightarrow
\Aut(M,\Gamma,U)$ takes values in the group $\sfG(C) \subset
\Aut(M,\Gamma,U)$ of Galilean symmetries that preserve the crystal.
We assume that $\Pi \subset \delta(\sfH)$, i.e. the action of $\sfH$ contains the full lattice of space translations preserving the crystal.
Recall that $i \colon (U \oplus V) \cap \delta(\sfH) \hookrightarrow \sfH$ is an inclusion of a normal subgroup.
Thus, in particular, we obtain an inclusion of $U \cap \delta(\sfH) \hookrightarrow \sfH$ as a normal subgroup.

The Hilbert space of the corresponding quantum system is associated
only to the spatial part of the background Galilean spacetime with
time direction $(M, \Gamma, U)$; states can be regarded as
wavefunctions on the orbit space $M/U$.
We will make this more explicit in
Section~\ref{sect:Bloch_and_Berry_explicit} below.
Therefore we define the group of symmetries acting on the quantum Hilbert space $\sfG \coloneqq \sfH/i(U \cap \delta(\sfH))$ as the quotient of $\sfH$ by time translations.
The homomorphism $\delta$ induces a homomorphism on the quotient groups, $\gamma \colon \sfG \rightarrow \sfG(C)/U \subset \Aut(M,\Gamma,U)/U$, where $\sfG(C)/U$ is the crystallographic group of the crystal $(C,\Pi)$ from Definition~\ref{def:Crystal_symmetries}.
Similarly, the inclusion $i$ induces an inclusion $j \colon V \cap \gamma(\sfG) \hookrightarrow \sfG$ as a normal subgroup.
By assumption $\Pi \subset V \cap \gamma(\sfG)$, and we assume that $j_{|\Pi}: \Pi \hookrightarrow \sfG$ is an inclusion of a normal subgroup as well.
Defining $\sfG'' = \sfG/\Pi$, there is a commutative diagram of short exact sequences
\begin{equation}
\label{diag:Spatial_crystal_symmetries}
	\xymatrix{
		1 \ar@{->}[r] & \Pi \ar@{->}[r]^-{j} \ar@{=}[d] & \sfG \ar@{->}[r]^-{\pi} \ar@{->}[d]^-{\gamma} & \sfG'' \ar@{->}[r] \ar@{->}[d]^-{\gamma''} & 1\\
		1 \ar@{->}[r] & \Pi \ar@{->}[r] & \sfG(C)/U \ar@{->}[r] & \widehat{\sfP}(C) \ar@{->}[r] & 1
	}
\end{equation}
where $\widehat{\sfP}(C)$ is the magnetic point group of $(C,\Pi)$ from Definition~\ref{def:Crystal_symmetries}.

The electrons in the crystal are described by wavefunctions, or more generally by states in a Hilbert space $\scH$ with a Hamiltonian $H$.
The group $\sfG$ should act on this quantum system as a group of projective quantum symmetries, thus inducing a quantum extension $(\sfG_\qm,\iota_\sfG, q_\sfG)$.
We assume that the quantum extension of the subgroup $\Pi \subset \sfG$ is trivialised, i.e. that crystal translations act on $\scH$ in a canonical way.
This will be satisfied whenever we consider the Hilbert space to consist of wavefunctions on the spatial leaves of $M$.

More precisely, consider a gapped quantum system $(\scH,H)$ with symmetry class $(\sfG_\qm, \iota_\sfG, q_\sfG)$ (see Definition~\ref{def:Gapped_system_with_symmetry_class}) such that $\Pi$ is contained in $\sfG$ as a normal subgroup $\Pi \subset \sfG$ and the quantum extension $(\sfG_\qm, \iota_\sfG, q_\sfG)$ of $(\sfG, \phi_\sfG, \sfc_\sfG)$ is trivial over $\Pi \subset \sfG$.
The crucial consequence in this case is that the representation $\rho_\qm \colon \sfG_\qm \rightarrow \Aut_\qm(\scH,H)$ induces a representation $\rho = \rho_{\qm|\Pi} \colon \Pi \rightarrow \Aut_\qm(\scH,H)$ of $\Pi$ on the Hilbert space $\scH$.
By construction the grading of translations is trivial, so that $\rho(\xi)$ acts unitarily and commutes with the Hamiltonian for all $\xi \in \Pi$.
That is, $\rho$ is a continuous unitary representation of the abelian group $\Pi$ preserving the Hamiltonian.

With these assumptions, there is a commutative diagram of topological groups with exact rows and columns given by
\begin{equation}
\label{diag:Pi_and_qm_extensions}
\begin{tikzcd}[row sep=0.5cm, column sep=0.75cm]
	& 1 \ar[d] & 1 \ar[d]  & 1 \ar[d]  & \\
	1 \ar[r] & 1 \ar[r] \ar[d] & \sfU(1) \ar[r] \ar[d] & \sfU(1) \ar[r] \ar[d] & 1\\
	1 \ar[r] & \Pi \ar[r] \ar[d] & \sfG_\qm \ar[r] \ar[d] & \sfG''_\qm \ar[r] \ar[d] & 1\\
	1 \ar[r] & \Pi \ar[r] \ar[d] & \sfG \ar[r] \ar[d] & \sfG'' \ar[r] \ar[d] & 1\\
	 & 1 & 1 & 1  & %
\end{tikzcd}
\end{equation}
The group $\sfG_\qm$ is represented on $\scH$ by $\rho_\qm$.
The inclusion $\Pi \hookrightarrow \sfG_\qm$ stems from the assumed choice of trivialisation of the restriction of the $\sfU(1)$-bundle $\sfG_\qm \rightarrow \sfG$ over $\Pi \subset \sfG$.
The group $\sfG''_\qm$ is defined by this diagram.

We now specialise the discussion of Section~\ref{sect:Decomp_reps} to the case of a lattice system with $\sfG' = \Pi$ and with $\sfG'_\qm \cong \Pi \times \sfU(1)$ trivialised. We work with the following assumptions, gleaned from~\cite[Hypothesis 10.9]{Freed-Moore--TEM}:
\begin{itemize}
	\item[(A1)]
	The group $\sfG''$ is a compact Lie group.
	
	\item[(A2)]
	There exists a $\sfG_\qm$-equivariant smooth Hilbert bundle $E \rightarrow X_\Pi$ and a $\sfG_\qm$-equivariant unitary isomorphism
	\begin{equation}
	\label{eq:Brillouin_decomp_of_H}
		\rmU \colon \scH \arisom \rmL^2(X_\Pi, E)
	\end{equation}
	from $\scH$ to square-integrable sections of $E$ with respect to the Haar measure on $X_\Pi$, such that 
	\begin{equation}
		\big( \rmU \circ \rho(\xi)\, \psi \big) (\lambda) = \lambda(\xi)\, \rmU\psi (\lambda)
	\end{equation}
for all $\lambda \in X_\Pi$, $\xi \in \Pi$ and $\psi \in \scH$.
	
	\item[(A3)]
	Under the isomorphism~\eqref{eq:Brillouin_decomp_of_H} of Hilbert spaces the Hamiltonian $H$ induces a continuous family of selfadjoint (densely defined) operators $\{H_{|\lambda}\}_{\lambda \in X_\Pi}$ on the fibres of $E$.
	The operators $H_{|\lambda}$ are called \emph{Bloch Hamiltonians}.
	
	\item[(A4)]
	Each Bloch Hamiltonian is automatically still gapped, whence, as illustrated in Section~\ref{sect:Reps_and_gapped_Hamiltonians}, it splits its respective fibre into two orthogonal subspaces.
	We assume that in each fibre of $E$ the negative spectral projection of its Hamiltonian has finite rank.
	In~\cite{Freed-Moore--TEM} it is proven that this assumption is sufficient to show that the positive and negative spectral projections on the fibres of $E$ then split $E$ into two Hilbert subbundles, $E = E^+ \oplus E^-$, where $E^-$ has finite rank.
	
\end{itemize}

With these assumptions in hand, we can now borrow~\cite[Definition 10.7]{Freed-Moore--TEM} to present the key definitions of this paper.

\begin{definition}[Band insulator, Bloch bundle]
\label{def:Band_Insulator_Bloch_bundle}
Let $(M,\Gamma,U)$ be a Galilean spacetime with a time direction.
A \emph{band insulator} in $(M,\Gamma,U)$ consists of the following data:
\begin{myenumerate}
	\item A crystal $(C,\Pi)$ in $(M,\Gamma,U)$.
	
	\item A Galilean symmetry group $(\sfH,\delta,i)$ preserving
          the crystal, i.e.~the homomorphism $\delta \colon \sfH
          \rightarrow \Aut(M,\Gamma,U)$ takes values in $\sfG(C)
          \subset \Aut(M,\Gamma,U)$, which defines a space symmetry group $(\sfG,\gamma,j)$ with quotient group $\sfG'' = \sfG/\Pi$.
	
	\item A quantum extension $(\sfG''_\qm,\iota_{\sfG''},
          q_{\sfG''})$ of $(\sfG'', \phi_{\sfG''}, \sfc_{\sfG''})$
          which pulls back to a quantum extension $(\sfG_\qm,
          \iota_\sfG, q_\sfG)$ of $\sfG$ that is trivial over the
          translation lattice $\Pi \subset \sfG$.
	
	\item A gapped quantum system $(\scH,H)$ with symmetry class $(\sfG_\qm,\iota_\sfG, q_\sfG)$.
\end{myenumerate}
We require this data to satisfy the assumptions (A1)--(A4) above.\\
A band insulator is of \emph{type I} if $E^+$ has infinite rank, and of \emph{type F} if $E^+$ has finite rank.\\
The finite-rank Hilbert bundle $E^- \rightarrow X_\Pi$ is called the \emph{Bloch bundle} of the band insulator.
\end{definition}

The compact topological abelian group $X_\Pi$ is usually called the \emph{Brillouin zone} or the \emph{Brillouin torus} of the crystal $(C,\Pi)$.
Its elements are interpreted as periodic momenta of electrons in the
crystal, where momenta $k \in V^*$ are equivalent if they differ by
translations by some $\xi' \in \Pi^*$.\footnote{The group $X_\Pi$ is
  canonically isomorphic to the quotient $V^*/\Pi^*$ of the dual
  vector space $V^*$ of $V$ by the dual lattice $\Pi^*$ of $\Pi$.}
The spectrum of $H$ consists of at least two connected subsets of $\FR$.
Each connected component of $\sigma(H)$ is called an \emph{energy band} of $H$.
The bands below the gap, or the {Fermi energy}
$\varepsilon_{\scH,H} = 0 \notin \sigma(H)$ (see
Section~\ref{sect:Reps_and_gapped_Hamiltonians}), are called
\emph{valence bands}, whereas the bands above the Fermi energy are called \emph{conduction bands}.

\subsection{Topological phases and K-theory}
\label{sect:Brillouin_gerbe_and_KT_classification}

We will now sketch the idea behind classifying topological phases of
quantum matter via K-theory, as pioneered
in~\cite{Kitaev:Periodic_table} and further elaborated on by~\cite{Freed-Moore--TEM}.
To a given $\ZZ$-bigraded topological group $(\sfG, \phi_\sfG,
\sfc_\sfG) \in \TopGrp_{\ZZ^2}$ and a quantum extension $(\sfG_\qm,
\iota_\sfG, q_\sfG)$ we define a category $\GapSys(\sfG_\qm,
\iota_\sfG, q_\sfG)$, whose objects are gapped quantum systems with
symmetry class $(\sfG_\qm, \iota_\sfG, q_\sfG)$ (see
Definition~\ref{def:Gapped_system_with_symmetry_class}), and whose
morphisms are generated by unitary isomorphisms of Hilbert spaces
which are compatible with the respective Hamiltonians and intertwine
the $\sfG_\qm$-actions, together with homotopies, i.e. continuous
deformations, of gapped quantum systems with symmetry class $(\sfG_\qm,
\iota_\sfG, q_\sfG)$ as described in Section~\ref{sect:Reps_and_gapped_Hamiltonians}.
This category is a symmetric monoidal groupoid, with monoidal
structure given by the direct sum of Hilbert spaces with gapped
Hamiltonians and with $\sfG_\qm$-actions. The zero Hilbert space provides the unit object.

For a category $\scC$, we let $\pi_0 \scC$ denote the collection of isomorphism classes of objects of its groupoid completion.%
\footnote{The elements of $\pi_0 \scC$ are precisely the path
  connected components of the classifying space of $\scC$, whence the notation $\pi_0$.
The categories $\scC$ under consideration here are already groupoids,
so that $\pi_0 \scC$ is simply the collection of isomorphism classes
of objects of $\scC$.}
Borrowing from~\cite[Definition~5.1]{Freed-Moore--TEM}, we can now
formulate another central concept to this paper.

\begin{definition}[Topological phases]
Let $(\sfG,\phi_\sfG,\sfc_\sfG)$ be a $\ZZ$-bigraded topological group
with quantum extension $(\sfG_\qm,\iota_\sfG,q_\sfG)$. A \emph{topological phase} with symmetry class $(\sfG_\qm, \iota_\sfG, q_\sfG)$ is an isomorphism class of objects in $\GapSys(\sfG_\qm, \iota_\sfG, q_\sfG)$.
We denote the commutative monoid of topological phases with symmetry class $(\sfG_\qm, \iota_\sfG, q_\sfG)$ by
\begin{equation}
	\TP(\sfG_\qm, \iota_\sfG, q_\sfG) = \pi_0\, \GapSys(\sfG_\qm,
        \iota_\sfG, q_\sfG)\ ,
\end{equation}
and its Grothendieck completion by
\begin{equation}
	\RTP(\sfG_\qm, \iota_\sfG, q_\sfG) = \mathrm{Gr} \big( \pi_0\,
        \GapSys(\sfG_\qm, \iota_\sfG, q_\sfG)\big)\ .
\end{equation}
The elements of the abelian group $\RTP(\sfG_\qm, \iota_\sfG, q_\sfG)$ are called \emph{reduced topological phases} with symmetry class $(\sfG_\qm, \iota_\sfG, q_\sfG)$.
\end{definition}

Note that our construction of K-theory groups via Grothendieck completion agrees with the definition in~\cite{Freed-Moore--TEM} via $\ZZ$-graded vector bundles; an isomorphism is given by sending a $\ZZ$-graded bundle $E^+ \oplus E^-$ to the formal difference $E^+ - E^-$ (see also~\cite[p.~57]{Freed-Moore--TEM}).
Denoting by $\TP_\rmF$ and $\TP_\rmI$ the submonoids of topological
phases which have a representative given by a band insulator of type F
or type I as in Definition~\ref{def:Band_Insulator_Bloch_bundle}, respectively, we can now use the Bloch bundle with the fibrewise Bloch Hamiltonians to derive a classification of such phases.

For band insulators, translating the $\hat L$-twisted $\sfG''$-equivariant structure from
Section~\ref{sect:Decomp_reps} to the Hilbert bundle $E$ corresponding
to the sheaf of Hilbert spaces $\CS$ under the assumptions (A1)--(A4) from Section~\ref{sect:Band_insulators}, this is equivalent to the existence of an isomorphism of Hilbert bundles
\begin{equation}
	\hat{\alpha} \colon (\hat{L} \otimes d_0^* E)^{[\phi_{\sfG}]} \arisom d_1^*E
\end{equation}
over $\hat{\scG}_1$, which is compatible with the action of $\hat{\mu}$ in the sense that
\begin{equation}
	\hat{\alpha}_{|(\hat{\lambda}, g''\,)} \circ \big( \One
        \otimes \hat{\alpha}_{|(R_{g''} \hat{\lambda},\, h''\,)}
        \big) = \hat{\alpha}_{|(\hat{\lambda},\, g''\, h''\,)}\, \circ \big( \hat{\mu}_{|(\hat{\lambda},\, g'',\, h''\,)} \otimes \One \big)
\end{equation}
for all $g'', h'' \in \sfG''$ and $\hat{\lambda} \in
\hat{X}_{\sfG_\qm'}$.\footnote{Thoughout we freely use the isomorphism
from Lemma~\ref{st:QM_extensions_and_Pontryagin_duals}.}
We will see an explicit example of this in Section~\ref{sect:Bloch_and_Berry_explicit} below.
The bundle $E$ carries a continuous family of gapped Hamiltonians $\{H_{|\lambda}\}_{\lambda \in X_\Pi}$, thus inducing a $\ZZ$-grading on the fibres.

Morphisms of $\hat{L}$-twisted $\sfG''$-equivariant $\ZZ$-graded Hilbert bundles are linear maps $\psi: (E,\hat{\alpha}) \rightarrow (F,\hat{\beta})$ which satisfy $\hat{\beta} \circ (\One \otimes \psi) = \psi \circ \hat{\alpha}$.
This defines a category of $\hat{L}$-twisted $\sfG''$-equivariant $\ZZ$-graded Hilbert bundles.
We denote its subcategory of $\hat{L}$-twisted $\sfG''$-equivariant
$\ZZ$-graded Hermitean vector bundles, i.e. $\hat{L}$-twisted $\sfG''$-equivariant $\ZZ$-graded Hilbert bundles of finite rank, by $\HVBdl^{\hat{L},\sfG''}_{\ZZ}(X_{\Pi})$.
One can show \cite[Section 10]{Freed-Moore--TEM} that homotopies of
gapped systems together with isomorphisms of quantum representations of $\sfG$ induce isomorphisms of $\hat{L}$-twisted $\sfG''$-equivariant $\ZZ$-graded bundles $E \rightarrow X_\Pi$.

For a category $\scC$ denote the groupoid obtained from $\scC$ by
forgetting all noninvertible morphisms in $\scC$ by $\scC_\sim$. We
can then state~\cite[Theorem 10.11]{Freed-Moore--TEM} as

\begin{theorem}
Sending a band insulator of type F to the isomorphism class of $E^+ \oplus E^-$ induces an isomorphism of commutative monoids
\begin{equation}
\label{eq:TP classification}
	\TP_\rmF(\sfG_\qm,\iota_\sfG, q_\sfG) \arisom \pi_0 \, \HVBdl^{\hat{L},\sfG''}_{\ZZ\, \sim}(X_\Pi) 
\end{equation}
from topological phases of type F to isomorphism classes of
$\hat{L}$-twisted $\sfG''$-equivariant $\ZZ$-graded Hermitean vector
bundles on the Brillouin torus $X_\Pi$.
\end{theorem}

This map is compatible with Grothendieck completion on both sides (the
Grothendieck completion is functorial) and as such induces an
isomorphism of abelian groups; the isomorphisms used here do not
generally respect the Hermitean metrics.
The Grothendieck completion of the category of $\hat{L}$-twisted
$\sfG''$-equivariant $\ZZ$-graded Hermitean vector bundles on $X_\Pi$ is called
the \emph{$\hat{L}$-twisted $\sfG''$-equivariant K-theory group} of
  $X_\Pi$ and is denoted as
\begin{equation}
	\rmK^{\hat{L},\sfG''}(X_\Pi) \coloneqq \mathrm{Gr} \big( \pi_0
        \, \HVBdl^{\hat{L},\sfG''}_{\ZZ\, \sim}(X_\Pi) \big)\ .
\end{equation}
Then~\cite[Theorem 10.15]{Freed-Moore--TEM} reads as

\begin{theorem}
\label{st:KT_classification}
The isomorphism~\eqref{eq:TP classification} of commutative monoids induces an isomorphism
\begin{equation}
	\RTP_\rmF(\sfG_\qm,\iota_\sfG, q_\sfG) \arisom \rmK^{\hat{L},\sfG''}(X_\Pi)
\end{equation}
between the abelian group of reduced topological phases admitting a
realisation as a type F band insulator and the $\hat L$-twisted
$\sfG''$-equivariant K-theory group of the Brillouin torus $X_\Pi$.
\end{theorem}

In the case of type I band insulators the situation is more
complicated because the bundle $E^+$ has infinite rank and therefore does not fit into the framework of topological K-theory.
The map from band insulators to twisted equivariant K-theory is then
given by sending a band insulator to the K-theory class of its Bloch
bundle $E^-$, and an analogous version of
Theorem~\ref{st:KT_classification} is formulated
in~\cite{Freed-Moore--TEM}, but now methods for treating infinite-rank bundles have to be invoked from~\cite{FHT:Loop_groups_and_twisted_KT_I}.

\begin{remark}
It is a non-trivial fact that every class in the twisted K-theory groups we consider here can indeed be represented by an $\hat{L}$-twisted $\sfG''$-equivariant $\ZZ$-graded Hermitean vector bundle.
It is proven in~\cite[Remark~7.37, Appendix~E]{Freed-Moore--TEM}, and the proof relies on the specific construction of $\sfG''$ and $X_\Pi$ in the present context.

For generic topological spaces with a continuous group action, it is generally not true that twisted equivariant K-theory is represented by twisted equivariant $\ZZ$-graded vector bundles.
Instead, these K-theory classes can always be constructed as homotopy classes of sections of certain bundles of Fredholm operators -- see in particular~\cite{FHT:Loop_groups_and_twisted_KT_I}.
Some specific cases where twisted equivariant K-theory is representable by bundles are known (in addition to Theorem~\ref{st:KT_classification} above):
In~\cite[Remark~7.37, Appendix~E]{Freed-Moore--TEM}, it is pointed out that for $\sfG''$ a finite group acting on a nice topological space $X$, the twisted equivariant K-theory of $X$ is representable by bundles; here the problem is related to the presence of torsion classes (compare also~\cite{BCMMS,BSS--HGeoQuan}), which are absent in the case of finite $\sfG''$.
The proof relies on invoking the construction of a universal twisted Hilbert bundle in~\cite{FHT:Loop_groups_and_twisted_KT_I}.
In~\cite{Gomi:FM_K-Theory}, Gomi develops the twisted equivariant K-theory introduced in~\cite{Freed-Moore--TEM} into a fully general theory.
He proves that these K-theory groups can be represented by vector bundles whenever the grading $\sfc$ is trivial (cf.~\cite[Remark~7.36]{Freed-Moore--TEM}).
He points out that the reason why that fails in general is that the Atiyah-J\"anich-type argument, which one would like to use in order to relate the Fredholm and the vector bundle approaches to topological K-theory, fails for general twisted central extensions of groupoids.
Nevertheless, for our purposes it is sufficient to consider finite-rank vector bundles in order to describe twisted equivariant K-theory for topological phases by~\cite[Remark~7.37, Appendix~E]{Freed-Moore--TEM}.
\qen
\end{remark}

\subsection{Bloch wavefunctions and the Berry connection}
\label{sect:Bloch_and_Berry_explicit}

Having defined band insulators on a Galilean spacetime $(M,\Gamma)$ in a very general fashion, we now specialise to the case of band insulators whose Hilbert space is the space of wavefunctions on the spatial slices of $M$.
For this, recall from Definition~\ref{def:Band_Insulator_Bloch_bundle} that $(M,\Gamma)$ is assumed to have a time direction, i.e. the translations split into time and space translations as $W = U \oplus V$.
We can therefore decompose the Galilean spacetime into a direct product
\begin{equation}
	M \cong M_\rmt \times M_\rms \coloneqq M/V \times M/U \ .
\end{equation}
Recall from
  Definition~\ref{def:Galilean_structure_and_spacetime} that $U$ and
  $V$ are endowed with positive definite inner products, making
  $M_\rmt$ and $M_\rms$ into Euclidean spaces with Euclidean
  structures $\Gamma_\rmt$ and $\Gamma_\rms$ on $(M_\rmt,U)$ and $(M_{\rms},V)$.
The vector space $U$ acts transitively on $M_\rmt$, whereas $V$ acts transitively on $M_\rms$.
The isomorphism is given by sending $x \in M$ to its pair of equivalence classes $x \mapsto ([x]_\rmt, [x]_\rms)$.
Its inverse is constructed as follows.
Given $([x]_\rmt, [x]_\rms) \in M_\rmt \times M_\rms$, for any
auxiliary choice of origin $x_0 \in M$ we define $v_0([x]_\rms,x_0)
\in V$ to be the unique space translation such that 
\begin{equation}
[x_0 + v_0([x]_\rms,x_0)]_\rms = [x]_\rms\ .
\end{equation}
Similarly we define $u_0([x]_\rmt,x_0) \in U$.
Then the inverse isomorphism is given by
\begin{equation}
\label{eq:inverse_splitting_for_Galilean_spacetime}
	([x]_\rmt, [x]_\rms) \longmapsto x_0 + v_0([x]_\rms,x_0) +
        u_0([x]_\rmt,x_0)\ .
\end{equation}
If we shift $x_0 \mapsto x_0 + u' + v'$ for some $u' \in U$ and $v' \in V$, then $v_0([x]_\rms,x_0 + u' + v'\,) = v_0([x]_\rms,x_0) - v'$, and analogously for $u_0([x]_\rms,x_0)$, so that the map $M_\rmt \times M_\rms \arisom M$ defined by~\eqref{eq:inverse_splitting_for_Galilean_spacetime} is indeed well-defined independently of the choice of origin in $M$.

We now consider the quantum Hilbert space $\scH = \rmL^2(M_\rms,
V_\intern)$, where $V_\intern$ is some finite-dimensional Hilbert
space of internal degrees of freedom (for instance $V_\intern = \FC^2$
for particles of spin~$\frac{1}{2}$).
By construction, the spatial translations $V$ act canonically on $\scH$.
Let $C \subset M$ be a crystal.
The crystal translations $\Pi$ are embedded into $V$, and thus they act on $\scH$ as space translations.
We denote the space image of the crystal by $\overline{C} = C/U \subset M_\rms$.
Here it also becomes obvious why the seemingly cumbersome definition of the symmetry group $\sfG$ of a band insulator in Section~\ref{sect:Band_insulators} is actually a good choice: the symmetries $\sfG$ act as symmetries of $M_\rms$ which preserve the subset $\overline{C} \subset M_\rms$.
That is, the morphism $\gamma \colon \sfG \to
\Aut(M_\rms,\Gamma_\rms)$ takes values in the crystallographic group
$\sfG(\,\overline{C}\,) \coloneqq \sfG(C)/U$ of $C$.
The Hamiltonian $H$ can be taken to be a Schr\"odinger operator, subject to the condition that it be gapped and invariant under the space translations $\Pi$.

Writing $\sfG'_\qm = \sfG' \times\sfU(1)$ under the fixed isomorphism,
the bundle $E \rightarrow \hat{X}_{\sfG'_\qm}$ stems from the
decomposition of the quantum representation of $\sfG'$ on $\scH =
\rmL^2(M_\rms,V_\intern)$ into the $1$-dimensional irreducible quantum
representations $\hat{\lambda} \in \hat X_{\sfG_\qm'}$.
Using the isomorphism from Lemma~\ref{st:QM_extensions_and_Pontryagin_duals}, this is equivalent to decomposing $\scH$ into 1-dimensional representations of $\sfG'$.
This decomposition will in general not be as a direct sum, but only as a direct integral; for $\sfG' = \Pi$ the translation lattice this manifests itself in the fact that th so-called quasiperiodic functions $\psi \colon M_\rms \to \FC$, i.e.~those $\psi$ with
\begin{equation}
	\rho(\xi)(\psi) = \lambda(\xi)\, \psi \quad \forall\, \xi \in \Pi
\end{equation}
cannot be square-integrable.

\begin{remark}
Assuming for a moment that such functions $\psi$ were elements of $\scH$ after all, it would follow that the class in $\PP\scH$ of $\psi$, i.e.~the actual quantum state represented by $\psi$, is invariant under the action of $\sfG'$.
In particular, for crystal systems $\sfG' = \Pi$ is the group of
space lattice translations, and functions with
this quasiperiodic transformation behaviour are called \emph{Bloch
  wavefunctions}.
\qen
\end{remark}

The approach to the classification of topological phases via \emph{topological} K-theory rests on the observation that there exists a Hilbert bundle $E \to \hat{X}_{\sfG'_\qm}$ with a canonical unitary isomorphism
\smash{$\rmU^{-1} \colon \rmL^2(\hat{X}_{\sfG'_\qm},\, E) \arisom
\rmL^2(M_\rms,\, V_\intern) = \scH$}.
The bundle $E$ is constructed as follows:
first, we define a hermitean line bundle $\CL \to X_{\sfG'_\qm} \times M_\rms/\sfG'$, called the \emph{Poincar\'e line bundle}, via
\begin{equation}
\label{eq:Poincare LBdl}
	\CL = \big( X_{\sfG'_\qm} \times M_\rms \times \FC \big)/\sim\,,
	\qquad
	(\hat{\lambda}, x, z) \sim \big( \hat{\lambda}, \gamma(\xi)(x), \hat{\lambda}(\hat{\xi})\,z \big)\,,
\end{equation}
where $\xi$ is the projection of $\hat{\xi}$ to $\sfG'$, and where $\gamma$ is the action of $\sfG$ on $M_\rms$.
It is shown in~\cite[Appendix D]{Freed-Moore--TEM} that the fibres
\begin{equation}
	E_{|\hat{\lambda}} \coloneqq \rmL^2 \big( \{\hat{\lambda}\} \times M_\rms/\sfG', \CL \big)
\end{equation}
glue together to give a Hilbert bundle $E \rightarrow
\hat{X}_{\sfG'_\qm}$ as desired.
The canonical unitary isomorphism $\rmU$ is given explicitly by
\begin{equation}
	\rmU^{-1}(\psi)(x) = \int_{\hat{X}_{\sfG'_\qm}} \,
        \psi(\hat{\lambda}, x)\ \dd \nu(\hat{\lambda})\ ,
\end{equation}
where $\nu$ is the pullback of a fixed Haar measure on $X_{\sfG'}$
along the isomorphism of Lemma~\ref{st:QM_extensions_and_Pontryagin_duals}, and we have made use of the explicit construction~\eqref{eq:Poincare LBdl}.

Dropping the $\sfU(1)$ factor in the split $\sfG'_\qm \cong \sfG' \times \sfU(1)$ and recalling the isomorphism from
Lemma~\ref{st:QM_extensions_and_Pontryagin_duals}, one can perform an analogous construction over $X_{\sfG'}$; the resulting
Hilbert bundles are canonically isomorphic, and we denote both by $E$. We let $\langle-,-\rangle$ denote the canonical dual pairing between the vector
space $V$ and its dual $V^*$.

\begin{theorem}
\phantomsection\label{st:Equivariance_of_E_and_PT}
\begin{myenumerate}
	\item The bundle $E$ is $\hat{L}$-twisted $\sfG''$-equivariant: there is a unitary isomorphism
	\begin{equation}
	\label{eq:def_alpha_Bloch}
		\hat{\alpha} \colon (\hat{L} \otimes d_0^*
                E)^{[\phi_\sfG]} \arisom d_1^*E \ ,\quad
		\big[ (\hat{\lambda},\, \hat{g}),\, z \big] \otimes \psi \longmapsto
		\rho_\qm(\hat{g}) ( z\, \psi)\ ,
	\end{equation}
	which is compatible with the isomorphism $\hat\mu$ from Proposition~\ref{st:Brillouin_gerbe}.
	
	\item Let $\lambda_{(-)} \colon [0,1] \to X_\Pi \cong V^*/\Pi^*$, $t \mapsto \lambda_t$ be a continuous path in $X_\Pi$.
	If $\sfG'_\qm = \Pi \times \sfU(1)$, there is a parallel transport on $E$ given by
	\begin{equation}
	\label{eq:PT_on_E}
		\big(P^E_{\lambda_t} (\psi) \big)(x) = \exp\big(
                -2\pi\,\iu\, \< k_t - k_0 ,\, x - x_0 \> \big)\,
                \psi(x)\ ,
	\end{equation}
	where $k_t$ is any lift of $\lambda_t$ to $V^*$ and $x_0$ is any fixed origin in $M_\rms$.
\end{myenumerate}
\end{theorem}

\begin{proof}
To prove (1) we observe that if $\psi \in E_{|R_{g''} \hat{\lambda}}$, i.e. $\rho_\qm(\hat{\xi}\,) (\psi) = \hat{\lambda}(\hat{g}\, \hat{\xi}\, \hat{g}^{-1})^{\phi_\sfG(g''\,)}\, \psi$, then
\begin{equation}
\begin{aligned}
	\big( \rho_\qm(\hat{\xi}\,)\, \rho_\qm(\hat{g}) \big) (\psi) &= \big( \rho_\qm(\hat{g})\, \rho_\qm(\hat{g}^{-1}\, \hat{\xi}\, \hat{g}) \big) (\psi)\\[4pt]
	&= \rho_\qm(\hat{g})\, \big( \hat{\lambda}(\hat{\xi}\,)^{\phi_\sfG(g''\,)}\, \psi \big)\\[4pt]
	&= \hat{\lambda}(\hat{\xi}\,)\, \rho_\qm(\hat{g})(\psi)\ .
\end{aligned}
\end{equation}
Hence $\hat\alpha$ maps between the correct fibres over $\hat X_{\sfG'_\qm} \times \sfG''$.
Moreover, it is well-defined as we see from
\begin{equation}
\begin{aligned}
	&\hat\alpha_{|(\hat{\lambda}, g''\,)} \big( \big[ (\hat{\lambda},\, \hat{g}\, \hat{\xi}\,),\, \hat{\lambda}(\hat{g}\, \hat{\xi}\, \hat{g}^{-1})^{-\phi_\sfG(g''\,)}\, z \big] \otimes \psi \big)\\
	&\hspace{6cm} = \rho_\qm(\hat{g})\, \rho_\qm(\hat{\xi}\,)\, \big( \hat{\lambda}(\hat{g}\, \hat{\xi}\, \hat{g}^{-1})^{-\phi_\sfG(g''\,)}\, z\, \psi \big)\\[4pt]
	&\hspace{6cm} = \rho_\qm(\hat{g})\, \big( \hat{\lambda}(\hat{g}\, \hat{\xi}\, \hat{g}^{-1})^{-\phi_\sfG(g''\,)}\, z\, \hat{\lambda}(\hat{g}\, \hat{\xi}\, \hat{g}^{-1})^{\phi_\sfG(g''\,)}\, \psi \big)\\[4pt]
	&\hspace{6cm}= \hat\alpha_{|(\hat{\lambda}, g''\,)} \big( \big[
        (\hat{\lambda},\, \hat{g}),\, z \big] \otimes \psi \big)\ .
\end{aligned}
\end{equation}
Since $\hat\alpha$ is invertible and unitary, it is an isomorphism.
Finally, compatibility with $\hat\mu$ is evident from
\begin{equation}
\begin{aligned}
	&\hat\alpha_{|(\hat{\lambda}, g''\,)} \circ (\One \otimes \hat\alpha_{|(R_{g''} \hat{\lambda},\, h''\,)}) \big( \big[ (\hat{\lambda},\, \hat{g}),\, z \big] \otimes \big[ (R_{g''} \hat{\lambda},\, \hat{h}),\, z'\, \big] \otimes \psi \big)\\*
	&\qquad \qquad \qquad \qquad = \rho_\qm(\hat{g}) \big( z\, \rho_\qm(\hat{h}) (z'\, \psi) \big)\\[4pt]
	&\qquad \qquad \qquad \qquad = \rho_\qm(\hat{g}\, \hat{h}) \big( z'\, \theta_{h''}(z)\, \psi \big)\\[4pt]
	&\qquad \qquad \qquad \qquad = \hat\alpha_{|(\hat{\lambda},\, g''\,h''\,)}\, \big( \big[ (\hat{\lambda},\, \hat{g}\, \hat{h}),\, z'\, \theta_{h''}(z) \big] \otimes \psi \big)\\[4pt]
	&\qquad \qquad \qquad \qquad = \hat\alpha_{|(\hat{\lambda},\, g''\,h''\,)}\,
        \circ (\hat\mu_{|(\hat{\lambda},\, g'',\, h''\,)} \otimes \One)
        \big( \big[ (\hat{\lambda},\, \hat{g}),\, z \big] \otimes
        \big[ (R_{g''} \hat{\lambda},\, \hat{h}),\, z'\, \big] \otimes
        \psi \big)\ .
\end{aligned}
\end{equation}

To prove (2), note that the map $P^E$ is linear and unitary by construction.
It is independent of the choice of lift $k_t$ of $\lambda_t$, and moreover $P^E$ is compatible with concatenation of paths, as for any two paths which can be concatenated in $X_\Pi$ there exist lifts which can also be concatenated in $V^*$.
\end{proof}

By the assumptions (A1)--(A4) of Section~\ref{sect:Band_insulators}, for $\sfG'_\qm = \Pi \times \sfU(1)$ the Hamiltonian of the topological insulator induces an endomorphism of $E$ which is compatible with the $\hat{L}$-twisted $\sfG''$-action.
Therefore the $\hat{L}$-twisted $\sfG''$-action is also compatible
with the grading that $H_{|\lambda}$ induces on $E_{|\lambda}$, making
$E^\pm \rightarrow X_\Pi$ into $\hat{L}$-twisted $\sfG''$-equivariant
Hilbert bundles as well (see assumption~(A4) from Section~\ref{sect:Band_insulators}).
Using the inclusions $\imath_{E^\pm} \colon E^\pm \hookrightarrow E$
together with the projections $\pr_{E^\pm} \colon E \rightarrow E^\pm$, the bundles $E^\pm$ even carry connections given by
\begin{equation}
\label{eq:Berry_connection}
	\nabla^{E^\pm} \coloneqq \pr_{E^\pm} \circ \nabla^E \circ
        \imath_{E^\pm}\ ,
\end{equation}
where $\nabla^E$ is the covariant derivative on $E$ corresponding to the parallel transport $P^E$.
The bundles $E^\pm \rightarrow X_\Pi$ constructed in this way are the
bundles used in the K-theory classification of gapped topological
phases in~\cite{Freed-Moore--TEM}, as we discussed in Section~\ref{sect:Brillouin_gerbe_and_KT_classification}.

\begin{definition}[Berry connection]
The Hermitean connection $\nabla^{E^-}$ on the finite-rank Bloch bundle $E^-$ is called the \emph{Berry connection}.
\end{definition}

\subsection{Sections of quantum extensions}
\label{sect:splittings_and_equivariance}

In preparation for our later treatments of time-reversal symmetry and
the Kane-Mele invariant, we need a brief technical interlude on maps $\hat s \colon \sfG'' \to \sfG_\qm$ which split the projection $\sfG_\qm \to \sfG''$.
The presence of such maps is weaker than the existence of splittings
of the short exact sequence $\sfG'_\qm \to \sfG_\qm \to \sfG''$, since the latter splittings
are additionally group homomorphisms.
However, it turns out that topological splittings $\hat s$ already have
strong implications on the structure of the $\hat{L}$-twisted
$\sfG''$-action on the bundle $E$.

Let $I=\hat{\scG}_1\times\FC$ denote the trivial Hermitean line bundle
on $\hat{\scG}_1= \hat{X}_{\sfG'_\qm}\times \sfG''$ with its $\phi_{\sfG}$-graded multiplication denoted by $m$.

\begin{lemma}
\label{st:sections 2-cocycles and trivialisation}
Let $\sfG'_\qm \cong \Pi \times \sfU(1)$. A section $\hat{s}
\colon \sfG'' \to \sfG_\qm$ of the $\sfG'_\qm$-bundle $\sfG_\qm \to \sfG''$ gives rise to:
\begin{myenumerate}
	\item A $\phi_\sfG$-graded group $2$-cocycle $\beta = (\zeta, u) \colon \sfG'' \times \sfG'' \to \sfG'_\qm$.
	That is
	\begin{align}
		\hat{s}(g''\,)\, \hat{s}(h''\,) &= \hat{s}(g'' \,
                                                  h''\,)\,
                                                  \beta(g'',h''\,)\ ,\\[4pt]
		\zeta(g'', h'' \, k''\,)\, \zeta(h'', k''\,) &=
                                                               \zeta(g'',
                                                               h''\,)\,
                                                               \zeta(g''\,
                                                               h'',
                                                               k''\,)\
                                                               ,\\[4pt]
		\label{eq:graded_cocycle_U(1)-part}
		u(g'', h'' \, k''\,)\, u(h'',k''\,) &= u(g'',
                                               h''\,)^{\phi_\sfG(k''\,)}\,
                                               u(g''\, h'', k''\,)\ ,
	\end{align}
for all $g'',h'',k'' \in \sfG''$.
	
	\item A trivialisation $\chi \colon I \arisom \hat{L}$ of $\hat{L}$ as a Hermitean line bundle, whose incompatibility with $\hat\mu$ is given by
	\begin{align}
		\chi_{|(\hat{\lambda},\, g'' \, h''\,)}^{-1} \circ \hat\mu \circ \big( \chi_{|(\hat{\lambda},\, g''\,)} \otimes \chi_{|(R_{g''} \hat{\lambda},\, h''\,)} \big)
		&= \hat{\lambda} \big( \beta(g'',h''\,)
                  \big)^{\phi_\sfG(g''\, h''\,)} \circ m\\[4pt]
		&= u(g'', h''\,)^{\phi_\sfG(g''\, h''\,)}\,
                  \lambda \big( \zeta(g'',h''\,)
                  \big)^{\phi_\sfG(g''\, h''\,)} \circ m \notag
	\end{align}
	 for all $\hat{\lambda} \in \hat{X}_{\sfG'_\qm}$ and $g'',h'' \in \sfG''$.
\end{myenumerate}
\end{lemma}

\begin{proof}
(1) is a straightforward consequence of the definition of $\beta(g'',h''\,)$.
It only has to be noted that in the equation deduced from associativity in the group $\sfG''$ we have to move $u(g'',h''\,)$ past $\hat{s}(k''\,)$ to the right, thus yielding $\phi_\sfG(k''\,)$ in the exponent in~\eqref{eq:graded_cocycle_U(1)-part}.

For (2), recall the construction~\eqref{eq:Brillouin_gerbe_qm}, and set $\chi_{|(\hat{\lambda},\,g''\,)} (\hat{\lambda},\, z) = [(\hat{\lambda},\, \hat{s}(g''\,)),\, z]$.
This defines an isomorphism $\chi:I \to \hat{L}$.
The failure of multiplicativity controlled by $\beta$ then follows from
\addtocounter{equation}{1}
\begin{align*}
	&\chi_{|(\hat{\lambda},\, g''\, h''\,)}^{-1} \circ \hat\mu \circ \big( \chi_{|(\hat{\lambda},\, g''\,)} \otimes \chi_{|(R_{g''} \hat{\lambda},\, h''\,)} \big) \big( (\hat{\lambda},z) \otimes (R_{g''} \hat{\lambda}, z'\,) \big)\\
	&\hspace{4cm}= \chi_{|(\hat{\lambda},\, g''\, h''\,)}^{-1} \circ \hat\mu \big( \big[ \big( \hat{\lambda},\, \hat{s}(g''\,) \big),\, z \big] \otimes \big[ \big( R_{g''} \hat{\lambda},\, \hat{s}(h''\,) \big),\, z'\, \big] \big)\\[4pt]
	&\hspace{4cm}= \chi_{|(\hat{\lambda},\, g'' \, h''\,)}^{-1} \big[ \big( \hat{\lambda},\, \hat{s}(g''\,)\, \hat{s}(h''\,) \big),\, \theta_{h''}(z)\, z'\, \big] \tag{\theequation}\\[4pt]
	&\hspace{4cm}= \chi_{|(\hat{\lambda},\, g''\, h''\,)}^{-1}
          \big[ \big( \hat{\lambda},\, \hat{s}(g''\, h''\,) \big),\,
          \hat{\lambda} \big( \beta(g'',h''\,) \big)^{\phi_\sfG(g''\, h''\,)}\, \theta_{h''}(z)\, z'\, \big]\\[4pt]
	&\hspace{4cm}= \big( \hat{\lambda},\, \hat{\lambda} \big(
          \beta(g'',h''\,) \big)^{\phi_\sfG(g''\, h''\,)}\, \theta_{h''}(z)\, z'\, \big)\\[4pt]
	&\hspace{4cm}= \hat{\lambda} \big( \beta(g'',h''\,)
          \big)^{\phi_\sfG(g''\, h''\,)}\, m\big( (\hat{\lambda},z)
          \otimes (R_{g''} \hat{\lambda}, z'\,) \big)\ ,
\end{align*}
as claimed.
\end{proof}

Such a trivialisation of $\hat{L}$ almost allows us to remove the
twisting line bundle from the twisted $\sfG''$-action on $E$, but the
failure of multiplicativity will be visible in a nonclosure of the
candidate maps for the lifted $\sfG''$-action.

\begin{proposition}
\label{st:Sections_of_Gqm_and_Bloch_bdl}
Let $\sfG'_\qm \cong \Pi \times \sfU(1)$ and $\hat{s} \colon \sfG'' \to \sfG_\qm$ a section of the $\sfG'_\qm$-bundle $\sfG_\qm \to \sfG''$.
Denote the $\phi_\sfG$-graded $\sfG'_\qm$-valued $2$-cocycle it defines by $\beta$, and the induced trivialisation of $\hat{L}$ by~$\chi$.
\begin{myenumerate}
	\item Composing the action morphism $\hat{\alpha}$ from~\eqref{eq:def_alpha_Bloch} with $\chi$ yields a continuous family of bundle isomorphisms
	\begin{equation}
	\label{eq:twisted action on E}
		\alpha_{\hat{s},g''} = \hat{\alpha}_{|(-,g''\,)} \circ \chi^{[\phi_\sfG(g''\,)]} \colon \big( R_{g''}^* E \big)^{[\phi_\sfG(g''\,)]} \arisom E
	\end{equation}
	parameterised by $\sfG''$, which satisfy
	\begin{equation}
		\alpha_{\hat{s},h''|\hat{\lambda}} \circ \big( R_{h''}^* \alpha_{\hat{s},g''} \big)^{[\phi_\sfG(h''\,)]}_{|\hat{\lambda}}
		=  \hat{\lambda} \big( \beta(h'',g''\,)
                \big)^{\phi_\sfG(h''\, g''\,)}\ \alpha_{\hat{s},h''\, g'' |\hat{\lambda}}
	\end{equation}
for all $g'',h'' \in \sfG''$ and $\hat\lambda\in\hat X_{\sfG_\qm'}$.
	Thus $\alpha_{\hat{s},-}$ are multiplicative up to a $\TopAb(X_\Pi, \sfU(1))$-valued $\phi_\sfG$-graded $2$-cocycle
	\begin{equation}
		\ev \circ \beta \colon \sfG'' \times \sfG''
                \longrightarrow \TopAb(X_\Pi, \sfU(1))\ , \quad
		\big(\ev \circ \beta(g'',h''\,)\big) (\hat{\lambda}) =
                \hat{\lambda}\big(\beta(g'',h''\,)\big) \ .
	\end{equation}
	
	\item If $\zeta(g'',h''\,) = 0$ for all $g'',h'' \in \sfG''$, the $2$-cocycle $\ev \circ \beta$ on $(g'',h''\,)$ is multiplication by the constant $u(g'',h''\,)$.
\end{myenumerate}
\end{proposition}

\begin{proof}
Let $\psi \in E_{|R_{g''} \hat{\lambda}}$.
Then
\begin{equation}
	\alpha_{\hat{s},g''} (\psi)
	= \hat{\alpha} \big( \big[ (\hat{\lambda},\, \hat{s}(g''\,)),\, 1 \big] \otimes \psi \big)
	= \rho_\qm\big(\hat{s}(g''\,)\big) (\psi)\ .
\end{equation}
Thus for $\psi \in E_{|R_{g''\, h''}\hat{\lambda}}$ we find
\begin{align}
	\alpha_{\hat{s},h''|\hat{\lambda}} \circ \big( R_{h''}^* \alpha_{\hat{s},g''} \big)^{[\phi_\sfG(h''\,)]}_{|\hat{\lambda}} (\psi)
	&= \rho_\qm \big( \hat{s}(h''\,)\, \hat{s}(g''\,) \big) (\psi)\nonumber\\[4pt]
	&= \rho_\qm \big( \hat{s}(h''\, g''\,)\, \beta(h'',g''\,) \big) (\psi)\\[4pt]
	&= \hat{\lambda} \big( \beta(h'',g''\,) \big)^{\phi_\sfG(h''\,
          g''\,)}\, \rho_\qm \big( \hat{s}(h''\, g''\,) \big) (\psi)\nonumber\\[4pt]
	&= \hat{\lambda} \big( \beta(h'',g''\,) \big)^{\phi_\sfG(h''\,
          g''\,)}\, \alpha_{\hat{s},h''\, g'' |\hat{\lambda}} (\psi)\
          . \nonumber
\end{align}
This proves (1).

Statement (2) is seen from
\begin{equation}
	\hat{\lambda}\big(\beta(g'',h''\,)\big) =
        \hat{\lambda}\big(\zeta(g'',h''\,),\, u(g'',h''\,)\big) =
        u(g'',h''\,)\, \lambda\big(\zeta(g'',h''\,)\big)\ .
\end{equation}
Assuming that $\zeta$ is trivial then yields the statement.
\end{proof}

\begin{corollary}
If $\hat s \colon \sfG'' \to \sfG_\qm$ is a group homomorphism, i.e. a
splitting of the short exact sequence $\sfG'_\qm \to \sfG_\qm \to
\sfG''$ of groups, then $\chi \colon I \to \hat L$ is compatible with
the graded multiplicative structures and $E$ is a graded
$\sfG''$-equivariant bundle on $X_\Pi$ with the lift of the $\sfG''$-action to $E$ given by $g'' \mapsto \alpha_{\hat{s},g''}$.
\end{corollary}

In the case where the $2$-cocycle $\ev \circ \beta$ does not depend on $\lambda$, i.e. for every $g'', h'' \in \sfG''$ the function $\ev (\beta(g'', h''\,)) \colon X_\Pi \to \sfU(1)$ is constant, we can use averaging to obtain connections on $E$ that are compatible with the family $\alpha_{\hat{s},g''}$.

\begin{proposition}
\label{st:averaged_connections}
In the setting of
Proposition~\ref{st:Sections_of_Gqm_and_Bloch_bdl}~(2) and the
assumption~(A1) from Section~\ref{sect:Band_insulators}, we can average any connection on $E$ over $\sfG''$ to obtain a connection which is invariant under all $\alpha_{\hat{s},g''}$: for any connection $\nabla^E$ on $E$, set
\begin{equation}
	\av(\nabla^E) \coloneqq \int_{\sfG''}\, \alpha_{\hat{s},g''} \circ (R_{g''}^* \nabla^E)^{[\phi_\sfG(g''\,)]} \circ \alpha_{\hat{s},g''}^{-1}\ \dd \nu_{\sfG''}(g''\,)\ ,
\end{equation}
where $\nu_{\sfG''}$ is the normalised Haar measure on $\sfG''$ and $\nabla^{E[-1]} = \nabla^{E^{[-1]}}$ is the connection induced on $\overline{E}$ by $\nabla^E$.
Then
\begin{equation}
\label{eq:2-cocycle-twisted_equivariance_of_connection}
	\alpha_{\hat{s},g''} \circ \big( R_{g''}^*\, \av(\nabla^E) \big)^{[\phi_\sfG(g''\,)]} \circ \alpha_{\hat{s},g''}^{-1} = \av(\nabla^E)
\end{equation}
for all $g'' \in \sfG''$. In particular, in this setting there exist connections on $E$ and $E^-$ satisfying~\eqref{eq:2-cocycle-twisted_equivariance_of_connection}.
\end{proposition}

\begin{proof}
We compute
\addtocounter{equation}{1}
\begin{align*}
	&\alpha_{\hat{s},g''} \circ \big( R_{g''}^*\, \av(\nabla^E) \big)^{[\phi_\sfG(g''\,)]} \circ \alpha_{\hat{s},g''}^{-1}\\
	&= \int_{\sfG''}\, \alpha_{\hat{s},g''} \circ \Big(
          R_{g''}^*\big( \alpha_{\hat{s},h''} \circ (R_{h''}^*
          \nabla^E)^{[\phi_\sfG(h''\,)]} \circ
          \alpha_{\hat{s},h''}^{-1} \big) \Big)^{[\phi_\sfG(g''\,)]} \circ \alpha_{\hat{s},g''}^{-1}\ \dd \nu_{\sfG''}(h''\,) \tag{\theequation}\\[4pt]
	&= \int_{\sfG''}\, \alpha_{\hat{s},g''} \circ \big(
          R_{g''}^*\, \alpha_{\hat{s},h''} \big)^{[\phi_\sfG(g''\,)]}
          \circ (R_{g''\, h''}^* \nabla^E)^{[\phi_\sfG(g''\, h''\,)]} \circ \big( R_{g''}^*\, \alpha_{\hat{s},h''}^{-1} \big)^{[\phi_\sfG(g''\,)]} \circ \alpha_{\hat{s},g''}^{-1}\ \dd \nu_{\sfG''}(h''\,)\\[4pt]
	&= \int_{\sfG''}\, u(g'',h''\,)^{\phi_\sfG(g''\, h'')}\
          \alpha_{\hat{s},g''\, h''}\circ (R_{g''\, h''}^*
          \nabla^E)^{[\phi_\sfG(g''\, h''\,)]}\circ
          \alpha_{\hat{s},g''\, h''}^{-1}\
          u(g'',h''\,)^{-\phi_\sfG(g''\, h''\,)}\ \dd \nu_{\sfG''}(h''\,)\\[4pt]
	&= \av(\nabla^E)\ .
\end{align*}
Here it was crucial that $u(g'',h''\,)$ does not depend on $\lambda$.
That is, we need $\zeta$ to be trivial in order to obtain proper invariance of the connection from averaging.
In the last equality we have used the invariance of the Haar measure.
As there always exists a connection on $E$ as constructed in Theorem~\ref{st:Equivariance_of_E_and_PT}, we can always use its average over $\sfG''$.
The same argument then applies to the Berry connection $\nabla^{E^-}$
on the Bloch bundle $E^-$ from~\eqref{eq:Berry_connection}, and the
result thus follows.
\end{proof}

\subsection{Time-reversal symmetric quantum phases}
\label{sect:pure_time_reversal_in_TEM}

We now turn to an explicit example of a quantum particle in a crystal with a single time-reversal symmetry.
The situation treated here will be the framework for the analysis in Part~\ref{part:KM and localisation}, and this section forms the transition from Part~\ref{part:TPMs} to Part~\ref{part:KM and localisation}.
Consider a crystal $C \subset M = \FR^{d+1}$ with space translation symmetries $\Pi = \RZ^d$.
A Hamiltonian $H$ given by a Schr\"odinger operator on $\FR^d$ is
unbounded from above.
Then, if such a system gives rise to a topological insulator as in Definition~\ref{def:Band_Insulator_Bloch_bundle}, it will be of type I.
Consequently, there can be no chiral time-reversing symmetries
$\hat{T} \colon \scH \to \scH$ by
Proposition~\ref{st:time-reversal_and_spectrum_of_H}.\footnote{Chiral
  symmetries can be obtained by combining time-reversal with
  particle-hole symmetries. These latter
  symmetries are relevant to topological superconductors, which we do
not consider in this paper.}
Every time-reversing symmetry $\hat T$ on quantum systems of this type must be implemented as antiunitary operators; that is, $\phi_{\scH}(\hat{T}) = -1$ in the notation of Section~\ref{sect:QM_symmetries}.
There is a small but important observation about such symmetries, which is part of the more extensive result of~\cite[Lemma~6.17]{Freed-Moore--TEM}.

\begin{lemma}
\label{st:Quantum_extensions_of_ZZ}
Let $T \in \Aut_\qm(\PP\scH,H)$ be a projective quantum symmetry with $\phi_\scH(T) = -1$ and $T^2 = 1_{\PP\scH}$.
Then any lift $\hat{T} \in \Aut_\qm(\scH,H)$ of $T$ to a quantum symmetry automatically satisfies $\hat{T}^2 = \pm\, \One$.
\end{lemma}

\begin{proof}
We have $\hat{T}^2 = u\, \One$ for some $u \in \sfU(1)$ since
$\hat{T}^2$ is a lift of $1_{\PP\scH}$. Using the antilinearity of
$\hat{T}$, we then have $\hat{T} = \hat{T}^2\, \hat{T}\, \hat{T}^{-2} = u\, \hat{T}\, u^{-1} = u^2\, \hat{T}$, whence $u = \pm\, 1$.
\end{proof}

We wish to consider the situation $\sfG'' = \ZZ$.
Time-reversal is supposed to have no effect on the space configuration, whence our symmetry groups fit into the short exact sequence
\begin{equation}
	\xymatrix{
		1 \ar@{->}[r] & \RZ^d \ar@{->}[r] & \RZ^d \times \ZZ
                \ar@{->}[r] & \ZZ \ar@{->}[r] & 1\ .
	}
\end{equation}
As before we assume that the quantum extension $\sfU(1) \rightarrow \sfG_\qm \rightarrow \sfG$, where now $\sfG = \RZ^d \times \ZZ$, is trivial over $\sfG' = \Pi = \RZ^d$.
Thus so far we have
\begin{equation}
	\sfG_\qm = \big( \RZ^d \times \sfU(1) \big) \times (\ZZ)_\qm\ ,
\end{equation}
where $(\ZZ)_\qm$ denotes some quantum extension of $\ZZ$.
Up to isomorphism there are only two such extensions by Lemma~\ref{st:Quantum_extensions_of_ZZ} given by $\hat{T}^2 = \pm\, 1$, where $T = -1$ is the generator of $\ZZ$ and $\hat{T}$ is any lift of $T$ to $(\ZZ)_\qm$.%
\footnote{There is an unfortunate clash of notation here: if $\hat{T}$ lifts $-1 \in \ZZ$ to $(\ZZ)_\qm$, what we actually mean by $\hat{T}^2 = -1$ is that $\hat{T}^2 = \iota_{\ZZ}(-1)$, i.e. the image of $-1 \in \sfU(1)$ under the inclusion $\iota_{\ZZ}$ of the quantum extension $((\ZZ)_\qm, \iota_{\ZZ},q_{\ZZ})$.}
Any such choice of lift is a section of the extension $p \colon (\ZZ)_\qm \rightarrow \ZZ$, i.e.~it gives rise to a map $s \colon \ZZ \to (\ZZ)_\qm$ such that $p \circ s = 1_{\ZZ}$; however, this does not induce a splitting of the extension since $\hat{T}^2 = -1$ does not agree with the lift of $T^2 = 1$.
Nevertheless, a choice of $\hat{T}$ allows us to decompose the group $(\ZZ)_\qm$ as follows: given $\hat{T}\in \Aut_\qm(\scH,H)$ and $u,v \in \sfU(1)$, we have the
operator products
\begin{equation}
\begin{aligned}
	&(\hat{T} \, u)\,(\hat{T} \, v) = \pm\, \One\, \overline{u} \,
        v\ ,\\[4pt]
	&(\hat{T} \, u)\,(\One \, v) = \hat{T}\, u \, v\ ,\\[4pt]
	&(\One \, u)\,(\hat{T} \, v) = \hat{T}\, \overline{u} \, v\ ,\\[4pt]
	&(\One \, u)\,(\One \, v) = \One\, u \, v\ .
\end{aligned}
\end{equation}
Thus $(\ZZ)_\qm \cong \sfU(1) \times \sfU(1)$ as topological spaces,
but with the product $\boldsymbol\cdot$ given as follows:
Write $(\ZZ)_\qm \cong \sfU(1)_+ \times \sfU(1)_-$.
This is a $\ZZ$-graded topological group with grading $\phi_{\ZZ}(u) = \pm\, 1$ if and only if $u \in \sfU(1)_\pm$, and
\begin{equation}
	u \boldsymbol\cdot v =%
	\begin{cases}
		\ \pm\, \overline{u} \, v \ , & u, v \in \sfU(1)_- \ , \\
		\ u \, v \ , & u \in \sfU(1)_- \ ,\ v \in \sfU(1)_+ \ , \\
		\ \overline{u} \, v \ , & u \in \sfU(1)_+ \ ,\ v \in
                \sfU(1)_- \ , \\
		\ u \, v \ , & u, v \in \sfU(1)_+\ ,
	\end{cases}
\end{equation}
where the grading of the product is given by the product of the gradings.

For the action of $\sfG'' = \ZZ$ on the $d$-dimensional Brillouin torus $X_\Pi =
X_{\RZ^d} = \TT^d$, from~\eqref{eq:Brillouin_action} we get
\begin{equation}
\label{eq:action_of_time-reversal_on_Brillouin_torus}
	\tau(\lambda) \coloneqq R_{-1} \lambda =
        \lambda^{\phi_{\ZZ}(-1)} \circ \alpha(-1) = \lambda^{-1}\ .
\end{equation}
This is precisely what we expect physically from the interpretation of $X_\Pi \cong V^*/\Pi^*$ as periodic momenta of particles in $M_\rms$:
a pure time-reversal should only reflect momenta
$k \in V^*$, $k \mapsto -k$, so that in $X_\Pi$ we have $\lambda =
\exp(2\pi\,\iu\, \<k,-\>) \mapsto \exp(2\pi\,\iu\, \<-k,-\>) =
\lambda^{-1}$, while $-1 \in \ZZ$ acts as the identity on $M_\rms$.

We can extend the section $s \colon \ZZ \to (\ZZ)_\qm$ induced by a choice of lift $\hat{T}$ of $-1 \in \ZZ$ to a section $\hat{s} \colon \ZZ \to (\RZ^d \times \sfU(1)) \times (\ZZ)_\qm$ via
\begin{equation}
	\hat{s}(1) = (0,1,1) \qquad \mbox{and} \qquad \hat{s}(-1) = (0,1,\hat{T})\ .
\end{equation}
In the notation of Section~\ref{sect:splittings_and_equivariance} its $2$-cocycle $\beta$ is trivial if $\hat T^2=+1$, whereas if $\hat
T^2=-1$, which we henceforth consider, it is given by
\begin{equation}
\label{eq:2-cocycle_for_time-reversal}
	\beta(a_0, a_1) =
	\begin{cases}
		\ (0,1,1)\ , & a_0 \neq -1\ \text{ or }\ a_1 \neq -1 \ , \\
		\ (0,1,-1)\ , & a_0 = a_1 = -1\ .
	\end{cases}
\end{equation}
Thus \eqref{eq:twisted action on E} yields a bundle isomorphism $\alpha_{\hat{s},-1} \colon \overline{\tau^*E}  \arisom E$, while the isomorphism $\alpha_{\hat{s},1}$ is the identity on $E$.
From~\eqref{eq:2-cocycle_for_time-reversal} and Proposition~\ref{st:Sections_of_Gqm_and_Bloch_bdl} we see that $\zeta$, the $\Pi$-valued part of $\beta$, is trivial while the only nontrivial composition of isomorphisms is
\begin{equation}
	\alpha_{\hat{s},-1} \circ \overline{\tau^* \alpha_{\hat{s},-1}} = - \One\ .
\end{equation}
As all bundle morphisms here are compatible with the restriction to the Bloch bundle $E^-$, we obtain an isomorphism
\begin{equation}
	\hat{\tau} := \alpha_{\hat{s},-1|E^-} \colon \overline{\tau^*E^-} \arisom E^- \qquad \text{with} \quad
	\hat{\tau} \circ \overline{\tau^* \hat{\tau}} = - \One\ .
\end{equation}

Let $(X,\tau)$ be an involutive manifold, i.e.~a manifold $X$ with an automorphism $\tau$ such that $\tau^2 = 1_X$. Involutive manifolds are objects in
a category $\Mfd_{\ZZ}$ whose morphisms $f\in\Mfd_{\ZZ}\big((X,\tau)\,,\,
(X',\tau'\,) \big)$ are
equivariant smooth maps, i.e. $f\circ \tau=\tau'\circ f$; we will simply write
$X$ for the pair $(X,\tau)$ when there is no risk of
confusion. Recall~\cite{AtiyahReal} that a Hermitean vector bundle on an
involutive manifold $(X,\tau)$ is \emph{Real} if the involution $\tau$ lifts to an antiunitary involution on the total space. Because of the minus sign occuring above, which stems entirely from the nontriviality of the quantum extension of $\ZZ$, we need a slightly different notion~\cite{Dupont}.

\begin{definition}[Quaternionic vector bundle]
\label{def:Quaternionic_vector_bundle}
Let $(X,\tau)\in\Mfd_{\ZZ}$. A \emph{Quaternionic vector bundle} on $X$ is a Hermitean vector bundle $E \in \HVBdl(X)$ together with an isomorphism $\hat{\tau} \colon \overline{\tau^* E} \arisom E$ of Hermitean vector bundles satisfying $\hat{\tau} \circ \overline{\tau^* \hat{\tau}} = -\One$.%
\footnote{Note that Quaternionic vector bundles are related to quaternionic vector bundles (whose fibres are quaternionic modules) just as Real vector bundles are related to real vector bundles, i.e. by taking $\tau$ to be the identity map.}
A \emph{Quaternionic vector bundle with connection on $(X,\tau)$} is a hermitean vector bundle $E$ with connection $\nabla^E$ and a connection-preserving isomorphism $\hat{\tau} \colon \overline{\tau^* E} \arisom E$.

A \emph{morphism of Quaternionic vector bundles} $\varphi \colon (E, \hat{\tau}_E) \to (F,\hat{\tau}_F)$ is a morphism $\varphi \colon E \to F$ of Hermitean vector bundles such that $\hat{\tau}_F \circ \overline{\tau^*\varphi} = \varphi \circ \hat{\tau}_E$.
This defines a category $\QVBdl(X,\tau)$ of Quaternionic vector bundles on $X$.
A \emph{morphism of Quaternionic vector bundles with connection} is a morphism of the underlying Quaternionic vector bundles that preserves the connections.
This defines a category $\QVBdl^\nabla(X,\tau)$.
\end{definition}

\begin{remark}
If the isomorphism $\hat\tau$ in Definition~\ref{def:Quaternionic_vector_bundle} instead satisfies $\hat{\tau} \circ \overline{\tau^* \hat{\tau}} = +\One$, then the pair $(E,\hat\tau)$
defines a Real vector bundle over $(X, \tau)$.
\qen
\end{remark}

Hence $(E^-,\hat{\tau})$ is a Quaternionic vector bundle over the involutive manifold $(X_\Pi,\tau)$.
The Berry connection $\nabla^{E^-}$ is not generally compatible with
this structure, but by Proposition~\ref{st:averaged_connections} there
always exists a Quaternionic connection on $E^-$, such as the average of $\nabla^{E^-}$.
Analogously, if $(\ZZ)_\qm$ is the trivial extension, then the same constructions yield Real structures instead of Quaternionic ones.
The results of this section specialised to the case of time-reversal symmetry can be summarised as

\begin{proposition}
\label{prop:TsymTIs}
Let $\sfG_\qm = (\RZ^d \times \sfU(1)) \times (\ZZ)_\qm$.
\begin{myenumerate}
	\item
	If $(\ZZ)_\qm$ is the trivial quantum extension of $\ZZ$, then $(E^-,\nabla^{E^-})$ has the structure of a Real vector bundle with connection over $X_\Pi$.
	
	\item
	If $(\ZZ)_\qm$ is the nontrivial quantum extension of $\ZZ$,
        then $E^-$ has the structure of a Quaternionic vector bundle
        over $X_\Pi$. There exists a connection $\nabla$ on $E^-$ such that $(E^-,\nabla)$ is a Quaternionic vector bundle with connection on $X_\Pi$.
\end{myenumerate}
\end{proposition}

Topological insulators which fit into part~(1) of
Proposition~\ref{prop:TsymTIs} are of type AI in the Altland-Zirnbauer
classification~\cite{AltlandZirnbauer}, and correspond to particles with integer spin. Those described by part~(2) are of type
AII and correspond to half-integer spin particles. As Type AI topological insulators do not have a Kane-Mele invariant, and we are anyway interested in systems of fermions, only type AII will be considered in the remainder of this
paper. This is the starting point for many topological and geometric
investigations of time-reversal symmetric topological insulators, such
as the recent
accounts~\cite{Gawedzki--2DFKM_and_WZW,Gawedzki:BGrbs_for_TIs,Monaco-Tauber:Berry_WZW_FKM}.

\newpage

\part{The Kane-Mele Invariant}
\label{part:KM and localisation}

This part is the crux of the present paper, which is concerned with developing various new topological and
geometric perspectives on the Kane-Mele invariant for time-reversal
symmetric band insulators in three dimensions. We begin by
recalling the definition of the Kane-Mele invariant, and two
equivalent descriptions of it, in Section~\ref{sect:The KM invariant
  in Top and DG}; the lengthy proof of the main result implying the
first description is provided in Appendix~\ref{app:proof_of_1st_Chern_knows_everything}.
We then proceed to rephrase the Kane-Mele invariant as a homotopy
theoretic obstruction to a block decomposition of the sewing matrix
and propose an interpretation of the invariant based on this point of
view.

The integral representation of the invariant is investigated in detail
in Section~\ref{sect:Localisation in equivar HdR}, where we introduce
a Mayer-Vietoris long exact sequence that intertwines Real and
equivariant de Rham cohomology; we summarise the pertinent features of
de~Rham cohomology groups in the presence of $\ZZ$-actions in Appendix~\ref{app:HdR and group actions}.
This is subsequently used to localise the Kane-Mele invariant to the fixed points of the $\ZZ$-action on the Brillouin zone, as in the original discrete definition of the invariant as a Pfaffian formula~\cite{Kane-Mele:Z2_top_order_and_QSHE,FKM}.

The homotopy theoretic obstruction in Section~\ref{sect:The KM invariant
  in Top and DG} and the cohomological formulas in Section~\ref{sect:Localisation in equivar HdR} have a geometric refinement which lends a different point of view on the Kane-Mele invariant.
A language suitable to formalise this refinement is through bundle gerbes, which
are hence the central objects in Section~\ref{sect:holonomies_and_localisation}.
Bundle gerbes were introduced in~\cite{Murray:Bundle_gerbes}.
Their morphisms were defined in~\cite{Murray-Stevenson:BGrbs--Stable_iso_and_local_theory} and developed further in~\cite{Waldorf--More_morphisms,Waldorf--Thesis}.
The relevance of bundle gerbes to time-reversal symmetric topological
phases was first pointed out by Gawedzki~\cite{Gawedzki--2DFKM_and_WZW,Gawedzki:BGrbs_for_TIs,Gawedzki:2017whj}.
We investigate the holonomy of Jandl gerbes on surfaces with an orientation-preserving involution.
We show that, in certain situations such as the ones of interest for the Kane-Mele invariant, these holonomies again allow for a localisation technique which follows from the $2$-category theory of bundle gerbes.

\section{The Kane-Mele invariant in three dimensions}
\label{sect:The KM invariant in Top and DG}

\subsection{Definition of the Kane-Mele invariant}
\label{sect:KM invar:Def}

In Section~\ref{sect:Bloch_and_Berry_explicit} we gave the
construction of the Bloch bundle with the Berry connection $(E^-,\nabla^{E^-})$ over the Brillouin torus $X_\Pi$ for general band insulators.
Here we consider a situation similar to that in Section~\ref{sect:pure_time_reversal_in_TEM}.
We specialise to a crystal in three space dimensions, with space translation lattice $\Pi \cong \RZ^3$ and the only additional symmetry being time-reversal $\tau$.
Thus $X_\Pi \cong \TT^3$ where $\tau \colon \TT^3 \to \TT^3$ is
the inversion, which has an antiunitary lift $\hat{\tau}$ to $(E^-,\nabla^{E^-})$ satisfying $\hat{\tau} \circ \overline{\tau^*\hat{\tau}} = -\One$.
The material in Sections~\ref{sect:KM invar:Def} and~\ref{sect:KM as CS} is not new (see, for instance,~\cite{DeNittis-Gomi:DGInvariants_Real_case,DeNittis-Gomi:DGInvariants_ld-Quat_case,Monaco-Tauber:Berry_WZW_FKM,Panati:Triviality_of_Bloch_and_Dirac}), but it serves as a motivation for the following sections.
On the $d$-dimensional torus $\TT^d$, the inversion $\tau^d$ preserves
orientation if $d$ is even and reverses orientation if $d$ is odd.
Recall that for a generic diffeomorphism $\sigma \colon X \to X$ on an
oriented $d$-dimensional manifold $X$, one has
\begin{equation}
\label{eq:orientation_and_integrals}
	\int_X\, \sigma^* \omega = \ori(\sigma) \, \int_X\, \omega 
\end{equation}
for all $\omega \in \Omega^d(X)$, where $\ori(\sigma) = 1$ if $\sigma$ preserves the orientation of $X$ and $\ori(\sigma) = -1$ if $\sigma$ reverses the orientation of $X$.

\begin{proposition}
\label{st:Blochs_rank_and_1st_Chern_class}
Let $(E^-, \hat{\tau}) \in \QVBdl(\TT^3, \tau)$ be a Quaternionic vector bundle on $\TT^3$.
Then:
\begin{myenumerate}
	\item
	The rank of $E^-$ is even.
	
	\item
	The first Chern class of $E^-$ is trivial: $c_1(E^-) = 0$ in $\rmH^2(\TT^3,\RZ)$.
\end{myenumerate}
\end{proposition}

\begin{proof}
(1)
Here we follow the arguments of~\cite{Gawedzki--2DFKM_and_WZW,Freed-Moore--TEM}.
First, note that $\hat{\tau}_{|\lambda}$ defines an antiunitary map
\begin{equation}
	\hat{\tau}_{|\lambda} \colon \overline{(\tau^*E^-)_{|\lambda}}
        \cong \overline{E^-_{|\tau(\lambda)}} \arisom E^-_{|\lambda}\ .
\end{equation}
In particular, if $\lambda_0 \in \big(\TT^3\big)^\tau$ is a fixed point of the $\ZZ$-action then $\hat{\tau}_{|\lambda_0}$ is an antiunitary map $\hat{\tau}_{|\lambda_0} \colon \overline{E^-_{|\lambda_0}} \arisom E^-_{|\lambda_0}$ which squares to $-\One$.
That is, the Hermitean vector space $E^-_{|\lambda_0}$ carries a quaternionic structure.
In this situation, to every $\psi \in E^-_{|\lambda_0}$ we obtain $\hat{\tau}(\psi) \in E^-_{|\lambda_0}$ with
\begin{equation}
	h_{E^-} \big( \psi, \hat{\tau}(\psi) \big)
	= h_{E^-} \big( \hat{\tau}^2(\psi), \hat{\tau}(\psi) \big)
	= - h_{E^-} \big( \psi, \hat{\tau}(\psi) \big)\ ,
\end{equation}
where $h_{E^-}$ denotes the Hermitean metric on $E^-$ and in the first equality we have made use of the antiunitarity of $\hat{\tau}$.
Hence $\psi$ and $\hat{\tau}(\psi)$ are orthogonal.
In particular, for every basis vector of $E^-_{|\lambda_0}$ we obtain
another orthogonal basis vector by applying $\hat{\tau}$, and applying
$\hat\tau$ again yields minus the original vector.
Consequently $E^-_{|\lambda_0}$ is an even-dimensional complex vector
space, and since the rank $\rank(E^-)$ is constant along the fibres,
the bundle $E^-$ has even rank.

(2)
Since $\rank(E^-)$ is even, the isomorphism $\det(\hat{\tau}) \colon \overline{\tau^*\det(E^-)} \arisom \det(E^-)$ satisfies $\det(\hat{\tau}) \circ \tau^* \det(\hat{\tau}) = \One$, where $\det(E^-) \to \TT^3$ is the determinant line bundle of $E^-$.
Then $\det(E^-)$ is trivialisable (even as an equivariant Hermitean line bundle) by~\cite[Lemma 11.29]{Freed-Moore--TEM}.
Now use $c_1(E^-) = c_1(\det(E^-))$.%
\footnote{Since $\rmH^2(\TT^3,\RZ)$ is torsion-free, here this
  identity can be seen immediately from the Chern-Weil representatives
  $F_{\det(E^-)} = \tr(F_{E^-})$.}
\end{proof}

\begin{remark}
The proof of part~(1) of
Proposition~\ref{st:Blochs_rank_and_1st_Chern_class} immediately
generalises to any involutive manifold $(X,\tau)\in\Mfd_{\ZZ}$ for which
the action of $\ZZ$ on $X$
has at least one fixed point. On the other hand, the statement
of~\cite[Lemma 11.29]{Freed-Moore--TEM}, and hence statement~(2) of
Proposition~\ref{st:Blochs_rank_and_1st_Chern_class},  is proven only
for the case $X = \TT^3$ and does not generalise to arbitrary
involutive $3$-manifolds $(X,\tau)$: it strongly relies on the fact that the cohomology group $\rmH^2(\TT^3, \RZ)$ can be generated from $\tau$-invariant cocycles and that $\rmH^1(\TT^3, \RZ)$ is free.
\qen
\end{remark}

Proposition~\ref{st:Blochs_rank_and_1st_Chern_class} has strong implications: it implies that $E^-$ is in fact trivialisable itself~\cite{Gawedzki--2DFKM_and_WZW,Gawedzki:BGrbs_for_TIs}.
This follows from a statement which can be found as~\cite[Proposition 4]{Panati:Triviality_of_Bloch_and_Dirac}.
The proof relies on showing that there exists a `$4$-equivalence' $\sfB\sfU(m) \to \sfB\sfU(1) \simeq K(\RZ,2)$, where $m=\rank(E^-)$.
Here we have denoted the classifying space of a topological group $\sfG$ by $\sfB\sfG$, and the Eilenberg-MacLane space at level $n \in \NN_0$ for an abelian group $\sfA$ by $K(\sfA,n)$.
A continuous map $f \colon X \to Y$ of topological spaces is an \emph{$n$-equivalence} if the induced homomorphisms on homotopy groups $\pi_k(f):\pi_k(X)\to\pi_k(Y)$ are isomorphisms for $k < n$ and $\pi_n(f)$ is a surjection (see~\cite[Definition~3.17]{Switzer:AT}).
The proof of
Proposition~\ref{st:1st_Chern_class_knows_everything_in_dleq3} below can be made a little more explicit by realising that it is essentially the determinant map which establishes the $4$-equivalence, and so we give a full proof in Appendix~\ref{app:proof_of_1st_Chern_knows_everything}.

\begin{proposition}
\label{st:1st_Chern_class_knows_everything_in_dleq3}
Let $X$ be a connected manifold of dimension $d \leq 3$, and let $m \in 2\,\NN$ be an even positive integer.
Then the first Chern class induces an isomorphism of pointed sets
\begin{equation}
	c_1 \colon \pi_0\,\HVBdl_{m\, \sim}(X) \arisom \rmH^2(X,\RZ)
\end{equation}
between isomorphism classes of Hermitean vector bundles of rank $m$ on $X$ and the second integer cohomology of $X$.
\end{proposition}

This can be loosely summarised as saying ``in dimensions less than
four, the first Chern class classifies all Hermitean vector bundles of any fixed even rank''.
The combination of
Propositions~\ref{st:Blochs_rank_and_1st_Chern_class}
and~\ref{st:1st_Chern_class_knows_everything_in_dleq3} then
immediately yields

\begin{corollary}
\label{st:Bloch bdl Hermitean trivialisable}
Any Quaternionic vector bundle $(E^-, \hat{\tau}) \in \QVBdl(\TT^3, \tau)$ has an underlying Hermitean vector bundle $E^-$ which is trivialisable.
\end{corollary}

\begin{example}
Our fundamental example is the Bloch bundle $(E^-, \hat{\tau}) \in \QVBdl(\TT^3, \tau)$ on the 3-dimensional Brillouin torus of the space translation lattice $\Pi$.
Corollary~\ref{st:Bloch bdl Hermitean trivialisable} states that $E^- \to \TT^3$ is trivialisable as a Hermitean vector bundle on $\TT^3$.
However, a generic trivialisation of $E^-$ as a Hermitean vector
bundle is generally incompatible with the Quaternionic structure
$\hat{\tau}$ on $E^-$, i.e. the time-reversal symmetry of the
topological insulator.
Measuring how far this incompatibility can be circumvented by choosing
a different trivialisation of $E^-$ is the crux of the definition of
the Kane-Mele invariant that we shall give below.
\qen
\end{example}

Consider a Quaternionic vector bundle $(E^-, \hat{\tau})$ on a $3$-dimensional involutive manifold $(X, \tau)$.
Assume that there exists a trivialisation $\varphi \colon E^- \to X \times \FC^m$ of $E^-$ as a Hermitean vector bundle.
We then obtain a commutative diagram
\begin{equation}
\label{eq:trivialisations_of_E}
	\xymatrixrowsep{1.2cm}
	\xymatrixcolsep{1.2cm}
	\xymatrix{
		E^- \ar@{->}[r]^-{\hat{\tau}^{-1}} \ar@{->}[d]_-{\varphi}	&	\overline{\tau^* E^-} \ar@{->}[d]^-{\overline{\tau^*\varphi}}\\
		X \times \FC^m \ar@{-->}[r]_-{w_\varphi \circ \theta}	&	\overline{X \times \FC^m}
	}
\end{equation}
of Hermitean vector bundles on $X$, where $\theta \colon \FC^m \to
\overline{\FC^m}$ is complex conjugation with respect to a fixed
orthogonal basis $\{e_i\}$, i.e. $\theta(v^i\, e_i) = \overline{v^i}\, e_i$, and $w_\varphi \colon X \to \sfU(m)$ is the uniquely defined $\sfU(m)$-valued function making the diagram commute.
The map $w_\varphi$ is sometimes referred to as the \emph{sewing matrix}~\cite{Gawedzki--2DFKM_and_WZW} and it contains the obstruction to trivialising $E^-$ in a $\ZZ$-equivariant manner, i.e. in a way that is compatible with the time-reversal symmetry.
The condition $\hat{\tau} \circ \overline{\tau^*\hat{\tau}} = -\One$ translates to $\tau^*(w_\varphi \circ \theta) \circ (w_\varphi \circ \theta) = -\One$.
Since $w_\varphi$ is unitary this is equivalent to
\begin{equation}
\label{eq:trafo_of_w_under_involution}
	\tau^*w_\varphi = -w_\varphi^\trans \ ,
\end{equation}
with $w_\varphi^\trans(\lambda)$ denoting the transpose of $w_\varphi(\lambda) \in \sfU(m)$.
The map $w_\varphi$ therefore corresponds to a Quaternionic structure on the trivial Hermitean vector bundle $X \times \FC^m$.
We say that $E^-$ is \emph{Quaternionic trivial} if there exists a
trivialisation $\varphi$ such that $w_\varphi = w_0$ with
\begin{equation}
	w_0 := 
	\begin{pmatrix}
		0 & \One
		\\
		- \One & 0
	\end{pmatrix}\ .
\end{equation}

Any other trivialisation is of the form $\varphi' = \varphi \circ a$
for some function $a \colon X \to \sfU(m)$.
Hence the new transition function $w_{\varphi'} \circ \theta = \overline{\tau^*\varphi'}{}^{-1} \circ \hat{\tau}^{-1} \circ \varphi'$ reads as $w_{\varphi'} = \tau^*\,\overline{a}{\,}^{-1}\, w_\varphi\, \overline{a}$. Let
$H_{\sfU(m)}$ be the canonical $3$-form on the unitary group $\sfU(m)$. Its pullback under the sewing matrix
\begin{equation}
\label{eq:WZW3form}
w_\varphi^*H_{\sfU(m)} = \frac{1}{24 \pi^2}\, \tr \big((w_\varphi^{-1}\, \dd w_\varphi)^{\wedge 3}\big)
\end{equation}
is the Wess-Zumino-Witten $3$-form associated to the map $w_\varphi:X\to\sfU(m)$, and for any pair of maps $g, h \colon X \to \sfU(m)$ it satisfies the Polyakov-Wiegmann formula~\cite{Gawedzki-Waldorf:PW_formula_and_multiplicative_gerbes}
\begin{equation}
\label{eq:mulitplication_and_canonical_3-form}
	(g\,h)^*H_{\sfU(m)} = g^*H_{\sfU(m)} + h^*H_{\sfU(m)} +
        \frac{1}{8 \pi}\, \dd\, \tr \big( g^{-1}\, \dd g \wedge h\,
        \dd h^{-1} \big)\ .
\end{equation}
Assuming that $X$ is a closed and oriented manifold,\footnote{We say
  that a manifold $X$ is \emph{closed} if it is compact and $\partial X = \varnothing$.} and that $\tau$ reverses the orientation of $X$, we thus have
\begin{equation}
\label{eq:KM_independent_of_trivialisation}
\begin{aligned}
	\int_{X}\, w_{\varphi'}^*H_{\sfU(m)} &= \int_{X}\, \big( - \tau^*\,\overline{a}{\,}^*H_{\sfU(m)} + w_\varphi^*H_{\sfU(m)} + \overline{a}{\,}^*H_{\sfU(m)} \big)\\[4pt]
	&= \int_{X}\, w_\varphi^*H_{\sfU(m)} + 2\, \int_{X}\, \overline{a}{\,}^*H_{\sfU(m)}
\end{aligned}
\end{equation}
where we used \eqref{eq:orientation_and_integrals}.
The second term only contributes even integer values, so that the quantity
\begin{equation}
	\exp \bigg( \pi\, \iu\,
        \int_{X}\, w_\varphi^*H_{\sfU(m)} \bigg) \ \in \ \ZZ
\end{equation}
is well-defined independently of the choice of trivialisation $\varphi$.
It is therefore associated to the data of the Quaternionic vector bundle $(E^-, \hat{\tau})$ alone.

There is a symmetric monoidal functor
\begin{equation}
	U \colon \QVBdl(X, \tau) 
        \longrightarrow \HVBdl(X) \ ,
	\quad
	(E^-, \hat{\tau}) \longmapsto E^-
\end{equation}
that forgets the $\ZZ$-structures.
We define $\QVBdl_U(X, \tau)$ to be the full symmetric monoidal subcategory of $\QVBdl(X, \tau)$ on the objects $(E^-, \hat{\tau})$ such that $E^- = U(E^-, \hat{\tau})$ admits a trivialisation.

\begin{definition}[Kane-Mele invariant]
Let $(X, \tau)$ be a closed and oriented involutive manifold of
dimension $d = 3$ with orientation-reversing involution $\tau$ and
nonempty fixed point set $X^\tau \neq \varnothing$.
For $(E^-, \hat{\tau}) \in \QVBdl_U(X, \tau)$ and any trivialisation $\varphi$ of $E^-$, the \emph{Kane-Mele invariant} of $(E^-, \hat{\tau})$ is 
\begin{equation}
\label{eq:KM_invariant}
	\rmK\rmM(E^-, \hat{\tau}) \coloneqq \exp \bigg( \pi\, \iu\,
        \int_{X}\, w_\varphi^*H_{\sfU(m)} \bigg)\ \in\ \ZZ\ .
\end{equation}
Given a time-reversal symmetric band insulator on a $3$-dimensional
Bravais lattice with Bloch bundle $(E^-, \hat{\tau}) \in \QVBdl(\TT^3, \tau)$,
its \emph{Kane-Mele invariant} is the mod~$2$ index $\rmK\rmM(E^-, \hat{\tau}) \in \ZZ$.
\end{definition}

\begin{remark}
The category $\QVBdl_U(X, \tau)$ does not generally cover all possible Quaternionic vector bundles, i.e. it is a proper subcategory of $\QVBdl(X, \tau)$.
Therefore, from a mathematical point of view, the definition above of the Kane-Mele invariant is somewhat unsatisfactory, as it only applies to a very restricted class of Quaternionic vector bundles.
An invariant simultaneously refining and generalising the Kane-Mele invariant is introduced and investigated in~\cite{DeNittis-Gomi:Topological_nature_of_FKMM}.
In the case $(X, \tau) = (\TT^3, \tau)$ of interest from the physical point of view, however, Corollary~\ref{st:Bloch bdl Hermitean trivialisable} implies the equality (not just equivalence) of categories
\begin{equation}
	\QVBdl_U(\TT^3, \tau) = \QVBdl(\TT^3, \tau)\ .
\end{equation}
That is, we do not lose any information by restricting to $\QVBdl_U(\TT^3, \tau)$ for $3$-dimensional time-reversal symmetric topological insulators: the Kane-Mele invariant is in fact defined for all possible time-reversal symmetric band insulators in three space dimensions as constructed in Section~\ref{sect:pure_time_reversal_in_TEM}.
\qen
\end{remark}

The Kane-Mele invariant was first introduced in~\cite{Kane-Mele:Z2_top_order_and_QSHE,Fu-Kane:Time-reversal_polarisation,FKM}.
Its form derived here can also be found in~\cite{Monaco-Tauber:Berry_WZW_FKM,Gawedzki--2DFKM_and_WZW,Freed-Moore--TEM} for the case $X = \TT^3$.
The language of Quaternionic vector bundles is particularly emphasised by~\cite{DeNittis-Gomi:Classification_of_quaternionic_Bloch_bdls,DeNittis-Gomi:FKMM_in_low_dimension}.

Let $\rmK\rmQ(X, \tau) := \mathrm{Gr}( \pi_0\, \QVBdl_\sim(X, \tau))$ denote the Quaternionic K-theory group of $(X, \tau)$, and let $\rmK\rmQ_U(X, \tau) := \mathrm{Gr}( \pi_0\, \QVBdl_{U\,\sim}(X, \tau))$ denote the subgroup of $\rmK\rmQ(X, \tau)$ generated by the isomorphism classes of objects from $\QVBdl_U(X, \tau)$.

\begin{proposition}
\phantomsection\label{st:properties of KM}
\begin{myenumerate}
	\item The Kane-Mele invariant $\rmK\rmM(E^-, \hat{\tau})$ is
          independent of the choice of trivialisation $\varphi$ of $E^-$, i.e. it is a topological invariant.
	
	\item Isomorphic Quaternionic vector bundles $(E,
          \hat{\tau}_E)\cong(F, \hat{\tau}_F)$ in $\QVBdl_U(X, \tau)$ have
          the same Kane-Mele invariant: $\rmK\rmM(E,\hat\tau_E)=\rmK\rmM(F,\hat\tau_F)$.
	
	\item The Kane-Mele invariant defines a homomorphism of commutative monoids
	\begin{equation}
		\rmK\rmM \colon \pi_0\, \QVBdl_{U\,\sim}(X,\tau) \longrightarrow \ZZ\ .
	\end{equation}
	
	\item If $(E, \hat{\tau}_E), (F, \hat{\tau}_F) \in \QVBdl_U(X, \tau)$, then setting
	\begin{equation}
		\rmK\rmM \big( [(E, \hat{\tau}_E)\,,\,(F, \hat{\tau}_F)] \big)
		\coloneqq \rmK\rmM(E, \hat{\tau}_E) \cdot \rmK\rmM(F, \hat{\tau}_F)
	\end{equation}
	induces a homomorphism of abelian groups
	\begin{equation}
		\rmK\rmM \colon \rmK\rmQ_U(X, \tau) \longrightarrow \ZZ\ .
	\end{equation}
\end{myenumerate}
\end{proposition}

\begin{proof}
(1)
This has been checked already in~\eqref{eq:KM_independent_of_trivialisation}.

(2)
An isomorphism of Quaternionic vector bundles $\phi \colon (E, \hat{\tau}_E) \to (F, \hat{\tau}_F)$ allows us to induce a trivialisation of $E$ from any trivialisation of $F$ by precomposing with $\phi$.
This choice of trivialisation on $E$ makes the two $\sfU(m)$-valued transition functions on $X$ of $(E, \hat{\tau}_E)$ and $(F, \hat{\tau}_F)$ agree.

(3)
The compatibility with direct sum follows from the explicit form \eqref{eq:WZW3form} and the properties of the trace.

(4)
To see that $\rmK\rmM$ is well-defined on $\rmK\rmQ_U(X, \tau)$, for
any $(R,\hat\tau_R)\in\QVBdl_U(X,\tau)$ we observe that
\addtocounter{equation}{1}
\begin{align*}
	\rmK\rmM\big((E, \hat{\tau}_E) \oplus (R, \hat{\tau}_R)\big) \cdot \rmK\rmM\big((F, \hat{\tau}_F) \oplus (R, \hat{\tau}_R)\big)
	&= \rmK\rmM(E, \hat{\tau}_E) \cdot \rmK\rmM(F, \hat{\tau}_F) \cdot \rmK\rmM(R, \hat{\tau}_R)^2
	\\[4pt]
	&= \rmK\rmM(E,\hat{\tau}_E) \cdot \rmK\rmM(F,\hat{\tau}_F) \tag{\theequation}
\end{align*}
since $\rmK\rmM$ is $\ZZ$-valued.
\end{proof}

\begin{remark}
There is equality $\rmK\rmQ_U(\TT^3, \tau) = \rmK\rmQ(\TT^3, \tau)$ by Corollary~\ref{st:Bloch bdl Hermitean trivialisable}, and in particular explicit computations give $\rmK\rmQ(\TT^3, \tau)\cong\RZ\oplus\ZZ^4$ (see e.g.~\cite{Kitaev:Periodic_table,Freed-Moore--TEM,Thiang2014}).
In the notations of
Sections~\ref{sect:Brillouin_gerbe_and_KT_classification}
and~\ref{sect:pure_time_reversal_in_TEM}, there is also an isomorphism
of K-theory groups $\rmK\rmQ(\TT^3, \tau) \cong
\rmK^{\hat{L},\ZZ}(\TT^3)$ given by sending an $\hat{L}$-twisted
$\ZZ$-equivariant $\ZZ$-graded Hermitean vector bundle
$(E,\hat{\alpha})$ to the Quaternionic vector bundle $(E^-,\alpha_{\hat{s},-1|E^-})$ for a section $\hat{s}$ given by a fixed lift $\hat{T}$ of $-1 \in \ZZ$.
Using this isomorphism, one can define the Kane-Mele invariant on
$\hat L$-twisted $\ZZ$-equivariant K-theory groups as well, i.e. on (reduced)
topological phases realised as time-reversal symmetric band insulators
in three dimensions.
\qen
\end{remark}

\subsection{The Kane-Mele invariant and Chern-Simons theory}
\label{sect:KM as CS}

We shall now proceed to more geometric formulations and perspectives on the Kane-Mele invariant.
It follows from Proposition~\ref{st:averaged_connections} that the
Bloch bundle $(E^-, \hat{\tau}) \in \QVBdl(\TT^3, \tau)$ can always be
endowed with a $\hat\tau$-invariant connection $\nabla$.
We thus consider a Quaternionic vector bundle with compatible
connection $(E^-, \nabla, \hat{\tau}) $ on a closed and oriented
involutive manifold $(X, \tau)$ of dimension $d=3$ with $\ori(\tau)=-1$, whose underlying Hermitean vector bundle $E^-$ is trivialisable.
Since a trivialisation $\varphi:E^-\to X\times\FC^m$ of $E^-$ as a Hermitean vector bundle does not
necessarily trivialise the connection $\nabla$ on $E^-$, we obtain an induced connection $\dd + A$ on $X \times \FC^m$ for a connection $1$-form $A \in \Omega^1(X, \fru(m))$, where $\fru(m)$ is the Lie algebra of $\sfU(m)$.
Similarly, since the values of $A$ are skewadjoint, the bundle $\overline{X \times \FC^m}$ carries an induced connection $\dd + \tau^*\,\overline{A} = \dd - \tau^*A^\trans$.
These choices of connections make the
diagram~\eqref{eq:trivialisations_of_E} into a commutative square of
Hermitean vector bundles with connections so that
\begin{equation}
\begin{aligned}
	-\tau^*A^\trans &= (w_\varphi \circ \theta)\, A\, (w_\varphi \circ \theta)^{-1} + (w_\varphi \circ \theta)\, \dd (w_\varphi \circ \theta)^{-1}\\[4pt]
	&= w_\varphi\, \overline{A}\, w_\varphi^{-1} + w_\varphi\, \dd w_\varphi^{-1}\\[4pt]
	&= w_\varphi\, (- A^\trans)\, w_\varphi^{-1} + w_\varphi\, \dd
        w_\varphi^{-1}\ ,
\end{aligned}
\end{equation}
or equivalently after transposing, applying $\tau^*$, and using
$\tau^*w_\varphi = -w_\varphi^\trans$ we have
\begin{equation}
	A = w_\varphi^{-1}\, (\tau^*A)\, w_\varphi - w_\varphi^{-1}\,
        \dd w_\varphi\ .
\end{equation}

The Chern-Simons $3$-form associated to $A$ is defined as~\cite{Chern-Simons:Char_Forms_and_Geom_invars,Freed:Classical_CS_I}
\begin{equation}
	\rmC\rmS(A) = \frac{1}{4 \pi}\, \tr(A \wedge F_A) -
        \frac{1}{24 \pi}\, \tr \big( A \wedge [A \wedge A] \big)\ , 
\end{equation}
where $F_A=\dd A+\frac12\, [A\wedge A]$ is the curvature of the
connection $\dd+A$.
The Chern-Simons $3$-forms of $A$ and $\tau^*A$ are related as
\begin{equation}
\label{eq:CS_trafo}
\begin{aligned}
	\rmC\rmS(A)
&= \tau^*\, \rmC\rmS(A) + \dd \Big(\, \frac{1}{4 \pi}\, \tr \big(
w_\varphi^{-1}\, (\tau^*A)\, w_\varphi \wedge w_\varphi^{-1}\, \dd
w_\varphi \big) \Big) - w_\varphi^*H_{\sfU(m)}\ .
\end{aligned}
\end{equation}
Consequently, since $X$ is closed we have
\begin{equation}
	\int_{X}\, \rmC\rmS(A) = \int_{X}\, \tau^*\, \rmC\rmS(A) -
        \int_{X}\, w_\varphi^*H_{\sfU(m)} \ ,
\end{equation}
and making use of~\eqref{eq:orientation_and_integrals} we conclude that
\begin{equation}
\label{eq:Quaternionic_CS-invariant}
	\int_{X}\, \rmC\rmS(A) = - \frac{1}{2}\, \int_{X}\,
        w_\varphi^*H_{\sfU(m)} \ \in \ \frac{1}{2}\, \RZ\ .
\end{equation}
Thus we obtain the generalised version of~\cite[Proposition
2]{Gawedzki--2DFKM_and_WZW} developed in~\cite{Freed-Moore--TEM} as

\begin{proposition}
If $(E^-, \hat{\tau}) \in \QVBdl_U(X, \tau)$, the Kane-Mele invariant $\rmK\rmM(E^-, \hat{\tau})$
coincides with the Chern-Simons invariant of any Quaternionic
connection on $E^-$:
\begin{equation}
\rmK\rmM(E^-, \hat{\tau}) = \exp\bigg(2\pi\, \iu\, \int_X\,
\rmC\rmS(A) \bigg) \ .
\end{equation}
\end{proposition}

This applies, in particular, to any Quaternionic vector bundle arising as the Bloch bundle of a $3$-dimensional time-reversal symmetric topological insulator.

\subsection{The Kane-Mele invariant as a split Quaternionic obstruction}
\label{sect:KM as obstruction to factorisation}

We shall now derive our first characterisation of Quaternionic vector bundles with trivial Kane-Mele invariant.
Consider again the sewing matrix $w_\varphi \colon X \to \sfU(m)$ from~\eqref{eq:trivialisations_of_E}.
Since $m \in 2\,\NN$ is even, the map $\det \circ w_\varphi \colon X \to \sfU(1)$ satisfies
\begin{equation}
	\det \circ w_\varphi \circ \tau
	= \det \circ (- w_\varphi^\trans)
	= \det \circ w_\varphi\ ,
\end{equation}
i.e. it is a $\ZZ$-invariant function on $X$.
Therefore this function does not wind around $\sfU(1)$, and consequently we may choose a function $f \colon X \to \FR$ such that
\begin{equation}
\label{eq:split off U1 factor}
	\exp \big( 2\pi\, \iu\, m\, f \big) = \det \circ w_\varphi\ .
\end{equation}
We then obtain a function
\begin{equation}
	\tilde{w}_\varphi \coloneqq \exp( - 2\pi\, \iu \, f)\,
        w_\varphi \colon X \longrightarrow \sfS\sfU(m)\ .
\end{equation}
This has essentially also been observed by~\cite{Gawedzki--2DFKM_and_WZW}.
A small but important new observation is that the maps $w_\varphi$ and $\tilde{w}_\varphi$ are homotopic under the inclusion $\sfS\sfU(m) \hookrightarrow \sfU(m)$.
This can be established by any deformation retract of $\FR$ to $0 \in
\FR$, for any such retraction will induce a homotopy from $f$ to the
constant map $0$. In particular, this implies that the Kane-Mele
invariant can be equivalently written
\begin{equation}
\rmK\rmM(E^-, \hat{\tau}) = \exp \bigg( \pi\, \iu\,
        \int_{X}\, \tilde w_\varphi^*H_{\sfS\sfU(m)} \bigg)
\end{equation}
using the canonical $3$-form $H_{\sfS\sfU(m)}$ on $\sfS\sfU(m)$; abusing notation
slightly, we will usually drop
the tilde and simply write $w_\varphi$ for $\tilde w_\varphi$.

The first few homotopy groups of $\sfS\sfU(m)$ for $m\in\NN$ with $m \geq 3$ read as~\cite[p.~151]{Nakahara}
\begin{equation}
\label{eq:homotopy groups of SU(n)}
	\pi_i(\sfS\sfU(m)) =
	\begin{cases}
		\ 0 \ , & i = 1,2,4\ ,
		\\
		\ \RZ \ , & i = 3\ ,
	\end{cases}
\end{equation}
as well as $\pi_{1}(\sfS\sfU(2)) = \pi_2(\sfS\sfU(2))=0$, $\pi_3(\sfS\sfU(2)) = \RZ$, and $\pi_4(\sfS\sfU(2)) = \ZZ$.
Moreover, since $\sfS\sfU(m)$ is a compact simple and simply-connected
Lie group, or alternatively by Hurewicz' Theorem, we infer that its third integer homology group reads $\rmH_3(\sfS\sfU(m),\RZ) \cong \RZ$.
Using~\eqref{eq:homotopy groups of SU(n)} and the Universal Coefficient Theorem, we infer that $\rmH^3(\sfS\sfU(m),\RZ) \cong 
\RZ$ as well.

A generator of $\rmH^3(\sfS\sfU(m),\RZ)$ corresponds to a homotopy class of maps $\sfu_3 \colon \sfS\sfU(m) \to K(\RZ,3)$ that induce an isomorphism
\begin{equation}
	\sfu_3^* \colon \rmH^3(K(\RZ,3),\RZ) \arisom
        \rmH^3(\sfS\sfU(m),\RZ)\ .
\end{equation}
Recall that there is a diffeomorphism $h \colon S^3 \to \sfS\sfU(2)$.
Consider the inclusion $\iota_{2,m} \colon \sfS\sfU(2) \to
\sfS\sfU(m)$, $u \mapsto u \oplus \One$, and observe that
\begin{equation}
	\iota_{2,m}^* H_{\sfS\sfU(m)} = H_{\sfS\sfU(2)}\ .
\end{equation}
Define homomorphisms of abelian groups
\begin{equation}
\Xi = (-)^*[H_{\sfS\sfU(m)}] \colon \pi_3(\sfS\sfU(m)) = [S^3,
\sfS\sfU(m)] \longrightarrow \rmH^3(S^3,\RZ)\ ,
		\quad [a] \longmapsto a^*[H_{\sfS\sfU(m)}] 
\end{equation}
and
\begin{equation}
\Psi \colon \rmH^3(S^3,\RZ) \longrightarrow \pi_3(\sfS\sfU(m))\ ,
	\quad k\, h^*[H_{\sfS\sfU(2)}] \longmapsto k\, [\iota_{2,m}
        \circ h]\ ,
\end{equation}
where $k \in \RZ$ and $[H_{\sfS\sfU(m)}]$ is the canonical generator of $\rmH^3(\sfS\sfU(m), \RZ) \cong \RZ$ represented by the canonical $3$-form $H_{\sfS\sfU(m)} \in \Omega^3(\sfS\sfU(m))$.
We compute
\begin{equation}
\begin{aligned}
	\Xi \circ \Psi \big( k\, h^*[H_{\sfS\sfU(2)}] \big)
	&= \Xi \big( k\, [\iota_{2,m} \circ h] \big)
	\\[4pt]
	&= k\, (\iota_{2,m} \circ h)^* [H_{\sfS\sfU(m)}]
	\\[4pt]
	&= k\, h^*[H_{\sfS\sfU(2)}]\ .
\end{aligned}
\end{equation}
Thus $\Psi$ is right inverse to $\Xi$, implying that $\Psi$ is injective and $\Xi$ is surjective.
In particular, $\Xi$ is a surjective homomorphism of abelian groups
whose source and target are both canonically isomorphic to $\RZ$, and consequently $\Xi$ is a group isomorphism.

The canonical isomorphism $\Phi_{\sfS\sfU(m)} \colon [\sfS\sfU(m),
K(\RZ,3)] \to \rmH^3(\sfS\sfU(m), \RZ)$ is given by $[f] \mapsto
f^*K_3$, where $K_3 \in \rmH^3(K(\RZ,3), \RZ)$ is the canonical
generator of the third integer cohomology group of $K(\RZ,3)$.
By construction we have $\Phi_{\sfS\sfU(m)}([\sfu_3])= \sfu_3^* K_3 = [H_{\sfS\sfU(m)}]$.
Using the canonical isomorphism $\Phi_{S^3} \colon \pi_3(K(\RZ,3)) \to
\rmH^3(S^3,\RZ)$ given by $\Phi_{S^3}([g]) = g^*K_3$, for $[f] \in \pi_3(\sfS\sfU(m))$ we thus have
\begin{equation}
	\Phi_{S^3} \circ \pi_3(\sfu_3) [f] = \Phi_{S^3} \big( [\sfu_3 \circ f] \big)
	= (\sfu_3 \circ f)^*K_3
	= f^*[H_{\sfS\sfU(m)}] \ .
\end{equation}
That is, there is a commutative diagram of abelian groups
\begin{equation}
\begin{tikzcd}[row sep=1.25cm, column sep=3cm]
	\pi_3(\sfS\sfU(m)) \ar[r, "{(-)^*[H_{\sfS\sfU(m)}]}", "\cong"'] \ar[dr, "{\pi_3(\sfu_3) }"'] & \rmH^3(S^3, \RZ)
	\\
	& \pi_3(K(\RZ, 3)) \ar[u, "\cong", "{\Phi_{S^3}}"']
\end{tikzcd}
\end{equation}
This shows that $\sfu_3$ induces an isomorphism of third
homotopy groups, and we have thus established

\begin{proposition}
\label{st:SUm to KZ3 4-equiv}
For $m \geq 2$, any continuous (and hence any smooth) map $\sfu_3
\colon \sfS\sfU(m) \to K(\RZ,3)$ such that $\Phi_{\sfS\sfU(m)}[\sfu_3]$ is a generator of $\rmH^3(\sfS\sfU(m),\RZ) \cong \RZ$ induces a $4$-equivalence from $\sfS\sfU(m)$ to $K(\RZ,3)$.
\end{proposition}

Combining this with~\cite[Theorem 6.31]{Switzer:AT} we obtain a
heuristic picture: from the point of view of any manifold of
dimension less than four, the group $\sfS\sfU(m)$ looks like the
Eilenberg-MacLane space $K(\RZ,3)$. More precisely, we have

\begin{corollary}
\label{st:3D CW see SUm as KZ3}
If $X$ is a CW-complex of dimension $d \leq 3$ and $m \geq 2$ then
\begin{equation}
	\rmH^3(X,\RZ) \cong [X, K(\RZ,3)] \cong [X, \sfS\sfU(m)]\ .
\end{equation}
\end{corollary}

The analysis in this section so far was previously known, but the thorough investigation of the action of the map $\tau$ in this context allows us to give a novel homotopy-theoretic perspective on the Kane-Mele invariant:
if $(X, \tau)$ is a $3$-dimensional closed connected and orientable involutive
manifold with orienta{-}tion-reversing involution $\tau$, pairing a
cocycle with the fundamental class of $X$ yields an isomorphism
\begin{equation}
	\mbox{$\int_X$} \, \colon \rmH^3(X,\RZ) \longrightarrow
        \RZ\ .
\end{equation}
From this it follows that $\int_X\, \omega \in 2\, \RZ$ if and only if
there exists $[\eta] \in \rmH^3(X,\RZ)$ such that $[\omega] = 2\,
[\eta]$. Combining this with the orientation-reversing involution
we deduce that $\int_X\, \omega \in 2\, \RZ$ if and only if there exists
$[\eta]$ such that $[\omega] = [\eta] - \tau^*[\eta]$.

Writing $m = 2n$, and using the fact that Corollary~\ref{st:3D CW see SUm as KZ3} holds for any $n \geq 2$, we can express the cohomology class $[\eta]$ as the homotopy class of a map $g \colon X \to \sfS\sfU(n)$.
We set
\begin{equation}
	w \coloneqq
	\begin{pmatrix}
		0 & g
		\\
		- (g \circ \tau)^\trans & 0
	\end{pmatrix}\ .
\end{equation}
Let $\kappa \colon \sfU(m) \to \sfU(m)$ be the
involution which sends $w \mapsto -w^\trans$.
Since $m$ is even, the map $\kappa$ preserves the subgroup $\sfS\sfU(m)$.

\begin{lemma}
\phantomsection\label{st:properties of w}
\begin{myenumerate}
	\item The map $w$ is $\sfS\sfU(m)$-valued and satisfies $w \circ \tau = \kappa \circ w$.
	
	\item $w^*[H_{\sfS\sfU(m)}]
		= g^*[H_{\sfS\sfU(n)}] - \tau^* \big( g^*[H_{\sfS\sfU(n)}] \big)
		= [\eta] - \tau^* [\eta]\ $.
\end{myenumerate}
\end{lemma}

\begin{proof}
It is straightforward to see that $w \colon X \to \sfU(m)$. The necessary restriction on the determinant $\det(w)= \det \big( g\, (g \circ \tau)^\trans \big)
	= 1$ follows from the fact that $m=2n$ is even and since $\det(g) = 1$.

The second claim is a consequence of the fact that
\begin{equation}
	w^*H_{\sfS\sfU(m)} = g^*H_{\sfS\sfU(n)} - \tau^* ( g^*H_{\sfS\sfU(n)})
\end{equation}
which can be seen by a direct computation.
\end{proof}

\begin{definition}[Split Quaternionic vector bundle]
A Quaternionic vector bundle $(E^-, \hat{\tau})$ of rank $m = 2n$ on an involutive manifold $(X, \tau)$ is \emph{split Quaternionic} if there exists a Hermitean vector bundle $F \to X$ of rank $n$ with an automorphism $\sigma$ and a unitary isomorphism $\phi \colon E^- \to \tau^*F \oplus F^*$ such that
\begin{equation}
\label{eq:split QVBdl condition}
	(\tau^* \flat_F \oplus \sharp_F) \circ \phi \circ \hat{\tau} \circ \tau^*\phi^{-1} =
	\begin{pmatrix}
		0 & \sigma
		\\
		- \tau^*\sigma^\trans & 0
	\end{pmatrix}\ ,
\end{equation}
where the superscript ${}^\trans$ denotes the fibrewise transpose of a morphism of vector bundles, and where $\flat_F \colon \overline{F} \to F^*$ and $\flat_F \colon \overline{F}^* \to F$ denote the musical isomorphisms of the hermitean vector bundle $F$.
\end{definition}

\begin{theorem}
\label{st:splitting up to homotopy of w_varphi}
Let $(X, \tau)$ be an orientable connected and closed 
$3$-dimensional CW-complex with orientation-reversing involution $\tau$,
and $X^\tau \neq \varnothing$. Let $(E^-, \hat{\tau}) \in \QVBdl_U(X, \tau)$ with $\rank(E^-) = m = 2n$.
The following statements are equivalent:
\begin{myenumerate}
	\item The Kane-Mele invariant of $(E^-, \hat{\tau})$ is trivial: $\rmK\rmM(E^-, \hat{\tau}) = 1$.
	
	\item There exists a trivialisation $\varphi$ of $E^-$ such
          that the induced Quaternionic structure on $X \times \FC^m$
          is homotopic to a split Quaternionic structure on $X \times (\FC^n \oplus
          \FC^n)$: there is a homotopy of $\sfS\sfU(m)$-valued maps
	\begin{equation}
		w_\varphi \simeq w=
		\begin{pmatrix}
			0 & g
			\\
			- (g \circ \tau)^\trans & 0
		\end{pmatrix}
	\end{equation}
	for some map $g \colon X \to \sfS\sfU(n)$.
	
	\item There exists a trivialisation $\varphi'$ of $E^-$ such
          that the induced Quaternionic structure on $X \times \FC^m$
          is homotopic to the trivial Quaternionic structure on $X
          \times (\FC^n \oplus \FC^n)$: there is a homotopy of $\sfS\sfU(m)$-valued maps
	\begin{equation}\label{eq:Quaternionic triv up to ho}
		w_{\varphi'} \simeq w_0=
		\begin{pmatrix}
			0 & \One
			\\
			- \One & 0
		\end{pmatrix}\ .
	\end{equation}
\end{myenumerate}
\end{theorem}

\begin{proof}
The equivalence of (1) and (2) has been established in Corollary~\ref{st:3D CW see SUm as KZ3} and Lemma~\ref{st:properties of w}.
The implication (3)$\Rightarrow$(1) is immediate.
If $\rmK\rmM(E^-, \hat{\tau}) = 1$ then there exists a trivialisation $\varphi$ of $E^-$ such that $\int_X\, w_\varphi^* H_{\sfS\sfU(m)} \in 2\, \RZ$ with $w_\varphi \colon X \to \sfS\sfU(m)$.
From the computation in~\eqref{eq:KM_independent_of_trivialisation}
together with Corollary~\ref{st:3D CW see SUm as KZ3}, we infer that
there exists $a \colon X \to \sfS\sfU(m)$ such that $\int_X\,
w_{\varphi'}^*H_{\sfS\sfU(m)} = 0$ for $\varphi' = a \circ \varphi$.
Thus again by Corollary~\ref{st:3D CW see SUm as KZ3} the map
$w_{\varphi'}$ is nullhomotopic, i.e. it is homotopic to any constant map $X \to \sfS\sfU(m)$, and we may choose this map to be $w_0$.
\end{proof}

The homotopies appearing here are not related to the involution $\tau$.
This reflects the fact that the $\ZZ$-valued Kane-Mele invariant does not detect the complete information about Quaternionic vector bundles on $3$-dimensional CW-complexes.
There is a more refined version of the Kane-Mele invariant, called the
FKMM invariant in~\cite{DeNittis-Gomi:Topological_nature_of_FKMM},
which is sufficient to distinguish Quaternionic vector bundles on involutive CW-complexes of dimension $d \leq 3$.

\begin{remark}
\label{rmk:homotopy interpretation of KM}
From the perspective of topological phases of quantum matter, we see from
Theorem~\ref{st:splitting up to homotopy of w_varphi} that the
Kane-Mele invariant neither detects fully whether a given Quaternionic
vector bundle with trivial underlying Hermitean vector bundle defines a nontrivial Quaternionic K-theory class nor whether it is trivialisable as a Quaternionic vector bundle.
Nevertheless, Theorem~\ref{st:splitting up to homotopy of w_varphi}
gives a precise geometric meaning to the invariant $\rmK\rmM(E^-,\hat\tau)$:
it measures whether there exists a global frame of $E^-$ which, in a
way made precise by~\eqref{eq:Quaternionic triv up to ho}, is a trivialisation of $(E^-, \hat{\tau})$ up to nonequivariant homotopy.

Thus the $3$-dimensional Kane-Mele invariant considered here, as well
as in~\cite{Freed-Moore--TEM,Gawedzki--2DFKM_and_WZW}, is a
considerably weaker invariant of time-reversal symmetric topological
phases than the invariants constructed
in~\cite{DeNittis-Gomi:Topological_nature_of_FKMM,FMP:Z2_invariants_as_geometric_obstructions,FKM}
which in three dimensions conclusively answer the question of whether Quaternionic vector bundles of the above type are Quaternionic trivialisable.
In both cases the invariants are valued in $\ZZ^4$ when the base space
coincides with the $3$-dimensional Brillouin torus $X=X_\Pi=\TT^3$ (containing both the strong $3$-dimensional and all three weak $2$-dimensional Kane-Mele invariants).
The invariant in~\cite{DeNittis-Gomi:Topological_nature_of_FKMM} can even be understood as the Quaternionic first Chern class and, similarly to Proposition~\ref{st:1st_Chern_class_knows_everything_in_dleq3}, this classifies Quaternionic vector bundles in low dimensions.
\qen
\end{remark}

\begin{remark}
The works of De Nittis and Gomi~\cite{DeNittis-Gomi:Classification_of_quaternionic_Bloch_bdls,DeNittis-Gomi:FKMM_in_low_dimension,DeNittis-Gomi:Topological_nature_of_FKMM,DeNittis-Gomi:DGInvariants_Real_case,DeNittis-Gomi:DGInvariants_ld-Quat_case,Gomi:KT_and_TD_for_Real_bundles} encompass a great understanding of the Kane-Mele invariant from several different perspectives, especially from the point of view of equivariant homotopy theory and cohomology.
In particular, in~\cite{DeNittis-Gomi:Classification_of_quaternionic_Bloch_bdls}, they introduce a refinement of the Kane-Mele invariant, the so-called FKMM invariant, that takes values in a $\ZZ$-equivariant cohomology theory rather than in $\ZZ$, but which reduces to the Kane-Mele invariant in certain cases.
It is similar to, but generalises, the approach to the Kane-Mele invariant presented in~\cite{Freed-Moore--TEM}.

In~\cite{DeNittis-Gomi:DGInvariants_ld-Quat_case}, building on their previous articles, De Nittis and Gomi use techniques to investigate Quaternionic Bloch bundles that are very similar to the techniques used in the present paper.
In particular, they use $\ZZ$-equivariant cohomology theory and the Wess-Zumino term to analyse the Kane-Mele invariant.
However, while the tools are similar, the actual methods employed are still different: the homotopy-theoretic approach we present in this section does not use equivariance, but identifies the Kane-Mele invariant as an obstruction in ordinary homotopy theory.
The long-exact sequences in equivariant cohomology from~\cite{DeNittis-Gomi:DGInvariants_ld-Quat_case} (see also~\cite{DeNittis-Gomi:DGInvariants_Real_case,Gomi:KT_and_TD_for_Real_bundles}) are similar to the long exact sequences we develop in Sections~\ref{sect:Localisation in equivar HdR} in that they interrelate different equivariant cohomology theories.
However, our use of the Mayer-Vietoris sequence complements their study of equivariant cohomology in the context of the Kane-Mele invariant, and leads to a previously unknown, geometric explanation for the localisation of the Kane-Mele invariant.
Our Section~\ref{sect:KM from BGrb hol} is the merging of the differential geometric formalism from Section~\ref{sect:Localisation in equivar HdR} with the Wess-Zumino approach to the Kane-Mele invariant; here our use of the theory of gerbe holonomy and 2-category theory to explain the localisation also complements the analysis in the works of De Nittis and Gomi.
\qen
\end{remark}

\section{Localisation of the Kane-Mele invariant}
\label{sect:Localisation in equivar HdR}

\subsection{The Kane-Mele invariant as a localisation formula}
\label{sect:KMfixedpoint}

Next we investigate the interplay between the geometry of the
canonical $3$-form $H_{\sfU(m)}$ on $\sfU(m)$, and the involutions on
the unitary group $\sfU(m)$ and on the Brillouin torus $X_\Pi=\TT^3$.
As before, let $\kappa \colon \sfU(m) \to \sfU(m)$ be the involution given by $\kappa(g) = -g^\trans$.
For differential forms $\omega$ and $\eta$ of
degrees $p$ and $q$, respectively, valued in $m\times m$ matrices with complex coefficients we have 
\begin{equation}
	(\omega \wedge \eta)^\trans = (-1)^{p\,q}\, \eta^\trans \wedge
        \omega^\trans\ .
\end{equation}
Consequently the canonical $3$-form $H_{\sfU(m)}$ on $\sfU(m)$ satisfies
\begin{equation}
	\kappa^*H_{\sfU(m)} = - H_{\sfU(m)}\ ,
\end{equation}
i.e. it is a \emph{Real} 3-form~\cite{Hekmati:2016ikh}, and we have $\kappa \circ w_\varphi = w_\varphi \circ \tau$.

We can now rewrite the Kane-Mele invariant as follows.
Let $X^3 \coloneqq \TT^3$ and let $F^3 \subset \TT^3$ denote one of the halves of the $3$-torus which provides a fundamental domain for the $\ZZ$-action $\tau$.
That is, for $\TT^3 = [-\pi, \pi]^3/{\sim}$ and $\tau$ acting as
inversion through the origin, we can take $F^3 = [-\pi,
\pi]^2/{\sim}\times [0,\pi]$.
Then $\partial F^3 = \overline{\partial (X^3 \backslash F^3)}$, where
overline denotes orientation-reversal, and $\tau \colon F^3 \to X^3 \backslash F^3$ is an orientation-reversing diffeomorphism; that is $\tau(F^3) = \overline{X^3 \backslash F^3}$.
Therefore $(\partial F^3, \tau_{|\partial F^3})$ is an involutive manifold with orientation-preserving involution.
As $F^3 \simeq \TT^2$, there exists $\eta^2 \in \Omega^2(F^3)$ such that $\dd \eta^2 = (w_\varphi^* H_{\sfU(m)})_{|F^3}$.
Over $X^3 \backslash F^3$, we have $(w_\varphi^* H_{\sfU(m)})_{|X^3 \backslash F^3} = - \tau^* \big( (w_\varphi^* H_{\sfU(m)})_{|F^3} \big) = - \tau^* \dd \eta^2$.
Putting everything together we have
\begin{equation}
\begin{aligned}
	\int_{X^3}\, w_\varphi^* H_{\sfU(m)}
	&= \int_{F^3}\, \dd \eta^2 + \int_{X^3 \backslash F^3}\, \big(- \dd\, \tau^*\eta^2\big)
	\\[4pt]
	&= \int_{\partial F^3}\, \eta^2 - \int_{\partial (X^3 \backslash F^3)}\, \tau^*\eta^2
	\\[4pt]
	&= \int_{X^2}\, \big(\eta^2 + \tau^* \eta^2\big)\ ,
\end{aligned}
\end{equation}
where we have set $X^2 \coloneqq \partial F^3$.
The computation of the Kane-Mele invariant has thus been reduced from an integral of a $3$-form over the $3$-dimensional manifold $X^3 = \TT^3 = X_\Pi$ to an integral of a $2$-form $\rho^2 \coloneqq (\eta^2 + \tau^* \eta^2)_{|X^2}$ over the $2$-dimensional involutive manifold $(X^2, \tau_{|X^2})$.
Whereas $w_\varphi^*H_{\sfU(m)}$ was Real for the
orientation-reversing involution $\tau$, the new $2$-form $\rho^2$ is invariant with respect to the orientation-preserving involution $\tau_{|X^2}$.

For a different choice $\eta^{2 \,\prime} = \eta^2 + \beta^2$, we must have $\dd \beta^2 = 0$ in order to satisfy $\dd \eta^{2 \,\prime} = (w_\varphi^* H_{\sfU(m)})_{|F^3}$.
Consequently
\begin{equation}
	\int_{\partial F^3}\, \eta^{2 \,\prime} = \int_{\partial
          F^3}\, \big( \eta^2 + \beta^2\big)
	= \int_{F^3}\, \big( \dd \eta^2 + \dd \beta^2\big)
	= \int_{\partial F^3}\, \eta^2\ ,
\end{equation}
so that
\begin{equation}
	\int_{X^2}\, \big(\eta^2 + \tau^* \eta^2 \big)
	= \int_{X^3}\, w_\varphi^* H_{\sfU(m)}
	= \int_{X^2}\, \big(\eta^{2 \,\prime} + \tau^* \eta^{2
          \,\prime}\, \big)\ .
\end{equation}

We also have $\dd \rho^2 = 0$ by dimensional reasons, so that we are
in a very similar situation to that we started with, only in one dimension lower, and with Reality swapped for invariance and orientation-reversal for orientation-preservation.
We can now iterate the above argument to obtain
\begin{equation}
	\int_{X^3}\, w_\varphi^*H_{\sfU(m)}
	= \int_{X^2}\, \rho^2
	= \int_{X^1}\, \rho^1
	= \int_{X^0}\, \rho^0\ ,
\end{equation}
where $X^d = \partial F^{d+1}$ is the boundary of a fundamental domain of $X^{d+1}$ and $F^{d+1}$ consists of a disjoint union of half tori, one in each connected component of $X^{d+1}$.
Each pair $(X^d,\tau_{|X^d})$ is an involutive manifold, where $\tau^d:=\tau_{|X^d}$ reverses the orientation if $d$ is odd and preserves the orientation if $d$ is even.
Similarly $\tau^{d\,*}\rho^d = (-1)^d\, \rho^d$.
By construction, $X^0 = (X^d)^{\tau^d} = \{ k \in X^d\, |\,
\tau^d(k) = k \}$ is the set of fixed points of the
$\ZZ$-action in every dimension, which has cardinality $2^3=8$.
Hence the Wess-Zumino-Witten action functional
\begin{equation}
	\int_{X^3}\, w_\varphi^*H_{\sfU(m)} = \sum_{k \in (X^3)^{\tau}}\, \rho^0(k)
\end{equation}
is given simply by a sum of the values of $\rho^0$ over the set of
fixed points in the Brillouin zone $X^3$, which are the time-reversal invariant momenta.
We have not yet commented on the possible dependence on the choice of
$\eta^d$ in each dimension, but we will do so in the more careful analysis below.

\subsection{A Mayer-Vietoris Theorem for involutive manifolds}
\label{sect:MV for involutive spaces}

We will now provide a more general topological explanation
for the fixed point formula expressing the Kane-Mele invariant above.
Let $(X^d,\tau^d)$ be an involutive $d$-dimensional manifold and for
$p\in\NN_0$ set
\begin{equation}
\begin{aligned}
	\Omega^p_\pm (X^d) &\coloneqq \big\{ \omega \in \Omega^p(X^d)\, \big|\, \tau^{d\,*}\omega = \pm\, \omega \big\}\ ,
	\\[4pt]
\Omega^p_{\pm\, \cl} (X^d) &\coloneqq \ker \big( \dd \colon
\Omega^p_\pm(X^d) \to \Omega^{p+1}_\pm(X^d) \big) \ , \\[4pt]
	\rmH^p_{\dR\, \pm}(X^d) &\coloneqq \frac{\Omega^p_{\pm\, \cl}
          (X^d) }{\dd \Omega^{p-1}_\pm(X^d)}\ .
\end{aligned}
\end{equation}
We refer to $\rmH^p_{\dR\, +}(X^d)$ as the \emph{$p$-th equivariant de Rham cohomology group} of $X^d$, and to $\rmH^p_{\dR\, -}(X^d)$ as the \emph{$p$-th Real de Rham cohomology group} of $X^d$.

\begin{lemma}
\label{st:HdR_pm_and_equivar_homotopy}
Let $(X^d,\tau^d),\, (Y^{d'},\tau^{\prime\,d'}\,)$ be involutive
manifolds. Let $f,g \in\Mfd_{\ZZ}\big( X^d,Y^{d'}\big)$ be smooth maps
which intertwine the involutions, and let $h \colon [0,1] \times X^d \to Y^{d'}$ be a smooth homotopy from $f$ to $g$ such that $h(t, \tau^d(x)) = \tau^{\prime\,d'} \circ h(t,x)$ for all $t \in [0,1]$ and $x \in X^d$.
Then the induced pullbacks in equivariant and Real de Rham cohomology agree: $f^* = g^* \colon \rmH^p_{\dR\, \pm}(Y^{d'}) \to \rmH^p_{\dR\, \pm}(X^d)$.
\end{lemma}

\begin{proof}
The compatibility of $f$ and $g$ with the involutions implies that
both $f^*$ and $g^*$ define homomorphisms from
$\rmH^p_{\dR\, \pm}(Y^{d'})$ to $\rmH^p_{\dR\,
  \pm}(X^d)$.
Denote by $\jmath_t \colon X^d \hookrightarrow [0,1] \times X^d$, $x \mapsto (t,x)$ the inclusion at $t \in [0,1]$, and let $\omega' \in \Omega^p_\pm(Y^{d'})$.
Let $\iota_Z \omega'$ be the contraction of a vector field $Z$ into the first slot of the differential form $\omega'$.
Then~\cite{Lee--Smooth_manifolds}
\begin{equation}
	g^*\omega' - f^*\omega'
	= \jmath_1^*\, h^* \omega' - \jmath_0^*\, h^* \omega'
	= \dd \bigg( \int_0^1\, \jmath_t^* \big( \iota_{\partial_t}\,
        h^*\omega'\, \big)\ \dd t \bigg)\ .
\end{equation}
Now
\begin{equation}
\begin{aligned}
	\tau^{d\,*} \bigg( \int_0^1\, \jmath_t^* \big( \iota_{\partial_t}\, h^*\omega'\, \big)\ \dd t \bigg)
	&= \int_0^1\, \jmath_t^* (1 \times \tau^{d\,*}) \big( \iota_{\partial_t}\, h^*\omega'\, \big)\ \dd t
	\\[4pt]
	&= \int_0^1\, \jmath_t^* \big( \iota_{\partial_t}\, (1 \times \tau^{d\,*}) h^*\omega'\, \big)\ \dd t
	\\[4pt]
	&= \int_0^1\, \jmath_t^* \big( \iota_{\partial_t}\, h^* \tau^{\prime\,d'\, *} \omega'\, \big)\ \dd t
	\\[4pt]
	&= \pm \, \int_0^1\, \jmath_t^* \big( \iota_{\partial_t}\,
        h^*\omega'\, \big)\ \dd t \ ,
\end{aligned}
\end{equation}
so that $f^*\omega'$ and $g^*\omega'$ differ by an element of $\dd \Omega^{p-1}_\pm (X^d)$.
\end{proof}

Let $F^d \subset X^d$ be a fundamental domain for a nontrivial
$\ZZ$-action $\tau^d$, i.e. $\tau^d$ restricts to a diffeomorphism
$ F^d \to \overline{X^d \backslash F^d}$.

\begin{proposition}
\label{st:HdR SES for fundamental domain}
Set $X^{d-1} = \partial F^d$ and assume that $X^{d-1} \subset X^d$ is an embedded submanifold.
\begin{myenumerate}
	\item There exists an open neighbourhood $U^d$ of $X^{d-1}$ in $X^d$ which is invariant under $\tau^d$ and which retracts $\ZZ$-equivariantly onto $X^{d-1}$.
	
	\item Let $V^d = F^d \cup U^d$.
	There exists a partition of unity on $X^d$ subordinate to the
        cover $X^d = V^d \cup \tau^d(V^d)$ of the form $(f,
        \tau^{d\,*}f)$ for some smooth function $f:V^d\to\FR$.	
	
	\item For any $p \in \NN_0$ there is a short exact sequence of $\FR$-vector spaces
	\begin{equation}
	\label{eq:SES of spaces of forms}
	\begin{tikzcd}[column sep=2cm]
		0 \ar[r] & \Omega^p_\mp(X^d) \ar[r, "\imath_d^*"] &
                \Omega^p(V^d) \ar[r, "(1 \pm \tau^{d\,*}) \,
                \jmath_d^*"] & \Omega^p_\pm(U^d) \ar[r] & 0\ ,
	\end{tikzcd}
	\end{equation}
where $\imath_d \colon V^d \hookrightarrow X^d$ and $\jmath_d \colon U^d \hookrightarrow V^d$ are inclusions of open subsets.
	
	\item There is a short exact sequence of complexes of vector spaces
	\begin{equation}
	\label{eq:SES of complexes of forms}
	\begin{tikzcd}[column sep=2cm]
		0 \ar[r] & \Omega^\bullet_\mp(X^d) \ar[r, "\imath_d^*"]
                & \Omega^\bullet(V^d) \ar[r, "(1 \pm \tau^{d\,*})
                \, \jmath_d^*"] & \Omega^\bullet_\pm(U^d) \ar[r] & 0\ .
	\end{tikzcd}
	\end{equation}
\end{myenumerate}
\end{proposition}

\begin{proof}
(1)
Consider an outward-pointing normal vector field $\check{N}^{d-1}_+$ on $X^{d-1} \subset F^d$.
Since $X^{d-1}$ is a closed submanifold of $X^d$, we can find an extension $N^{d-1}_+$ of $\check{N}^{d-1}_+$ to a smooth vector field on $X^d$.
Because it is smooth, $N^{d-1}_+$ is nowhere-vanishing on an open neighbourhood of $X^{d-1}$ in $X^d$.
We can make it $\ZZ$-invariant by setting $N^{d-1} = \frac{1}{2}\, (N^{d-1}_+ + \tau^d_*N^{d-1}_+)$.
Now there exists $\epsilon > 0$ such that the flow $\Phi^{N^{d-1}}$ of $N^{d-1}$ is defined on $(-\epsilon, \epsilon) \times X^{d-1}$ and invertible, i.e. a diffeomorphism onto its image
\begin{equation}
	U^{d} \coloneqq \Phi^{N^{d-1}} \big( (-\epsilon, \epsilon)
        \times X^{d-1} \big)\ .
\end{equation}
Since $N^{d-1}$ is an invariant vector field, we have $\Phi^{N^{d-1}}_t \circ \tau^d = (1_{(-\epsilon, \epsilon)} \times \tau^d) \circ \Phi^{N^{d-1}}_{-t}$, thus making $\Phi^{N^{d-1}} \colon (-\epsilon, \epsilon) \times X^{d-1} \to U^d$ a diffeomorphism of involutive manifolds if $(-\epsilon, \epsilon) \times X^{d-1}$ is endowed with the involution $(\inv \times \tau^{d-1}) (t,k) = (-t,\tau^d(k))$.
The claimed retraction is given by composing the canonical retraction of $(-\epsilon, \epsilon) \times X^{d-1}$ onto $X^{d-1}$ with $\Phi^{N^{d-1}}$.

(2)
There exists a partition of unity $(g_0, g_1)$ subordinate to $X^d = V^d \cup \tau^d(V^d)$, as this is locally finite, with $g_0$ defined on $V^d$ and $g_1$ on $\tau^d(V^d)$.
We set $f = \frac{1}{2}\, (g_0 + \tau^{d\,*}g_1)$.

(3)
For exactness at $\Omega^p_\mp(X^d)$, if $\omega \in \Omega^p_\mp(X^d)$ with $\imath_d^*\,\omega = \omega_{|V^d} = 0$, then $\omega_{|\tau^d(V^d)} = \mp\, \tau^{d\,*}(\omega_{|V^d}) = 0$ and hence $\omega = 0$.

Next we check exactness at $\Omega^p(V^d)$.
For $\eta \in \Omega^p(V^d)$ such that $\eta_{|U^d} \pm (\tau^{d\,*}
\eta)_{|U^d} = 0$, consider the form $\omega$ given by $\omega_{|V^d} = \eta$ and $\omega_{|\tau^d(V^d)} = \mp\, \tau^{d\,*}\eta$.
This is well-defined since $\eta_{|U^d} - (\mp\, \tau^{d\,*} \eta)_{|U^d} = 0$, and moreover $\omega \in \Omega^p_\mp(X^d)$ by construction with $\imath_d^*\,\omega = \eta$.

For exactness at $\Omega^p_\pm(U^d)$,
if $\rho \in \Omega^p_\pm(U^d)$ and $(f, \tau^{d\,*}f)$ is a partition of unity as in statement~(2), set
\begin{equation}
\label{eq:def Delta^p_0}
	(\Delta_0^p\, \rho)_{|k} =
	\begin{cases}
		\ (\tau^{d\,*}f)(k)\, \rho_{|k} \ , & k \in U^d\ ,\\
		\ 0 \ , & k \in V^d \backslash \supp(\tau^{d\,*}f)\ .
	\end{cases}
\end{equation}
This defines a smooth form $\Delta^p_0\, \rho \in \Omega^p(V^d)$.
Then
\begin{equation}
	(\Delta_0^p\, \rho)_{|U^d} \pm (\tau^{d\,*} \Delta_0^p\, \rho)_{|U^d}
	= (\tau^{d\,*}f)\, \rho \pm f\, \tau^{d\,*}\rho
	= (\tau^{d\,*}f + f)\, \rho
	= \rho\ ,
\end{equation}
thus proving surjectivity of $(1 \pm \tau^{d\,*}) \, \jmath_d^*$.

(4) 
The maps in the short exact sequence~\eqref{eq:SES of
            spaces of forms} are combinations of
pullbacks and thus compatible with the exterior derivative $\dd$.
\end{proof}

\begin{remark}
\label{rmk:HdR SES for open subset}
The exact same line of argument applies to any open subset $V^d \subset X^d$ and $U^d = V^d \cap \tau^d(V^d)$ such that $X^d = V^d \cup \tau^d(V^d)$.
That is, there exists a partition of unity subordinate to the cover
$(V^d,\tau^d(V^d))$ of the form $(f, \tau^{d\,*}f)$, and there is a
short exact sequence of cochain complexes of $\FR$-vector spaces of
the same form as \eqref{eq:SES of complexes of forms}.
\qen
\end{remark}

Applying the Zig-Zag Lemma to the short exact sequence of cochain complexes~\eqref{eq:SES of complexes of forms}, we immediately obtain

\begin{theorem}
\label{st:LES for open subset}
If $(X^d,\tau^d)$ is an involutive manifold, with $V^d \subset X^d$ an open subset such that $X^d = V^d \cup \tau^d(V^d)$ and $U^d = V^d \cap \tau^d(V^d)$,
then there is a long exact sequence of cohomology groups
\begin{equation}
\label{eq:LES for open subset}
\begin{tikzcd}[column sep=2cm, row sep=1cm]
	\cdots \ar[r] & \rmH^{p-1}_{\dR\, \mp}(X^d) \ar[r, "\imath_d^*"] & \rmH_\dR^{p-1}(V^d)\ar[r, "(1 \pm \tau^{d\,*})\,\jmath_d^*"] \ar[d, phantom, ""{coordinate, name=Z}] & \rmH^{p-1}_{\dR\, \pm}(U^d)
	\ar[dll, "\Delta^{p-1}"' pos=1, rounded corners, to path={-- ([xshift=2ex]\tikztostart.east) |- (Z) \tikztonodes -| ([xshift=-2ex]\tikztotarget.west) -- (\tikztotarget)}]
	\\
	 & \rmH^p_{\dR\, \mp}(X^d) \ar[r, "\imath_d^*"] & \rmH_\dR^p(V^d)\ar[r, "(1 \pm \tau^{d\,*})\,\jmath_d^*"] & \rmH^p_{\dR\, \pm}(U^d) \ar[r] & \cdots
\end{tikzcd}
\end{equation}
\end{theorem}

\begin{remark}
Theorem~\ref{st:LES for open subset} could as well be derived starting from the short exact sequence of complexes of sheaves of abelian groups
\begin{equation}
	\xymatrix{
		0 \ar@{->}[r] & \CA^\bullet_{\cup\, \mp} \ar@{->}[r] &
                \CA^\bullet \ar@{->}[r] & \CA^\bullet_{\cap\, \pm}
                \ar@{->}[r] & 0\ ,
	}
\end{equation}
where for open subsets $U \subset X^d$ we define $\CA^\bullet_{\cup\, \mp}(U) =
\Omega^\bullet_\mp(U \cup \tau^d(U))$, together with $\CA^\bullet(U) = \Omega^\bullet(U)$ and $\CA^\bullet_{\cap\, \pm}(U) = \Omega^\bullet_\pm(U \cap \tau^d(U))$, with the morphisms as in~\eqref{eq:SES of complexes of forms} promoted accordingly to morphisms of sheaves.
\qen
\end{remark}

\begin{corollary}
\label{st:LES for fdom}
If $(X^d,\tau^d)$ is an involutive manifold with a fundamental domain
$F^d$, if we set $(X^{d-1},\tau^{d-1}) := (\partial F^d,\tau^d_{|\partial
  F^d})$, and if $U^d,V^d\subset X^d$ are chosen as in
Proposition~\ref{st:HdR SES for fundamental domain}, then there is a long exact sequence of cohomology groups
\begin{equation}
\label{eq:LES for fundamental domain}
\begin{tikzcd}[column sep=1.25cm, row sep=1cm]
	\cdots \ar[r] & \rmH^{p-1}_{\dR\, \mp}(X^d) \ar[r, "\imath_d^*"] & \rmH_\dR^{p-1}(V^d)\ar[r, "\psi_d^{p-1}"] \ar[d, phantom, ""{coordinate, name=Z}] & \rmH^{p-1}_{\dR\, \pm}(X^{d-1})
	\ar[dll, "\Delta^{d,p-1}"' pos=1, rounded corners, to path={-- ([xshift=2ex]\tikztostart.east) |- (Z) \tikztonodes -| ([xshift=-2ex]\tikztotarget.west) -- (\tikztotarget)}]
	\\
	 & \rmH^p_{\dR\, \mp}(X^d) \ar[r, "\imath_d^*"] & \rmH_\dR^p(V^d)\ar[r, "\psi_d^p"] & \rmH^p_{\dR\, \pm}(X^{d-1}) \ar[r] & \cdots
\end{tikzcd}
\end{equation}
\end{corollary}

\begin{proof}
Lemma~\ref{st:HdR_pm_and_equivar_homotopy} shows that the pullback along the collapsing map $p_d \colon U^d \to X^{d-1}$ induces an isomorphism $p_d^* \colon \rmH^\bullet_{\dR\, \pm}(X^{d-1}) \to \rmH^\bullet_{\dR\, \pm}(U^d)$.
Its inverse is the pullback along the inclusion $\imath_{X^{d-1}} \colon X^{d-1} \hookrightarrow U^d$.
The homomorphisms in the long exact sequence~\eqref{eq:LES for fundamental domain} are defined via the commutative diagram
\begin{equation}
\label{eq:def of structure maps in fdom LES}
\begin{tikzcd}[column sep=2cm]
	\cdots\ \rmH^{p-1}_{\dR\, \mp}(X^d) \ar[r, "\imath_d^*"] \ar[d, shift left=0.2cm, equal] & \rmH_\dR^{p-1}(V^d)\ar[r, "(1 \pm \tau^{d\,*})\,\jmath_d^*"] \ar[d, equal] & \rmH^{p-1}_{\dR\, \pm}(U^d) \ar[r, "\Delta^{p-1}"] \ar[d, shift left=0.1cm, "\imath_{X^{d-1}}^*"] & \rmH^p_{\dR\, \mp}(X^d) \ar[d, shift right=0.2cm, equal] \ \cdots
	\\
	\cdots\ \rmH^{p-1}_{\dR\, \mp}(X^d) \ar[r, "\imath_d^*"] & \rmH_\dR^{p-1}(V^d)\ar[r, "\psi_d^{p-1}"] & \rmH^{p-1}_{\dR\, \pm}(X^{d-1}) \ar[r, "\Delta^{d,p-1}"] \ar[u, shift left=0.1cm, "p_d^*"] & \rmH^p_{\dR\, \mp}(X^d) \ \cdots
\end{tikzcd}
\end{equation}
which yields the sequence.
\end{proof}

As an application of this Mayer-Vietoris Theorem that we will use below, let us assume that $X^d$ is a closed oriented manifold, i.e. it is compact, oriented and $\partial X^d = \varnothing$.
Suppose that $\omega\in\Omega_\mp^d(X^d)$ is a top form with $\imath_d^*\,[\omega]=0$.
Then by Theorem~\ref{st:LES for open subset} there exists $\rho \in \Omega^{d-1}_{\pm\, \cl}(U^d)$ such that $[\omega] = \Delta^{d-1}[\rho]$.
We will now determine a choice of such closed form $\rho$ explicitly.

\begin{proposition}
\label{st:integral and connhom for V}
Let $(X^d, \tau^d)$ be an involutive manifold with a fundamental
domain $F^d$ such that $X^d$ is closed and
oriented, and let $V^d \subset X^d$ be open such that $X^d = V^d \cup \tau^d(V^d)$.
Set $U^d = V^d \cap \tau^d(V^d)$.
Suppose that $U^d$ is a $\ZZ$-equivariant deformation retract onto an invariant $(d-1)$-dimensional submanifold $\partial F^d \eqqcolon X^{d-1} \subset U^d$ with orientation induced from $\partial V^d$.
Then any class $[\omega] \in \ker(\imath_d^*)\subset\rmH^d_{\dR\, \mp}(X^d)$ has a preimage $[\rho] \in \rmH^{d-1}_{\dR\, \pm}(X^{d-1})$ under the connecting homomorphism $\Delta^{d,d-1}$ satisfying
\begin{equation}
	\int_{X^d}\, \omega = \int_{X^{d-1}}\, \rho\ .
\end{equation}
\end{proposition}

\begin{proof}
Let $\omega\in\Omega^d_\mp(X^d)$ with $\omega_{|V^d} = \dd \eta$ for
some $\eta\in\Omega^{d-1}(V^d)$, and set $\rho' = \eta_{|U^d} \pm
\tau^{d\,*} \eta_{|U^d}$. Then $\rho'\in\Omega^{d-1}_{\pm\, \cl}(U^d)$,
and the explicit form of the connecting homomorphism $\Delta^{d-1}$ in
the long exact sequence \eqref{eq:LES for open subset} obtained from
the Zig-Zag Lemma shows that $\Delta^{d-1} [\rho'\,] = [\omega]$.
Using $\omega_{|V^d}=\dd\eta$ and
$\omega_{|\tau^d(V^d)}=\mp\,\dd(\tau^{d\,*}\eta)$, and keeping
careful track of induced orientations, by an application of Stokes' Theorem  we then have
\begin{equation}
\begin{aligned}
	\int_{X^{d}} \, \omega &= \int_{F^d}\, \omega + \int_{X^d \backslash F^d}\, \omega
	\\[4pt]
	&= \int_{F^d}\, \dd \eta + \int_{X^d \backslash F^d}\, \mp\, \dd \tau^{d\,*} \eta
	\\[4pt]
	&= \int_{X^{d-1}}\, \big(\eta \pm \tau^{d\,*} \eta\big)
	\\[4pt]
	&= \int_{X^{d-1}}\, \rho'\ .
\end{aligned}
\end{equation}
Using the long exact sequence~\eqref{eq:LES for fundamental domain} and the definition of its structure maps from~\eqref{eq:def of structure maps in fdom LES}, we infer that $\imath_{X^{d-1}}^*[\rho'\,] = [\imath_{X^{d-1}}^* \rho'\,]$ satisfies
\begin{equation}
	\Delta^{d,d-1} \big( \imath_{X^{d-1}}^*[\rho'\,] \big) = \Delta^{d-1} [\rho'\,] = [\omega] \ .
\end{equation}
The result then follows by setting $\rho=\imath_{X^{d-1}}^*\rho'$.
\end{proof}

\begin{corollary}
\label{st:preimages under Delta and integrals}
Let $(X^d,\tau^d)$ be a closed and oriented involutive manifold with a fundamental domain $F^d$ and define $(X^{d-1},\tau^{d-1}) = (\partial F^d, \tau^d_{|\partial F^d})$.
Let $U^d$ be a small thickening of $X^{d-1}$ as in Proposition~\ref{st:HdR SES for fundamental domain}.
If $[\omega] \in \rmH^d_{\dR\, \mp}(X^d)$ is in $\ker(\imath_d^*)$
for $V^d = F^d \cup U^d$, then there exists a class $[\rho] \in \rmH^{d-1}_{\dR\, \pm}(X^{d-1})$ with $\Delta^{d, d-1} [\rho] = [\omega]$.
Any such preimage $[\rho]$ of $[\omega]$ under $\Delta^{d,d-1}$ satisfies
\begin{equation}
	\int_{X^{d-1}}\, \rho = \int_{X^d}\, \omega\ .
\end{equation}
\end{corollary}

\begin{proof}
The first assertion follows from Proposition~\ref{st:integral and connhom for V} together with Corollary~\ref{st:LES for fdom}.
For the second assertion, it follows from the long exact sequence~\eqref{eq:LES for fundamental domain} that there exists a $(d{-}1)$-form $\alpha_\pm \in \Omega_\pm^{d-1}(X^{d-1})$ and a form $\eta \in \Omega^{d-1}_\cl(V^d)$ such that
\begin{equation}
	\rho' - \rho = \eta_{|X^{d-1}} \pm \tau^{d\,*} \eta_{|X^{d-1}} + \dd \alpha_\pm \ .
\end{equation}
Consequently
\begin{equation}
\begin{aligned}
	\int_{X^{d-1}}\, \rho' &= \int_{X^{d-1}}\, \big(\rho + (\eta
        \pm \tau^{d\,*} \eta) + \dd \alpha_\pm \big)
	\\[4pt]
	&= \int_{X^{d-1}}\, \rho + \int_{F^d}\, \dd \eta \pm \ori(\tau^d)\, \int_{\tau^d(F^d)}\, \dd \eta
	\\[4pt]
	&= \int_{X^{d-1}}\, \rho
\end{aligned}
\end{equation}
since $X^{d-1} = \partial F^d$ and $\eta$ is closed.
\end{proof}

\subsection{A Localisation Theorem for filtered involutive manifolds}
\label{sect:Localisation via MV}

For a continuous map $f \colon X \to Y$ of topological spaces, we write $\pi_0 (f) \colon \pi_0(X) \to \pi_0(Y)$ for the homomorphism induced on connected components.

\begin{definition}[Componentwise involutive triple]
A \emph{componentwise involutive triple $(X^d, \tau^d, F^d)$} consists of a $d$-dimensional involutive manifold $(X^d, \tau^d)$ and a choice $F^d \subset X^d$ of a fundamental domain for $\tau^d$ such that:
\begin{myenumerate}
	\item $\pi_0 (\tau^d) = 1_{\pi_0 (X^d)}$, i.e. $\tau^d$ preserves connected components of $X^d$.
	
	\item $(X^{d-1}, \tau^{d-1}) \coloneqq (\partial F^d, \tau^d_{|\partial F^d})$ is a nonempty involutive manifold.
	
	\item $\pi_0 (\tau^{d-1}) = 1_{\pi_0(X^{d-1})}$,
          i.e. $\tau^d_{|\partial F^d}$ preserves connected components
          of $\partial F^d$.
\end{myenumerate}
\end{definition}

\begin{definition}[Filtered involutive manifold]
\label{def:filtered involutive manifold}
An \emph{$n$-dimensional filtered involutive manifold} is a sequence
$\{(X^d, \tau^d, F^d)\}_{d = 0,1, \ldots, n}$ of componentwise involutive
triples such that $X^n$ is a closed manifold, $X^{d-1} = \partial F^d$
and $\tau^{d-1} = \tau^d_{|\partial F^d}$ for all $d =2, \ldots, n$,
and $X^0 = \partial F^1$ with $\tau^0 = \tau_{|X^0} = 1_{X^0}$.
\\
An $n$-dimensional filtered involutive manifold is \emph{oriented} if
$X^n$ is oriented, all fundamental domains $F^d$ are endowed with the induced orientations from their ambient spaces, $X^d$ are endowed with the orientations inherited from $F^d$ and outward-pointing normal vector fields, and $\ori(\tau^d) = (-1)^d$ for all $d =0,1, \ldots, n$.
\end{definition}

By definition $X^0 = (X^d)^{\tau^d} \neq \varnothing$, for all $d =0,1,\ldots,n$, is the set of fixed points of all involutions $\tau^d$, which we require to be nonempty.

\begin{example}
\label{eg:T3 as filtered involutive manifold}
Consider the componentwise involutive triple $(\TT^n, \tau^n, F^n)$:
we take the $n$-torus $\TT^n = [-\pi,\pi]^n/{\sim}$ with $\tau^n (k^0,k^1, \ldots, k^{n-1}) = (-k^0,-k^1, \ldots, -k^{n-1})$ to obtain an involutive manifold of dimension $n$.
Its involution $\tau^n$ preserves the orientation if $n$ is even and reverses the orientation if $n$ is odd.
A fundamental domain is given by $F^n = [-\pi,\pi]^{n-1}/{\sim}\times [0,\pi]$, which yields $X^{n-1} = \partial F^n \cong \TT^{n-1} \sqcup \TT^{n-1}$.
Then $\tau^{n-1}$ acts on each of the two copies of
$\TT^{n-1}$ via $\tau^{n-1}(k^0, k^1, \ldots, k^{n-2}) = (-k^0,
-k^1, \ldots, -k^{n-2})$.
In this way we obtain the $n$-dimensional oriented filtered involutive manifold with
\begin{equation}
	(X^d, \tau^d, F^d) = \bigsqcup_{i = 1}^{2^{n-d}}\, \big(
        \TT^d,\, (k^0,k^1, \ldots, k^{d-1}) \mapsto (-k^0,-k^1,
        \ldots, -k^{d-1}),\, \TT^{d-1} \times[0,\pi] \big)
\end{equation}
for $1 \leq d \leq n$, and with $X^0$ consisting of the $2^n$ distinct
fixed points of $\tau^n$ on $\TT^n$ where $k^i=0,\pi$ for $i=0,1,\dots,n-1$.
\qen
\end{example}

\begin{example}
\label{eg:S3 as filtered involutive manifold}
Consider the standard round unit $n$-sphere $S^n \subset \FR^{n+1}$ with the involution $\tau^n$ induced by the restriction of
\begin{equation}
	\tau^{1,n} \colon \FR^{n+1} \longrightarrow \FR^{n+1}\ ,
	\quad
	\tau^{1,n}(k^0, k^1, \dots,k^n) = (k^0,-k^1,\dots,-k^n) \ .
\end{equation}
That is, $\tau^n$ is the inversion through the $k^0$-axis in $\FR^{n+1}$.
If $N \coloneqq (1, 0,\dots,0)$ is the north pole and $S:=(-1,0,\dots,0)$ the south pole of $S^n$, the derivative $\dd\tau^n$ induces the inversion through the origin in the tangent space $T_N S^n$ to $S^n$ at the north pole.
Hence $\tau^n$ preverses the orientation of $S^n$ if $n$ is even and reverses the orientation if $n$ is odd. A fundamental domain of the involutive manifold $(S^n, \tau^n)$ is given by the hemisphere
\begin{equation}
	F^n = \big\{ k \in S^n\ \big| \ k^1 \geq 0 \big\}\ .
\end{equation}
This manifold is diffeomorphic to an $n$-disk, so its boundary
\begin{equation}
	X^{n-1} = \partial F^n = \big\{ k \in S^n\ \big| \ k^1 = 0 \big\} \cong S^{n-1}
\end{equation}
is an $(n-1)$-sphere in $\FR^{n+1}$.
The involution $\tau^n$ restricts to an involution $\tau^{n-1}$ on $S^{n-1}$, which is again given by the inversion through the $k^0$-axis. Iterating these choices gives an $n$-dimensional oriented filtered involutive manifold $\{(X^d,\tau^d,F^d)\}_{d=0,1,\dots,n}$ with 
\begin{equation}
	X^d = \partial F^{d+1}= \big\{ k \in S^n\ \big| \ k^1 = \cdots= k^{n-d} = 0 \big\} \cong S^d \ ,
\end{equation}
together with $\tau^d$ the inversion through the $k^0$-axis and $F^d$
the hemisphere as described, for $1\leq d\leq n-1$, and
\begin{equation}
	X^0 = \big\{ S=(-1,0,\dots,0)\ ,\ N=(1,0,\dots,0) \big\}\ ,
\end{equation}
which is precisely the set of fixed points of the $\ZZ$-action $\tau^n$ on $S^n$.
\qen
\end{example}

Given an oriented filtered involutive manifold $\{(X^d, \tau^d, F^d)\}_{d =0, 1, \ldots, n}$, Corollary~\ref{st:preimages under Delta and integrals} yields a commutative diagram
\begin{equation}
\begin{tikzcd}
	\ker(\imath_n^*) \ar[rr, "\cong"', "\varphi_n"] \ar[dr,
        "\int_{X^n}"'] & & {\displaystyle \frac{\rmH^{n-1}_{\dR\, (-1)^{n-1}}(X^{n-1})}{\psi_{n}^{n-1}\, \rmH^{n-1}_\dR(V^n)} \ar[dl, "\int_{X^{n-1}}"]}
	\\
	& \FR &
\end{tikzcd}
\end{equation}
where $\psi_{n}^{n-1}$ is defined in~\eqref{eq:def of structure maps in fdom LES}.
If $\rmH^{n-1}_\dR(V^n) = 0$ we can apply the same reasoning a second
time, thus giving an $(n-2)$-form $\rho^{n-2} \in \Omega_{(-1)^n\, \cl}^{n-2}(X^{n-2})$ such that 
\begin{equation}
	\int_{X^{n-2}}\, \rho^{n-2} = \int_{X^n}\, \omega\ ,
\end{equation}
if $\omega \in \Omega^n_{(-1)^n\, \cl}(X^n)$ represents the class that we started with in dimension $n$.
Hence we deduce 

\begin{theorem}
\label{st:localisation in Z2-HdR}
Let $\{(X^d, \tau^d, F^d)\}_{d =0, 1, \ldots, n}$ be an oriented filtered
involutive manifold such that $\rmH_\dR^{d-1}(V^d) = 0$ with $V^d$
chosen as in Proposition~\ref{st:HdR SES for fundamental domain} for all $d = 1, \ldots, n$.
Then
\begin{equation}
	\rmH^n_{\dR\, (-1)^n}(X^n) \cong \rmH^n_\dR(X^n)\ ,
\end{equation}
and the integral of any $n$-form over $X^n$ can be
localised to a sum over the fixed point set $X^0 = (X^n)^{\tau^n}$.
\end{theorem}

\begin{proof}
Let $X^n = \bigsqcup_{c \in \pi_0(X^n)}\, X^n_c$ denote the decomposition of $X^n$ into its connected components.
Since $X^n$ is closed and oriented, its top degree de Rham cohomology reads
\begin{equation}
\label{eq:HdR for closed oriented mfd}
	\rmH^n_\dR(X^n) \cong \bigoplus_{c \in \pi_0(X^n)}\, \rmH^n_\dR(X^n_c)
	\cong \FR^{|\pi_0(X^n)|} \ .
\end{equation}
If we denote the inclusion of $X^n_c$ into $X^n$ by $\jmath_{X^n_c} \colon X^n_c \hookrightarrow X^n$, and let $(e_c)_{c \in \pi_0(X^n)}$ be the standard basis of $\FR^{|\pi_0(X^n)|}$, the isomorphism~\eqref{eq:HdR for closed oriented mfd} is given by
\begin{equation}
	[\omega] \longmapsto \sum_{c \in \pi_0(X^n)}\ \int_{X^n_c}\,
        \jmath_{X^n_c}^*\,\omega\ e_c\ .
\end{equation}
Since $\tau^n$ preserves connected components,
Example~\ref{eg:top-degree coho and Z2 actions} from
Appendix~\ref{app:HdR and group actions} below implies $\FR \cong \rmH^n_\dR(X^n_c) \cong \rmH^n_{\dR\, (-1)^n}(X^n_c)$ for all $c \in \pi_0(X^n)$.
The first statement then follows from the fact that $\rmH^\bullet_{\dR\, \pm}$ maps disjoint unions of invariant manifolds to direct sums of cohomology groups.

The map $\rmH^n_{\dR\, (-1)^n}(X^n) \rightarrow \rmH^n_\dR(X^n)$ is therefore a bijection in this case.
Thus given $[\omega] \in \rmH^n_\dR(X^n)$ we find $[\omega'\,] \in
\rmH^n_{\dR\, (-1)^n}(X^n)$ with the same de~Rham class, and hence the
same integral over $X^n$.
To this we can apply the localisation technique described above, i.e. by successively choosing preimages under the connecting homomorphisms
$\Delta^{d,d-1}$ in the long exact sequence~\eqref{eq:LES for
  fundamental domain}.
\end{proof}

\begin{remark}
Given an $n$-form on $X^n$ with the appropriate transformation
behaviour under $\tau^n$, any successive choices of fundamental
domains $F^{n-k}$ and preimage $(n-k)$-forms under the connecting
homomorphisms down to the level of fixed points $X^0 \subset X^n$ will
ultimately yield a set of real numbers $\rho^0(k)\in\FR$, $k \in X^0$, whose sum is by Theorem~\ref{st:localisation in Z2-HdR} the desired integral of the original $n$-form. For this, we
have to choose a sequence of forms $(\rho^n, \ldots,\rho^1, \rho^0)$,
where $\rho^d \in \Omega^d_{(-1)^d\, \cl}(X^d)$.
These forms are related to each other in the sense that $\Delta^{d,d-1}[\rho^{d-1}]= [\rho^d]$.
The localisation formula then reads
\begin{equation}
	\int_{X^n}\, \rho^n = \sum_{k\in X^0}\, \rho^0(k) \ .
\end{equation}
This is in some sense similar to the more familiar case of an
equivariantly closed differential form in the Cartan model for the
$\sfU(1)$-equivariant cohomology of a manifold $M$ with a smooth
$\sfU(1)$-action (see for instance \cite{Szabo:2000fs}). In this case,
unlike the situation in
Appendix~\ref{app:HdR and group actions} below, the equivariant exterior
derivative is $\dd + \iota_\xi$ where $\iota_\xi$ is the contraction with the fundamental
vector field $\xi$ generating the induced $\sfU(1)$-action on $\Omega^\bullet(M)$.
That is, a $\sfU(1)$-equivariantly closed differential form is a sequence $(\alpha^n, \ldots,\alpha^1, \alpha^0)$, $\alpha^d \in \Omega^d(M)$, satisfying $\dd \alpha^{d-1} + \iota_\xi \alpha^{d+1} = 0$ for all $d = 1, \ldots, n-1$.
As in our localisation formalism, if a given form $\alpha^n$ can be
completed to a $\sfU(1)$-equivariantly closed form then this can usually
be done in many ways, i.e. there are choices involved in picking the
forms $\alpha^d$.
\qen
\end{remark}

Setting $n=3$, we can now apply this localisation formalism to the Kane-Mele invariant in
certain cases. For a Bloch bundle $(E^-,\hat\tau^3)\in\QVBdl_U(X^3,\tau^3)$ over the Brillouin zone $X^3$ of a $3$-dimensional time-reversal symmetric topological
phase of type~AII band insulators, under suitable conditions we thereby
derive
\begin{equation}
\label{eq:Localised HdR form of KM}
	\rmK\rmM\big(E^-, \hat{\tau}^3\big)
	= \exp \bigg( \pi\, \iu\, \int_{X^3}\, w_\varphi^*H_{\sfU(m)} \bigg)
	= \prod_{k \in (X^3)^\tau}\, \exp \big( \pi\, \iu\, \rho^0(k) \big)\ .
\end{equation}

\begin{example}
Electrons interacting with a crystalline structure are described
in terms of Bloch theory with crystal momenta defined in the periodic
Brillouin zone $X^3=\TT^3$, as in
Section~\ref{sect:pure_time_reversal_in_TEM}. In this case we are in
the situation of Section~\ref{sect:KMfixedpoint} and
Example~\ref{eg:T3 as filtered involutive manifold} with $n=3$. Write
the class of the Wess-Zumino-Witten $3$-form in Real de~Rham cohomology as
$[w_\varphi^*H_{\sfU(m)}] = \frac{\upsilon}{(2\pi)^3}\, [\dd k^0\wedge\dd
k^1\wedge\dd k^2]$ where $\upsilon\in\RZ$, so that
$[w_\varphi^*H_{\sfU(m)}]=\upsilon$ in $\rmH_{\dR}^3(\TT^3)\cong\FR$. With
the effective Brillouin zone $F^3=\TT^2\times[0,\pi]\simeq\TT^2$ and $V^3 = \TT^2 \times (-\epsilon, \pi+\epsilon)$ for some $0< \epsilon <\frac{\pi}{2}$, we can
choose $\eta^2=\frac \upsilon{(2\pi)^3}\, k^2\, \dd k^0\wedge\dd k^1$ with
$\dd\eta^2=\big(w_\varphi^*H_{\sfU(m)}\big)_{|V^3}$. With $X^2=\partial
F^3= \TT^2_0\sqcup\TT^2_\pi$, where $\TT^2_\theta:=\TT^2\times\{\theta\}$ for $\theta=0,\pi$, we then find that
$\rho^2=(\eta^2+\tau^{3\,*}\eta^2)_{|X^2}$ is given by
$\rho^2_{|{\TT^2_0}}=0$ and $\rho^2_{|{\TT^2_\pi}} = \frac{\upsilon}{4\pi^2}\, \dd k^0\wedge\dd
k^1$; the integrals $\int_{\TT^2_\theta}\,
\rho^2=\frac{\theta\,\upsilon}\pi$ for $\theta=0,\pi$
exponentiate to the weak Kane-Mele invariants on $(\TT^3,\tau^3)$
which multiply to give the strong invariant $\rmK\rmM(E^-,\hat\tau^3)$. Iterating this construction twice more, we
arrive at the real numbers $\rho^0(k)$ for each of the eight
time-reversal invariant periodic momenta $k\in(\TT^3)^{\tau^3}$, where
$\rho^0(k) = 0$ whenever $k^i = 0$ for at least one $i \in \{0,1,2\}$
and $\rho^0(\pi,\pi,\pi) = \upsilon$; this calculation therefore
explicitly mimicks the interpretation of the parity of the integer
$\upsilon$ as the number of unpaired Majorana zero modes.
Then \eqref{eq:Localised HdR form of KM} yields $\rmK\rmM\big(E^-,
\hat{\tau}^3\big)=(-1)^\upsilon$ as required.
\qen
\end{example}

\begin{example}
We can also apply our general formalism to the continuum theory of
a system of gapped free fermions (invariant under continuous
translations by $\FR^3$) whose wavefunctions (after Fourier transformation)
are parameterised by momenta in the Brillouin zone $X^3=S^3$. 
In this case we are in the situation of Example~\ref{eg:S3 as filtered
  involutive manifold} with $n=3$.
First, we observe that any Hermitean vector bundle $E^- \to S^3$ has trivial first Chern class since $\rmH^2(S^3, \RZ) = 0$.
Hence Proposition~\ref{st:1st_Chern_class_knows_everything_in_dleq3} implies that $E^-$ is trivialisable.
This implies, in particular, that $\QVBdl_U(S^3, \tau^3) = \QVBdl(S^3,
\tau^3)$ just as in the case of Bloch electrons, and the corresponding Quaternionic K-theory is~\cite{Kitaev:Periodic_table,Freed-Moore--TEM,Thiang2014} $\rmK\rmQ(S^3,\tau^3)\cong\RZ\oplus\ZZ$.
Consequently, if $(E^-, \hat{\tau}^3)$ is any Quaternionic vector bundle
on $S^3$, arising for instance as a Bloch bundle, we can define the
Kane-Mele invariant in this situation in complete analogy with the
periodic case via \eqref{eq:KM_invariant}. 
As $\pi_1(S^3) = 0$, the determinant $\det \circ w_\varphi$ does not wind around the $1$-cycle of $\sfU(m)$ and we can in fact consider $w_\varphi$ to be $\sfS\sfU(m)$-valued after some (not necessarily equivariant) homotopy.

In Example~\ref{eg:S3 as filtered involutive manifold} we constructed
fundamental domains $F^d$ for the $3$-dimensional oriented filtered
involutive manifold $\{(S^d,\tau^d,F^d)\}_{d=0,1,2,3}$ with $F^d \simeq \FR^d$ for $d = 1,2,3$.
Therefore Theorem~\ref{st:localisation in Z2-HdR} yields a localisation technique for the Kane-Mele invariant $\rmK\rmM(E^-, \hat{\tau}^3)$ on $(S^3, \tau^3)$ as well.
The localisation can be carried out in a particularly easy fashion here:
for $d = 3,2,1$ the long exact sequence~\eqref{eq:LES for fundamental domain} decomposes into exact sequences
\begin{equation}
\begin{tikzcd}[row sep = 1cm, column sep = 1.5cm]
	0 \ar[r] & \rmH^{d-1}_{\dR\, (-1)^{d-1}}(S^{d-1}) \ar[r, "\Delta^{d, d-1}"] \ar[d, "\cong"] & \rmH^d_{\dR\, (-1)^d}(S^d) \ar[r] \ar[d, "\cong"] & 0
	\\
	& \rmH^{d-1}_\dR(S^{d-1}) & \rmH^d_\dR(S^d) &
\end{tikzcd}
\end{equation}
The vertical isomorphisms follow from Example~\ref{eg:top-degree coho
  and Z2 actions} below.
Hence in this case the relevant connecting homomorphism is an isomorphism in every degree which by Corollary~\ref{st:preimages under Delta and integrals} commutes with integration.
Again we write the class of the Wess-Zumino-Witten $3$-form as $[w_\varphi^*
H_{\sfU(m)}] = \upsilon\, [\vol_{S^3}] $, where $\upsilon\in\RZ$ and
$\vol_{S^3}$ is a unit volume form on $S^3$, so that $[w_\varphi^*
H_{\sfU(m)}] = \upsilon$ in $\rmH^3_\dR(S^3) \cong \FR$. Then choosing
$\rho^d = \upsilon \, \vol_{S^d}$ for $d = 1,2$, and the angle
function $\theta$ on $F^1 \subset S^1$ such that $\dd \theta =
\rho^1_{|F^1}$, yields $\rho^0(N) = \upsilon$ and $\rho^0(S) = 0$. Thus
we again obtain $\rmK\rmM(E^-, \hat{\tau}^3)=(-1)^\upsilon$ as required.
\qen
\end{example}

\section{The Kane-Mele invariant as a bundle gerbe holonomy}
\label{sect:holonomies_and_localisation}

\subsection{Bundle gerbes and their holonomy}
\label{sect:BGrbs}

Here we recall the basic notions of bundle gerbes with connections and
their morphisms, refering
to~\cite{Waldorf--More_morphisms,Waldorf--Thesis,Murray-Stevenson:BGrbs--Stable_iso_and_local_theory}
for further details. Consider a surjective submersion $\pi \colon Y \to X$ of manifolds.
We set $Y^{[p]} \coloneqq Y\times_X \cdots\times_X Y:= \{ (y_0,y_1, \ldots, y_{p-1}) \in Y^p\, |\,
\pi(y_0) = \pi(y_1)= \cdots= \pi(y_{p-1}) \}$, and we write $d_i \colon Y^{[p]} \to Y^{[p-1]}$ for the map that forgets the $i$-th entry of a $p$-tuple.
Defining $s_i \colon Y^{[p-1]} \to Y^{[p]}$ to be the map that duplicates the $i$-th entry, we turn the collection of $Y^{[p]}$ into a simplicial manifold $\check{N}_\bullet Y = Y^{[\bullet + 1]}$.%
\footnote{As the notation suggests, this is often called the \emph{\v{C}ech nerve} of the surjective submersion $\pi \colon Y \to X$.}

\begin{definition}[Bundle gerbe]
\label{def:bundle gerbes}
Let $X$ be a manifold.
A \emph{Hermitean bundle gerbe} on $X$ is a quadruple $(L, \mu, Y, \pi)$ of the following data:
the map $\pi \colon Y \to X$ is a surjective submersion, $L \to Y^{[2]}$ is a Hermitean line bundle, and $\mu \colon d_2^*L \otimes d_0^*L \to d_1^*L$ is an isomorphism of Hermitean line bundles over $Y^{[3]}$ satisfying an associativity condition over $Y^{[4]}$.
\\
A \emph{connection on a Hermitean bundle gerbe} is a pair $(\nabla^L,
B)$ consisting of a Hermitean connection $\nabla^L$ on the line bundle
$L$ together with a $2$-form $B \in \Omega^2(Y)$, called a
\emph{curving}, such that $\mu$ is parallel and $F_{\nabla^L} = d_0^*B
- d_1^*B$, where $F_{\nabla^L}$ is the curvature $2$-form of
$\nabla^L$. There exists a unique $3$-form $H_{\CG}\in\Omega^3_{\cl}(X)$ such that $\dd B=\pi^*H_{\CG}$~\cite{Murray:Bundle_gerbes}, and we call
$H_{\CG}$ the \emph{curvature} of $\CG=(L,\nabla^L,\mu,B,Y,\pi)$.
\end{definition}

\begin{definition}[1-isomorphism]
\label{def:1-morphism of BGrbs}
Let $\CG_i = (L_i, \nabla^{L_i}, \mu_i, B_i, Y_i, \pi_i)$ be Hermitean bundle gerbes with connection on $X$ for $i = 0,1$.
A \emph{$1$-isomorphism} $\CG_0 \to \CG_1$ is a quintuple $(J, \nabla^J, \alpha, Z, \zeta)$ of the following data:
the map $\zeta \colon Z \to Y_0 \times_X Y_1$ is a surjective submersion.
Write $\zeta_{Y_i} \coloneqq \pr_{Y_i} \circ \zeta \colon Z \to Y_i$
and $\zeta_{Y_i}^{[2]}: Z^{[2]}\to Y_i^{[2]}$.
Then $J \to Z$ is a Hermitean line bundle with Hermitean connection $\nabla^J$ such that $F_{\nabla^J} = \zeta_{Y_1}^*B_1 - \zeta_{Y_0}^*B_0$,
and $\alpha \colon \zeta_{Y_0}^{[2]\,*}L_0 \otimes d_0^*J \to d_1^*J
\otimes \zeta_{Y_1}^{[2]\,*}L_1$ is a parallel unitary isomorphism of
line bundles over $Z^{[2]}$ that intertwines $\mu_0$ and $\mu_1$ in
this sense.
\end{definition}

We will sometimes abbreviate $1$-isomorphisms $(J, \nabla^J, \alpha, Z,
\zeta)$ as $(J, \alpha)$, and simply refer to them as isomorphisms of bundle gerbes when no confusion is possible.

\begin{definition}[2-isomorphism]
\label{def:2-morphism of BGrbs}
Let $\CG_i = (L_i, \nabla^{L_i}, \mu_i, B_i, Y_i, \pi_i)$ be Hermitean bundle gerbes with connection on $X$ for $i = 0,1$.
Let $(J, \nabla^J, \alpha, Z, \zeta)$, $(J', \nabla^{J'}, \alpha', Z', \zeta'\,)$ be $1$-isomorphisms of bundle gerbes.
Set $Y_{01} = Y_0 \times_X Y_1$.
A \emph{$2$-isomorphism} $(J, \alpha) \Rightarrow (J',
\alpha'\,)$ is an equivalence class of triples $(W, \gamma, \phi)$,
where $\gamma \colon W \to Z \times_{Y_{01}} Z'$ is a surjective
submersion and $\phi \colon \gamma_Z^*J \to \gamma_{Z'}^*J'$ is a
parallel unitary isomorphism of line bundles intertwining $\alpha$ and
$\alpha'$; we have set $\gamma_Z = \pr_Z \circ \gamma$ and $\gamma_{Z'} = \pr_{Z'} \circ \gamma$.
Two such triples $(W, \gamma, \phi)$ and $(W', \gamma', \phi'\,)$ are \emph{equivalent} if there exist surjective submersions $\widehat{W} \to W$ and $\widehat{W} \to W'$ such that the pullbacks of $\phi$ and $\phi'$ agree on $\widehat{W}$.
\end{definition}

We will sometimes abbreviate $2$-isomorphisms $[W,\gamma,\phi]$ as $\phi$.

It is shown in~\cite{Waldorf--More_morphisms} that bundle gerbes with
connection, together with these notions of morphisms, assemble into
the structure of a $2$-category~\cite{Leinster--Basic_bicategories}; we denote this $2$-category by $\BGrb^\nabla(X)$.
With the conventions we have adopted here, all $1$-morphisms are invertible up to $2$-isomorphism, and all $2$-morphisms are isomorphisms.%
\footnote{There also exist noninvertible morphisms of bundle gerbes, and these are important in other contexts such as, to name just a few, the description of $B$-fields on D-branes in string theory~\cite{CJM--Holonomy_on_D-branes,Kapustin--D-branes_in_top_nontriv_B-field}, twisted K-theory~\cite{BCMMS}, and higher geometric quantisation~\cite{BSS--HGeoQuan,Bunk-Szabo--Fluxes_brbs_2Hspaces}, but we will not need them here.}
In other words, $\BGrb^\nabla(X)$ is in fact a $2$-groupoid. The $2$-groupoid $\BGrb^\nabla(X)$ is furthermore a Picard $2$-groupoid; that is, it carries a symmetric monoidal structure $\otimes$ with respect to which every object has an inverse.
The corresponding notions of morphisms of bundle gerbes without connection are obtained by forgetting all connections in Definitions~\ref{def:1-morphism of BGrbs} and~\ref{def:2-morphism of BGrbs}; we denote the corresponding $2$-groupoid by $\BGrb(X)$.

\begin{example}
\label{ex:basicbundlegerbe}
There exists a canonical bundle gerbe with connection $\CG_{\rm basic}$ on $\sfS\sfU(m)$ for any $m \geq 2$ called the \emph{basic bundle gerbe}.
Several constructions of this bundle gerbe and their generalisations to
the unitary group $\sfU(m)$ are known, see for
example~\cite{Gawedzki--2DFKM_and_WZW,Waldorf--Thesis,Murray-Stevenson:Basi_BGrb_on_unitary_groups}. It
will be important for us to know that its curvature $3$-form $H_{\CG_{\rm basic}} =
H_{\sfS\sfU(m)}$ coincides with the canonical $3$-form on $\sfS\sfU(m)$.

Recall from Section~\ref{sect:KM as obstruction to factorisation} that we are interested in maps $w_\varphi \colon X\to \sfS\sfU(m)$ which arise from trivialisations of Quaternionic vector bundles $(E^-, \hat{\tau})$ on a closed and orientable involutive $3$-manifold $(X, \tau)$ with orientation-reversing involution.
These Quaternionic vector bundles could, for example, be Bloch bundles of time-reversal symmetric topological insulators as considered in Section~\ref{sect:pure_time_reversal_in_TEM}.
As we can pull back surjective submersions and bundles, there exists a pullback operation on bundle gerbes with connection.
Thus the map $w_\varphi \colon X \to \sfS\sfU(m)$ from~\eqref{eq:trivialisations_of_E} canonically induces a bundle gerbe $w_\varphi^*\CG_{\rm basic} \in \BGrb^\nabla(X)$ on $X$.
\qen
\end{example}

For a $2$-groupoid $\scC$, we let $\pi_0 \scC$ denote the collection of $1$-isomorphism classes of objects.
For a manifold $X$ we let $\hat{\rmH}^p(X, \RZ)$ denote its $p$-th differential cohomology group as modelled, for instance, by Deligne cohomology; for details see~\cite{Brylinski--Loop_spaces_and_geometric_quantisation,Waldorf--More_morphisms,Murray-Stevenson:BGrbs--Stable_iso_and_local_theory}.
We then have the classification results~\cite{Waldorf--More_morphisms,Murray-Stevenson:BGrbs--Stable_iso_and_local_theory}

\begin{theorem}
\label{st:Classification of BGrbs}
There are natural isomorphisms of abelian groups
\begin{equation}
\hat\rmD \colon \pi_0 \, \BGrb^\nabla(X) \longrightarrow \hat{\rmH}^3(X,\RZ) \qquad \mbox{and} \qquad \rmD \colon \pi_0 \, \BGrb(X) \longrightarrow \rmH^3(X,\RZ)\ ,
\end{equation}
called the \emph{Deligne class} and the \emph{Dixmier-Douady class}, respectively.
\end{theorem}

\begin{remark}
Analogously to what happens in Chern-Weil theory for line bundles, the image of $\rmD(\CG)$ under the homomorphism $\rmH^3(X,\RZ)\to\rmH_{\dR}^3(X)$ coincides with the class of the curvature $3$-form $[H_{\CG}]$. In particular, if $\rmH^3(X,\RZ)$ is torsion-free then $\rmD(\CG)=[H_{\CG}]$.
\qen
\end{remark}

\begin{definition}[Trivialisation, dual bundle gerbe]
Let $X$ be a manifold and let $\CG = (L, \nabla^L, \mu, B, Y, \pi)$ be
a bundle gerbe with connection on $X$.
\begin{myenumerate}
	\item For $\rho \in \Omega^2(X)$, the \emph{trivial bundle
            gerbe with connection $(\dd,\rho)$} is $\CI_\rho = (X
          \times \FC, \dd, \, \cdot \, , \rho, X, 1_X)$.
	That is, it is the bundle gerbe whose surjective submersion to $X$ is the identity on $X$, whose line bundle is the trivial line bundle $X \times \FC$ with its canonical Hermitean metric and connection, whose bundle gerbe multiplication $\mu\big((x,z)\otimes(x,z'\,)\big)=(x,z\,z'\,)$ is given by multiplication in $\FC$, and whose curving is $\rho$.
	
	\item A \emph{trivialisation} of $\CG$ is a $1$-isomorphism $(T,\alpha) \in \BGrb^\nabla(X)(\CG, \CI_\rho)$ for some $\rho \in \Omega^2(X)$.
	
	\item The \emph{dual bundle gerbe with connection} of $\CG$ is given by
	\begin{equation}
		\CG^* = \big( (L, \nabla^L)^*, \mu^{-\trans}, -B, Y, \pi \big)\ ,
	\end{equation}
	where $(L, \nabla^L)^*$ is the dual Hermitean line bundle with
        connection.
\end{myenumerate}
\end{definition}

The trivial bundle gerbe $\CI_\rho$ has trivial Dixmier-Douady class,
$\rmD(\CI_\rho)=0$, while the dual bundle gerbe $\CG^*$ has the inverse Deligne class of $\CG$, $\hat\rmD(\CG^*) = - \hat\rmD(\CG)$.
The operation of taking the dual extends to an involutive functor on the $2$-category of bundle gerbes as introduced in Definition~\ref{def:bundle gerbes} which acts contravariantly on $1$-isomorphisms and covariantly on $2$-isomorphisms~\cite{Waldorf--Thesis}.
As for Hermitean line bundles with connection, this functor acts as a weak inverse with respect to the tensor product of bundle gerbes.
It can be shown~\cite{Waldorf--More_morphisms} that for any two trivialisations $(T,\alpha) \colon \CG \to \CI_\rho$ and $(T', \alpha'\,) \colon \CG \to \CI_{\rho'}$ the difference $2$-form $\rho' - \rho$ is closed and has integer periods.
In fact, we have~\cite{Waldorf--More_morphisms}

\begin{proposition}
\label{st:Isomps of trivial gerbes and HLBdls}
For any $\rho, \rho' \in \Omega^2(X)$, there is an equivalence of categories
\begin{equation}
	\sfR \colon \BGrb^\nabla(X)(\CI_\rho, \CI_{\rho'}) \arisom \HLBdl^\nabla_{\rho'-\rho}(X)\ ,
\end{equation}
where $\HLBdl^\nabla_{\rho'-\rho}(X)$ is the category of Hermitean line bundles with connection on $X$ whose curvature equals $\rho' - \rho$.
Under this equivalence, the composition of morphisms of trivial bundle gerbes is mapped to the tensor product of Hermitean line bundles with connection on $X$.
\end{proposition}

\begin{remark}
\label{rmk:star on 1-isos does not map to dual of LBdls}
Care has to be taken regarding duals of 1-isomorphisms:
we have $\sfR ((T,\beta)^*) \cong \sfR (T,\beta)$ and $\sfR ((T,\beta)^{-1}) \cong (\sfR (T,\beta))^*$.
A different involution on $\BGrb^\nabla(X)$ which reproduces the dual of line bundles under the reduction functor $\sfR$ is introduced in~\cite{BSS--HGeoQuan}.
\qen
\end{remark}

\begin{definition}[Bundle gerbe holonomy]
Let $\Sigma$ be a closed oriented surface and let $\CG$ be a bundle
gerbe with connection on $\Sigma$.
By Theorem~\ref{st:Classification of BGrbs} there exists a trivialisation $(T,\alpha) \colon \CG \to \CI_\rho$ on $\Sigma$.
The complex number
\begin{equation}
	\hol(\CG, \Sigma) \coloneqq \exp \Big( 2\pi\, \iu\, \int_\Sigma\, \rho
        \Big) \ \in \ \sfU(1)
\end{equation}
is called the \emph{holonomy} of $\CG$ around $\Sigma$; it is
independent of the choice of trivialisation.
\end{definition}

One can show~\cite{Waldorf--Thesis,Gawedzki:Topological_actions} that
for any compact oriented $3$-dimensional manifold $X^3$ and any bundle
gerbe $\CG$ on $X^3$ one has
\begin{equation}
\label{eq:BGrb hol and curvature}
	\hol(\CG, \partial X^3) = \exp \Big( 2\pi\, \iu\, \int_{X^3}\,
        H_{\CG} \Big)\ .
\end{equation}

\subsection{Jandl gerbes}
\label{sect:Real BGrbs}

Next we briefly recall the notion of Jandl structures on bundle
gerbes over involutive manifolds
from~\cite{Waldorf--Thesis,GSW:BGrbs_for_orientifold_Sigma_models}.\footnote{A
related but slightly simpler notion of Real structures on bundle gerbes
is introduced in~\cite{Hekmati:2016ikh}, however the theory of
connections in that setting has not yet been developed so we stick to
the original notion of Jandl structure.}
Similarly to the case of Real vector bundles, a Jandl structure on a bundle gerbe $\CG$ over an involutive manifold $(X,\tau)$ is given by a $1$-isomorphism $(A,\alpha) \colon \tau^* \CG \to \CG^*$.
In the case of bundle gerbes, however, we can now compare the $1$-isomorphisms $\tau^*(A,\alpha)$ and $(A,\alpha)^*$, and their composition is related to the identity $1$-isomorphism only weakly, i.e. by the additional data of a $2$-isomorphism.

We denote vertical composition in a $2$-category $\scC$ by
$\bullet$. We borrow the following definition from~\cite{Waldorf--Thesis,SSW:Unoriented_WZW_and_hol_of_BGrbs}.

\begin{definition}[Jandl gerbe]
\label{def:Real BGrb}
Let $(X, \tau)$ be an involutive manifold and let $\CG \in \BGrb^\nabla(X)$ be a bundle gerbe with connection on $X$.
A \emph{Jandl structure} on $\CG$ consists of a pair $((A,\alpha),\psi)$ of a $1$-isomorphism $(A,\alpha) \colon \tau^*\CG
\to \CG^*$ and a $2$-isomorphism $ \psi \colon
\tau^*(A,\alpha) \Rightarrow (A,\alpha)^*$ filling the diagram
\begin{equation}
\begin{tikzcd}[row sep=1cm, column sep=1.5cm]
	|[alias=A]| \tau^*(\tau^*\CG) \ar[r, "{\tau^*(A,\alpha)}"] \ar[d, equal] & \tau^*\CG^*
	\\
	\CG \ar[ur, "{(A,\alpha)^*}"', ""{name=U,inner sep=1pt,above}] & 
	\arrow[Rightarrow, from=A, to=U, "\psi"]
\end{tikzcd}
\end{equation}
and satisfying $\psi^* \bullet \tau^*\psi = 1_{(A,\alpha)}$.
A bundle gerbe with connection together with a choice of a Jandl structure is called a \emph{Jandl gerbe} on $(X,\tau)$.
\end{definition}

\begin{example}
\label{eg:Jandl structure on basic gerbe}
By~\cite[Proposition 4.2.4]{Waldorf--Thesis}, any bundle gerbe with
connection on $\sfS\sfU(m)$ admits exactly two nonisomorphic Jandl
structures. This applies, in particular, to the basic bundle gerbe $\CG_{\rm basic} \in \BGrb^\nabla(\sfS\sfU(m))$ introduced in Example~\ref{ex:basicbundlegerbe}.
See~\cite{Waldorf--Thesis,SSW:Unoriented_WZW_and_hol_of_BGrbs} for details.
\qen
\end{example}

Let us now make the following important observation:
let $(\CG, (A,\alpha), \psi)$ be a Jandl gerbe on $(X,\tau)$, and let $(T,\beta) \colon \CG \to \CI_\rho$ be a trivialisation of $\CG$.
This gives rise to a $1$-isomorphism $K(A,T) \colon \CI_{-\rho} \to \CI_{\tau^*\rho}$ defined by the diagram
\begin{equation}
\label{diag:LBdl from triv of Real gerbe}
\begin{tikzcd}[row sep=1.25cm, column sep=2cm]
	\CG^* \ar[r, "{(A,\alpha)^{-1}}"] & \tau^*\CG \ar[d, "{\tau^*(T, \beta)}"]
	\\
	\CI_{-\rho} \ar[u, "{(T,\beta)^*}"] \ar[r, dashed, "{K(A,T)}"'] & \CI_{\tau^*\rho}
\end{tikzcd}
\end{equation}
Using the equivalence from Proposition~\ref{st:Isomps of trivial gerbes and HLBdls}, this induces a Hermitean line bundle with connection $\sfR K(A,T) \in \HLBdl^\nabla_{\tau^*\rho + \rho}(X)$.

The involution $\tau$ acts on $K(A,T)$ as
\begin{equation}
\begin{aligned}
	\tau^*K(A,T)
	&= \tau^* \big( \tau^*(T, \beta) \circ (A,\alpha)^{-1} \circ (T,\beta)^* \big)
	\\[4pt]
	&= (T,\beta) \circ \tau^*(A,\alpha)^{-1} \circ \tau^*(T,
        \beta)^*\ .
\end{aligned}
\end{equation}
On the other hand
\begin{equation}
\begin{aligned}
	K(A,T)^*
	&= \big( \tau^*(T, \beta) \circ (A,\alpha)^{-1} \circ (T,\beta)^* \big)^*
	\\[4pt]
	&= (T,\beta) \circ (A,\alpha)^{-*} \circ \tau^*(T, \beta)^*\ .
\end{aligned}
\end{equation}
The image of the $2$-isomorphism $\psi$ under inversion of $1$-isomorphisms, composed horizontally by the respective identity $2$-isomorphisms, induces a $2$-isomorphism $K(\psi) \colon \tau^*K(A,T) \Rightarrow K(A,T)^*$, and the additional condition on $\psi$ implies that $K(\psi) \bullet \tau^*K(\psi) = 1_{K(A,T)}$.
Thus under the equivalence from Proposition~\ref{st:Isomps of trivial
  gerbes and HLBdls}, and with Remark~\ref{rmk:star on 1-isos does not
  map to dual of LBdls} in mind, we obtain an equivariant line bundle
$(\sfR K(A,T), \sfR K(\psi))$ on $X$ with curvature $F_{\sfR K(A,T)} = \tau^*\rho + \rho$.

In order to describe the Kane-Mele invariant, we consider a Jandl gerbe $(\CG, (A,\alpha), \psi)$ on a $2$-dimensional closed and oriented involutive manifold $(X^2, \tau^2)$ with orientation-preserving involution $\tau^2$.
Let $F^2$ be a fundamental domain for the involution, and set $X^1 \coloneqq \partial F^2$.
Then the restriction $\tau^2_{|X^1}$ is orientation-reversing.
We observe that
\addtocounter{equation}{1}
\begin{align*}
\label{eq:From gerbe hol to LBdl hol}
	\hol \big( \sfR K(A,T), X^1 \big)
	&= \exp \Big( 2\pi\, \iu\, \int_{F^2}\, F_{\sfR K(A,T)} \Big)
	\\[0.2cm]
	&= \exp \Big( 2\pi\, \iu\, \int_{F^2}\, \big( \tau^{2\,*}\rho + \rho\big) \Big) \tag{\theequation}
	\\[0.2cm]
	&= \exp \Big(2\pi\, \iu\, \int_{X^2}\, \rho \Big)
	\\*[0.2cm]
	&= \hol \big(\CG, X^2\big)\ .
\end{align*}
This shows, in particular, that any Jandl gerbe $(\CG, (A,\alpha), \psi)$ on an involutive surface $(X^2,\tau^2)$ as above gives rise to a holonomy $\hol(\sfR K(A,T), X^1)$ which is defined independently of the choice of trivialisation $(T,\beta)$ of $\CG$, of the choice of fundamental domain $F^2$ of $X^2$, and even of the explicit choice of Jandl structure on $\CG$.

Furthermore, the Jandl structure on $\CG$ together with the orientation-preserving action of $\tau^2$ on $X^2$ implies
\begin{equation}
	\hol(\CG, X^2)
	= \hol (\tau^{2\,*}\CG, X^2)
	= \hol(\CG^*, X^2)
	= \hol(\CG, X^2)^{-1} \ ,
\end{equation}
and hence $\hol(\CG, X^2) \in \ZZ$.

Note how closely the result of the calculation~\eqref{eq:From gerbe hol to LBdl hol} resembles the statement of Corollary~\ref{st:preimages under Delta and integrals}.
Here we have derived a geometric refinement of this statement:
while in Corollary~\ref{st:preimages under Delta and integrals} we reduced integrals of differential forms to integrals over the boundary of a fundamental domain, here instead we reduce the holonomy of a bundle gerbe with additional data to the holonomy of a line bundle constructed from this data over the boundary of a fundamental domain.
We summarise our findings here as

\begin{proposition}
\label{st:Jandl hol to LBdl hol}
Let $(X^2, \tau^2)$ be a closed and oriented involutive manifold of dimension $d=2$ with orientation-preserving involution $\tau^2$.
Let $(\CG, (A, \alpha), \psi)$ be a Jandl gerbe on $(X^2, \tau^2)$.
Then:
\begin{myenumerate}
	\item The holonomy of $\CG$ over $X^2$ squares to $1$, i.e. $\hol(\CG, X^2) \in \ZZ$.
	
	\item For any choice of trivialisation $(T,\beta)$ of $\CG$
          over $X^2$ and any choice of fundamental domain $F^2$ for
          $\tau^2$ with boundary $X^1 \coloneqq \partial F^2$, the
          holonomy of the $\ZZ$-equivariant Hermitean line bundle $\sfR K(A,T)$ on $X^2$ is $\hol(\sfR K(A,T), X^1) = \hol(\CG, X^2)$.
	That is, the holonomy of $\sfR K(A,T)$ around $X^1$ computes the holonomy of $\CG$ around $X^2$.
\end{myenumerate}
\end{proposition}

It is natural to ask at this stage if there exists a localisation
technique akin to that in Section~\ref{sect:Localisation via MV} but
based on a hierarchy of higher bundles with connection
(i.e. functions, line bundles, bundle gerbes, $\dots$) rather than a hierarchy of differential forms.
In fact, the Mayer-Vietoris sequence from Theorem~\ref{st:LES for open subset} almost generalises to this refined setting:
consider an open cover $X = V \cup V'$ and set $U = V \cap V'$.
Given a $\sfU(1)$-valued function $g_U$ on $U$ we obtain a Hermitean line bundle $L$ on $X$ by defining its transition function with respect to the cover $X = V \cup V'$ to be $g_U$.
A connection on $L$ is given by $A_{|V} \coloneqq f_{V'}\, g_U^{-1}\, \dd g_U$ and $A_{|V'} \coloneqq -f_V\, g_U^{-1}\, \dd g_U$, where $(f_V,f_{V'})$ is a partition of unity subordinate to the cover $X = V \cup V'$.
We would like to think of this assignment from $\sfU(1)$-valued
functions to Hermitean line bundles with connection as the analogue of
the connecting homomorphism in degree zero in a long exact sequence
similar to~\eqref{eq:LES for open subset}, but now involving Deligne
cohomology rather than de~Rham cohomology.
However, if $g_V$ and $g_{V'}$ are $\sfU(1)$-valued functions on $V$
and $V'$, respectively, and we consider $g_U = g_V\, g_{V'}^{-1}$,
then we obtain a trivialisable line bundle $L$, but because of the
partition of unity used in the construction of its connection, it is not trivialisable as a line bundle with connection.
That is, while the long exact sequence~\eqref{eq:LES for open subset} might lift to a categorical sequence of higher bundles, it does not lift to higher bundles with connections.
Despite the absence of a long exact sequence of geometric objects in
this sense, we will show in Section~\ref{sect:KM from BGrb hol} below
that an analogue of the localisation technique for the Kane-Mele
invariant can nevertheless be obtained in the higher bundle setting as well.

\subsection{Localisation of bundle gerbe holonomy}
\label{sect:KM from BGrb hol}

Putting together the discussion from Section~\ref{sect:Real BGrbs} with the definition of the Kane-Mele invariant from~\eqref{eq:KM_invariant}, we obtain a geometric perspective on its localisation.
Let $\{(X^d, \tau^d,F^d)\}_{d=0,1,2,3}$ be an oriented filtered
involutive $3$-manifold; recall that filtered involutive manifolds
have nonempty fixed point sets $X^0$ by definition. Let $(E^-, \hat{\tau}^3) \in \QVBdl_U(X^3, \tau^3)$.
Recall from Section~\ref{sect:KM as obstruction to factorisation} that the Kane-Mele invariant of $(E^-, \hat{\tau}^3)$ can be written as
\begin{equation}
\label{eq:KM_invariantS}
	\rmK\rmM(E^-, \hat{\tau}^3)
	= \exp \bigg( \pi\, \iu\, \int_{X^3}\, w_\varphi^*H_{\sfS\sfU(m)} \bigg)\ ,
\end{equation}
where we may assume that $w_\varphi \colon X^3 \to \sfS\sfU(m)$ after some smooth homotopy.
As before, using the involution $\tau^3$ this can be rewritten as
\begin{equation}
	\rmK\rmM(E^-, \hat{\tau}^3)
	= \exp \bigg( 2 \pi\, \iu\, \int_{F^3}\, w_\varphi^*H_{\sfS\sfU(m)} \bigg)
\end{equation}
for any choice of fundamental domain $F^3 \subset X^3$ for $\tau^3$.
From~\eqref{eq:BGrb hol and curvature} we recognise this as the
holonomy of a pullback of the basic bundle gerbe $\CG_{\rm basic}$ from Example~\ref{ex:basicbundlegerbe} around the boundary $X^2 \coloneqq \partial F^3$:
\begin{equation}
	\rmK\rmM\big(E^-, \hat{\tau}^3\big)
	= \hol\big( w_\varphi^* \CG_{\rm basic}, X^2 \big)\ .
\end{equation}
With a Jandl structure $((A,\alpha), \psi)$ on $\CG_{\rm basic}$ chosen (see Example~\ref{eg:Jandl structure on basic gerbe}), Proposition~\ref{st:Jandl hol to LBdl hol} shows that the data $w_\varphi^*(\CG_{\rm basic}, (A,\alpha), \psi)$ contains information about equivariant Hermitean line bundles on $X^2$ that compute the same holonomy on $X^1 \coloneqq \partial F^2$ for any choice of fundamental domain $F^2$ for $\tau^2 = \tau^3_{|X^2}$:
\begin{equation}
	\rmK\rmM\big(E^-, \hat{\tau}^3\big)
	= \hol\big( \sfR K(A,T), X^1 \big)\ .
\end{equation}

We can now repeat similar steps that we used for Jandl gerbes in Section~\ref{sect:Real BGrbs} to further localise the holonomy of the line bundles $\sfR K(A,T)$.
Let $F^1$ be a fundamental domain for $\tau^1 \coloneqq \tau^2_{|X^1}$.
For any such choice, $X^0 \coloneqq \partial F^1$ is the set of fixed points of the original $\ZZ$-action on $X^3$.
By dimensional reasons, there exist trivialisations of $\sfR K(A,T)$ as Hermitean line bundles on $X^1$.
Let $j \colon \sfR K(A,T) \to I$ be such a trivialisation with $I=X^1\times\FC$ denoting the trivial Hermitean line bundle.
The isomorphism $j$ does not generally trivialise the connection, but rather induces an isomorphism of Hermitean line bundles with connection of the form $j \colon \sfR K(A,T) \to I_\nu$ for some $\nu \in \Omega^1(X^1)$, where $I_\nu$ is the trivial line bundle with connection $\dd+\nu$.
Similarly to~\eqref{diag:LBdl from triv of Real gerbe}, we obtain an isomorphism of Hermitean line bundles with connection defined by the diagram
\begin{equation}
\begin{tikzcd}[row sep=1.25cm, column sep=1.5cm]
	\tau^{1\,*} \sfR K(A,T) \ar[r, "{\sfR K(\psi)}"] & \sfR K(A,T) \ar[d,"j"]
	\\
	I_{\tau^{1\,*}\nu} \ar[u, "{\tau^{1\,*}j^{-1}}"] \ar[r, dashed, "f_{(A,T,j)}"'] & I_\nu
\end{tikzcd}
\end{equation}
Here $f_{(A,T,j)}$ is an isomorphism of trivial Hermitean line bundles with connection, or equivalently a function $f_{(A,T,j)} \colon X^1 \to \sfU(1)$.
As a consequence of the equivariance of $\sfR K(A,T)$, it satisfies
\begin{equation}
\begin{aligned}
	\tau^{1\,*} f_{(A,T,j)}
	&= \tau^{1\,*} \big( j \circ \sfR K(\psi) \circ \tau^{1\,*}j^{-1} \big)
	\\[0.1cm]
	&= \tau^{1\,*} j \circ \tau^{1\,*} \sfR K(\psi) \circ j^{-1}
	\\[0.1cm]
	&= f_{(A,T,j)}^{-1}\ .
\end{aligned}
\end{equation}
Moreover, since $f_{(A,T,j)}$ is an isomorphism of line bundles with connection we have
\begin{equation}
	f_{(A,T,j)}^{-1}\, \dd f_{(A,T,j)}
	= \nu - \tau^{1\,*}\nu\ .
\end{equation}
Consequently, making use of the fact that $\tau^1$ is orientation-reversing, we derive once again that the Kane-Mele invariant is completely localised onto the fixed point set:
\addtocounter{equation}{1}
\begin{align}
\label{eq:KM localisation from holonomies}
	\rmK\rmM\big(E^-, \hat{\tau}^3\big)
	&=\hol \big( \sfR K(A,T), X^1 \big)
	\nonumber\\[0.2cm]
	&= \exp \Big(2\pi\, \iu\, \int_{X^1}\, \nu \Big)
	\nonumber\\[0.2cm]
	&= \exp \Big( 2\pi\, \iu\, \int_{F^1}\, \big(\nu - \tau^{1\,*} \nu\big) \Big) \tag{\theequation}
	\\[0.2cm]
	&= \exp \Big( 2\pi\, \iu\, \int_{F^1}\, f_{(A,T,j)}^{-1}\, \dd f_{(A,T,j)} \Big)
	\nonumber\\*[0.2cm]
	&= \prod_{k \in X^0}\, f_{(A,T,j)}(k)\ . \nonumber
\end{align}
This geometric localisation property strongly resembles the expression~\eqref{eq:Localised HdR form of KM}, but the procedure itself is richer: it uses geometric objects which can be understood as realising the forms appearing in Section~\ref{sect:Localisation in equivar HdR} as their curvatures.

From the analysis thus far, we can distill the following statements.
Firstly, there is an assertion purely about bundle gerbes on involutive manifolds given by

\begin{theorem}
Let $(X^2, \tau^2)$ be a closed and oriented $2$-dimensional
involutive manifold that can be made the top componentwise involutive triple of an oriented $2$-dimensional filtered involutive manifold $\{(X^d, \tau^d, F^d)\}_{d=0,1,2}$.
Let $(\CG, (A,\alpha), \psi)$ be a Jandl gerbe on $X^2$.
Then:
\begin{myenumerate}
	\item The holonomy of $\CG$ around $X^2$ squares to $1$:
	\begin{equation}
		\hol(\CG, X^2) \ \in \ \ZZ\ .
	\end{equation}
	
	\item The holonomy localises completely as
	\begin{equation}
	\label{eq:Real BGrb holonomy reduction}
		\hol\big(\CG, X^2\big)
		= \hol \big( \sfR K(A,T), X^1 \big)
		= \prod_{k \in X^0}\, f_{(A,T,j)}(k)\ ,
	\end{equation}
	where $X^0 \subset X^2$ is the set of fixed points of the involution $\tau^2$, which is independent of the extension of $(X^2, \tau^2)$ to an oriented filtered involutive manifold.
	The expression~\eqref{eq:Real BGrb holonomy reduction} is independent of the choice of Jandl structure and all trivialisations.
\end{myenumerate}
\end{theorem}

\begin{proof}
This follows immediately from our analysis above.
\end{proof}

Secondly, denoting the canonical pairing between cohomology and homology classes by
$\<-,-\>$, we can use our findings here to obtain a geometric version of our results from Section~\ref{sect:KM as obstruction to factorisation} as

\begin{theorem}
\label{st:KM via BGrbs}
Let $(X^3, \tau^3)$ be a closed connected and oriented involutive $3$-manifold
with orientation-reversing involution $\tau^3$, let $(E^-, \hat{\tau}^3)
\in \QVBdl_U(X^3, \tau^3)$ be a trivialisable Quaternionic vector
bundle on $X^3$, and let $\CG \in
\BGrb(X^3)$ be a bundle gerbe on $X^3$. Then:
\begin{myenumerate}
	\item There exists a bundle gerbe $\CH \in \BGrb(X^3)$ and an
          isomorphism $\CG \cong \tau^{3\,*}\CH \otimes \CH^*$ if and
          only if $\exp \big( \pi\, \iu\, \< \rmD(\CG), [X^3]\> \big) = 1$.
	
	\item $\rmK\rmM(E^-, \hat{\tau}^3) = 1$ if and only if for any trivialisation $\varphi$ of $E^-$ such that $w_\varphi$ is $\sfS\sfU(m)$-valued there exists a bundle gerbe $\CH \in \BGrb(X^3)$ with $w_\varphi^*\CG_{\rm basic} \cong \tau^{3\,*}\CH \otimes \CH^*$.
\end{myenumerate}
\end{theorem}

\begin{proof}
Statement (1) is seen as follows:
first, since $X^3$ is closed, connected and oriented we have $\rmH^3(X^3, \RZ) \cong \RZ$.
The fact that $\tau^3$ reverses the orientation of $X^3$ is then
equivalent to $\tau^{3\,*}$ acting as multiplication by $-1$ on $\rmH^3(X^3, \RZ)$.
Assuming that $\CG \cong \tau^{3\,*}\CH \otimes \CH^*$, we thus
compute the Dixmier-Douady class from Theorem~\ref{st:Classification
  of BGrbs} to get
\begin{equation}
\begin{aligned}
	\rmD(\CG)
	&= \rmD(\tau^{3\,*}\CH \otimes \CH^*)
	\\[4pt]
	&= \tau^{3\,*}\,\rmD(\CH) - \rmD(\CH)
	\\[4pt]
	&= -2\, \rmD(\CH)\ .
\end{aligned}
\end{equation}
Consequently $\exp \big( \pi\, \iu\, \< \rmD(\CG), [X^3] \>\big) = 1$.

Now assume that $\exp \big( \pi\, \iu\, \< \rmD(\CG), [X^3] \>\big) = 1$, i.e. that $\rmD(\CG) \in 2\,\RZ$.
By Theorem~\ref{st:Classification of BGrbs} we can then find a bundle gerbe $\CH \in \BGrb(X^3)$ such that $\rmD(\CH) = \frac{1}{2}\, \rmD(\CG)$ together with an isomorphism as required.

For (2), recall from~\eqref{eq:KM_invariantS} that
\begin{equation}
	\rmK\rmM\big(E^-,\hat\tau^3\big)
	= \exp \bigg( \pi\, \iu\, \int_{X^3}\, w_\varphi^* H_{\sfS\sfU(m)} \bigg)
	= \exp \big( \pi\, \iu\, \big\< \rmD(w_\varphi^*\CG_{\rm basic}),
        [X^3]\big\>\big) \ ,
\end{equation}
since $\rmD(\CG_{\rm basic})=[H_{\sfS\sfU(m)}]$.
This together with (1) proves the statement.
\end{proof}

\begin{remark}
By Theorem~\ref{st:Classification of BGrbs} and Corollary~\ref{st:3D
  CW see SUm as KZ3}, the group $\sfS\sfU(m)$ classifies bundle gerbes (without connection) on $X^3$ up to $1$-isomorphism.
Therefore for any $\CG \in \BGrb(X^3)$ there exists a map $\kappa \colon X^3 \to \sfS\sfU(m)$ and an isomorphism $\kappa^*\CG_{\rm basic} \cong \CG$.
The splitting of the bundle gerbe up to isomorphism (disregarding the
Jandl structure) in Theorem~\ref{st:KM via BGrbs} is thus equivalent
to the splitting of the sewing matrix $w_\varphi$ up to homotopy (disregarding
the $\ZZ$-action) in Theorem~\ref{st:splitting up to homotopy of w_varphi}.

In fact, statement (3) of Theorem~\ref{st:splitting up to homotopy of
  w_varphi}, which implies statement (2) of Theorem~\ref{st:KM via BGrbs}, can alternatively be derived from the classification in Theorem~\ref{st:Classification of BGrbs} together with the fact that the basic bundle gerbe $\CG_{\rm basic}$ on $\sfS\sfU(m)$ is \emph{multiplicative}~\cite{Waldorf:Multiplicative_BGrbs,Gawedzki-Waldorf:PW_formula_and_multiplicative_gerbes}.
This means that, if $\sfm \colon \sfS\sfU(m) \times \sfS\sfU(m) \to
\sfS\sfU(m)$ denotes the multiplication map on the group and $\pr_i
\colon \sfS\sfU(m) \times \sfS\sfU(m) \to \sfS\sfU(m)$, $i = 0,1$,
denote the projections onto the two factors, then there exists a
$1$-isomorphism of bundle gerbes (without connection)%
\footnote{This is reflected already in~\eqref{eq:mulitplication_and_canonical_3-form} at the level of Dixmier-Douady classes.}
\begin{equation}
	\sfm^*\CG_{\rm basic} \cong \pr_0^*\,\CG_{\rm basic} \otimes
        \pr_1^*\,\CG_{\rm basic}\ ,
\end{equation}
which satisfies coherence relations lifting the associativity and
unity in the group.
\qen
\end{remark}

We also have an involutive version of the relation between the
holonomy of a bundle gerbe and the integral of its curvature given by

\begin{proposition}
If $(X^2,\tau^2)$ can be embedded as the boundary of a fundamental
domain in a $3$-dimensional closed and oriented involutive manifold
$(X^3, \tau^3)$ with $\ori(\tau^3)=-1$, and $(\CG, (A,\alpha), \psi)$ is a Jandl gerbe
on $X^3$, then
\begin{equation}
	\hol\big(\CG, X^2\big)
	= \exp \Big( \pi\, \iu\, \int_{X^3}\, H_{\CG} \Big)
	= \exp \big( \pi\, \iu\, \big\< \rmD(\CG), [X^3] \big\> \big)
	\ \in \ \ZZ\ .
\end{equation}
In particular, the bundle gerbe holonomy $\hol(\CG, X^2)$ is independent of the choice of connection and Jandl structure on $\CG$, and if $X^3$ can be chosen connected then it is trivial if and only if there exists a bundle gerbe $\CH \in \BGrb(X^3)$ and an isomorphism of bundle gerbes without connection
\begin{equation}
	\CG \cong \tau^{3\,*}\CH \otimes \CH^* \ .
\end{equation}
\end{proposition}

\begin{proof}
This follows immediately from the computations and statements above.
\end{proof}

\subsection{Discrete versus local formulas}
\label{sect:unifiedKM}

The results obtained here apply straightforwardly to the examples from Section~\ref{sect:Localisation via MV}.
A localisation formula for the Kane-Mele invariant of the same
qualitative form as~\eqref{eq:Localised HdR form of KM}
and~\eqref{eq:KM localisation from holonomies} is already known in
terms of the Pfaffian of the $\sfS\sfU(m)$-valued sewing matrix (see for example~\cite{Fu-Kane:Time-reversal_polarisation,WQZ,Gawedzki:BGrbs_for_TIs,Freed-Moore--TEM}):
\begin{equation}
	\rmK\rmM(E^-, \hat{\tau}^3) = \prod_{k \in X^0}\, \pfaff \big( w_\varphi(k) \big)\ .
\end{equation}
When $X^3=\TT^3$, derivations of this discrete formula directly from the $3$-dimensional integral definition \eqref{eq:KM_invariant} can be found in~\cite{WQZ,KLWKKtheory}.
From our findings in this section and in Section~\ref{sect:Localisation in equivar HdR} it is evident that localisation formulas of this kind for the Kane-Mele invariant are not unique, but always allow for certain choices.
This is analogous to the nonuniqueness of localisation formulas in
standard equivariant localisation (see e.g.~\cite{Szabo:2000fs}).
To make contact with the Pfaffian formula one should take $\rho^0=\frac1{2\pi\,\iu}\, \log\circ \det\circ w_\varphi=\frac1{\pi\,\iu}\, \log\circ \pfaff\circ w_\varphi$ and $f_{(A,T,j)}=\exp\circ (\pi\,\iu\, \rho^0)$, but we have not found any canonical way to derive these choices from our localisation techniques. Our new explicit expressions for the localisation of the Kane-Mele
invariant may shed further light on the mathematical and physical
content of the invariant.

It is also interesting to compare our intermediate localisation formulas for the Kane-Mele invariant in terms of local quantities with previous approaches; see also~\cite{KLWKStiefel} for a similar dimensional hierarchy of decompositions
of the $3$-dimensional Kane-Mele invariant in terms of lower
dimensional indices. For the Brillouin torus $X^3=\TT^3$ with $X^2=\TT^2_0\sqcup\TT^2_\pi$, the $2$-dimensional localisation formulas
\begin{equation}
	\rmK\rmM\big(E^-, \hat{\tau}^3\big) = \exp\Big(\pi\,\iu\, \int_{X^2}\, \rho^2\Big) = \hol\big( w_\varphi^* \CG_{\rm basic}, X^2 \big)
\end{equation}
express the definition of the $3$-dimensional strong Kane-Mele invariant on $\TT^3$ as the difference of the $2$-dimensional weak Kane-Mele invariants over $\TT^2_0$ and $\TT^2_\pi$~\cite{FKM}. This is also pointed out by~\cite{Gawedzki:2017whj} in the setting of bundle gerbe holonomy, where by using the complicated expression for the holonomy in terms of a triangulation of the surface $X^2$ the explicit Pfaffian formula was derived by direct, though cumbersome, calculation. Here we have applied the $2$-category theory
of bundle gerbes and found a remarkably simple computation of the localisation formula for the Kane-Mele invariant over the time-reversal invariant crystal momenta that is independent of all choices involved, at the price of obtaining somewhat less explicit, but more general, formulas.

Finally, the $1$-dimensional localisation formulas
\begin{equation}
	\rmK\rmM\big(E^-, \hat{\tau}^3\big) = \exp\Big(\pi\,\iu\, \int_{X^1}\, \rho^1\Big) = \hol\big( \sfR K(A,T), X^1 \big)
\end{equation}
are analogous to the geometric formulas for the weak Kane-Mele invariants in terms of Berry phases~\cite{Fu-Kane:Time-reversal_polarisation}, though again we have not found any canonical choices equating the $1$-forms $\rho^1$ and $\nu$ with the trace of the Berry connection on $E^-$. The equivalence between the Berry phase formula of~\cite{Fu-Kane:Time-reversal_polarisation} and the bundle gerbe holonomy of~\cite{Gawedzki:BGrbs_for_TIs} is demonstrated explicitly by~\cite{Monaco-Tauber:Berry_WZW_FKM}.

We can elaborate further on these integral formulas in lower dimensions, which also yields a simple physical picture for the Kane-Mele invariant recasted as a winding number. Starting from the $3$-dimensional integral formula we write
\begin{equation}
\label{eq:KMwindingderiv}
	\rmK\rmM(E^-, \hat{\tau}^3)
	= \exp \bigg( \pi\, \iu\, \int_{\TT^3}\, w_\varphi^*H_{\sfS\sfU(m)} \bigg) = \exp\bigg(\, \int_{S^1}\, \Big(\pi\,\iu\, \int_{\TT^2_\theta}\, \iota_{\partial_\theta}\big(w_\varphi^*H_{\sfS\sfU(m)}\big)\Big)\ \dd\theta\, \bigg)
	\ .
\end{equation}
For each $\theta\in S^1$, let $\jmath_\theta:\TT^2_\theta\hookrightarrow\TT^3$ be the inclusion. 
Using transgression as in e.g.~\cite{BSS--HGeoQuan}, we can show that the holonomy around $\TT^2$ of an $S^1$-family of bundle gerbes
\begin{equation}
u(\theta):=\hol\big(\jmath_\theta^* w_\varphi^*\CG_{\rm basic}, \TT_\theta^2 \big)
\end{equation}
obeys
\begin{equation}
u(\theta)^{-1}\, \frac{\dd u(\theta)}{\dd\theta} = \int_{\TT^2_\theta}\, \iota_{\partial_\theta}\big(w_\varphi^*H_{\sfS\sfU(m)}\big) \ .
\end{equation}
Consequently the argument of the last exponential in
\eqref{eq:KMwindingderiv} is $\pi\,\iu$ times the winding number (or
degree) of the function $u:S^1\to S^1$, giving
\begin{equation}
\rmK\rmM(E^-, \hat{\tau}^3)
	= \exp \big( \pi\, \iu\ {\rm ind}(u)\big)
\end{equation}
and matching with the interpretation of the Kane-Mele invariant as a mod~$2$ index (see e.g.~\cite{KLWK}).
It would be interesting to physically interpret the 3-dimensional
Kane-Mele invariant in this form as arising from vortex lines extended
along $S^1$ and terminating in the bulk of the Brillouin zone, i.e.~in
the valence bands making up the Kramers pairs, along the lines of e.g.~\cite{Kruthoff:2017rjk}.
A rigorous vortex line/monopole interpretation of the Kane– Mele
invariant is given in~\cite{ThiangSatoGomi2017}.

\begin{appendix}

\section{Proof of Proposition~\ref{st:1st_Chern_class_knows_everything_in_dleq3}}
\label{app:proof_of_1st_Chern_knows_everything}

The idea of the proof of Proposition~\ref{st:1st_Chern_class_knows_everything_in_dleq3} is to find a $4$-equivalence $f \colon \sfB\sfU(m) \rightarrow \sfB\sfU(1)$.
The homotopy groups of $\sfU(m)$ read as~\cite[p.~151]{Nakahara}\footnote{The homotopy groups of $\sfU(m)$ are readily deduced from those of $\sfS\sfU(m)$ via the fibre sequence $\sfS\sfU(m) \to \sfU(m) \to \sfU(1)$.}
\begin{equation}
	\pi_k(\sfU(m)) =
	\begin{cases}
		\ \RZ\ , & k\ \text{odd}\ ,\\
		\ 0\ , & k\ \text{even}\ ,
	\end{cases}
\end{equation}
for $1<k \leq 2m - 1$.
The only nontrivial homotopy group of $\sfU(1) \simeq K(\RZ,1)$ is $\pi_1(\sfU(1)) = \RZ$, and both $\sfU(m)$ and $\sfU(1)$ are connected.
For any connected topological group $\sfG$ one generally has the
`delooping' property
\begin{equation}
	\pi_k(\sfB\sfG) \cong \pi_{k-1}(\sfG) 
\end{equation}
for all $k\geq1$. Thus for $m \geq 2$ we have $\pi_0(\sfB\sfU(m))
= 0=\pi_1(\sfB\sfU(m))$, $\pi_2(\sfB\sfU(m)) = \RZ$, $\pi_3(\sfB\sfU(m)) = 0$, and $\pi_4(\sfB\sfU(m)) = \RZ$, while the only nontrivial homotopy group of $\sfB\sfU(1)$ is $\pi_2(\sfB\sfU(1)) = \RZ$.
Consequently any continuous map $f \colon \sfB\sfU(m) \to \sfB\sfU(1)$ which induces an isomorphism $\pi_2(f) \colon \RZ \to \RZ$ is already a $4$-equivalence.

Set $\eta_m \coloneqq c_1(\sfE\sfU(m)) \in \rmH^2(\sfB\sfU(m),\RZ)$,
where $\sfE\sfG\to\sfB\sfG$ denotes the universal bundle over the
classifying space $\sfB\sfG$ of a topological group $\sfG$.
This is in fact a generator of $\rmH^2(\sfB\sfU(m),\RZ)$, as the cohomology ring $\rmH^\bullet(\sfB\sfU(m),\RZ)$ is a polynomial ring over the Chern classes of $\sfE\sfU(m)$~\cite[Theorem 16.10]{Switzer:AT}.
Equivalently $\eta_m = c_1(\det\,\sfE\sfU(m))$, which defines a map $\sfB\det \colon \sfB\sfU(m) \to \sfB\sfU(1)$ induced by the determinant on the unitary group $\sfU(m)$.
Then $\eta_m = \sfB\det^*\eta_1$ with $\eta_1 = c_1(\sfE\sfU(1))$.

Now we have a commutative diagram
\begin{equation}
\label{eq:pi_2_for_Bloch_trivialisation}
	\xymatrixrowsep{1.2cm}
	\xymatrixcolsep{1.2cm}
	\xymatrix{
		\pi_2(\sfB\sfU(m)) \ar@{->}[r]_-{\cong} \ar@{->}[d]_-{\pi_2(\sfB\det)} & \pi_0\,\HVBdl_{m\, \sim}(S^2) \ar@{->}[d]^-{c_1}\\
		\pi_2(\sfB\sfU(1)) \ar@{->}[r]^-{\cong} & \rmH^2(S^2,\RZ)
	}
\end{equation}
where the top row is an isomorphism of pointed sets by the defining properties of $\sfB\sfU(m)$, and the bottom row is an isomorphism by $\sfB\sfU(1) \simeq K(\RZ,2)$ and since $K(\RZ,2)$ classifies the second integer cohomology.
Our goal is to show that the left vertical map $\pi_2(\sfB\det)$ is an isomorphism.
This is equivalent to showing that the right vertical map is an isomorphism.

For this, first note that any Hermitean vector bundle on $S^2$ can be
expressed as a clutching construction with respect to the closed
locally finite cover $S^2 = S^2_+ \cup S^2_-$ with $S^2_\pm$
denoting the closed upper and lower hemispheres of $S^2$.
We denote their intersection, which is the equatorial circle of $S^2$, by $S^1_{+-} \hookrightarrow S^2$.
Then any Hermitean vector bundle $E \to S^2$ of rank~$m$ can be trivialised over
the contractible sets $S^2_\pm$, and all nontrivial data is contained
in a transition function $g \colon S^1_{+-} \to \sfU(m)$.
This induces a bijection $\pi_0\,\HVBdl_{m\, \sim}(S^2) \cong \pi_1(\sfU(m))$ (see~\cite[Statements 3.9, 3.14, 7.6]{Karoubi:KT}).
Here we can either compose with the determinant map $\det$, or instead
first form the determinant line bundle $\det(E)\to S^2$ and then trivialise it, to proceed further along $\pi_1(\det) \colon \pi_1(\sfU(m)) \to \pi_1(\sfU(1))$.
This map is in fact an isomorphism~\cite[Example~4D.7]{Hatcher:AT}.
Thus we have an isomorphism $\pi_0\,\HVBdl_{m\, \sim}(S^2) \cong [S^1_{+-},
\sfU(1)] \cong \pi_1(\sfU(1))$ given by assigning to a vector bundle
of rank $m$ 
on $S^2$ the homotopy class of the clutching function of its determinant line bundle.

By similar arguments we also obtain an isomorphism $\rmH^2(S^2,\RZ) \cong \pi_1(\sfU(1))$ by choosing a representative Hermitean line bundle for a class in $\rmH^2(S^2,\RZ)$, and taking an admissible clutching function.
Hence we can augment the commutative diagram~\eqref{eq:pi_2_for_Bloch_trivialisation} to
\begin{equation}
	\xymatrixrowsep{1.2cm}
	\xymatrixcolsep{1.7cm}
	\xymatrix{
		\pi_2(\sfB\sfU(m)) \ar@{->}[r]_-{\cong} \ar@{->}[d]_-{\pi_2(\sfB\det)} & \pi_0\,\HVBdl_{m\, \sim}(S^2) \ar@{->}[d]^-{c_1} \ar@{->}[dr]^-{\det} \ar@{->}[r]^-{\text{clutching}}_-{\cong} & \pi_1(\sfU(m)) \ar@{->}[r]_{\cong}^{\pi_1(\det)} & \pi_1(\sfU(1)) \ar@{=}[d]\\
		\pi_2(\sfB\sfU(1)) \ar@{->}[r]^-{\cong} & \rmH^2(S^2,\RZ) \ar@{->}[r]^-{\cong} & \pi_0\,\HLBdl_\sim(S^2) \ar@{->}[r]_-{\text{clutching}}^-{\cong} & \pi_1(\sfU(1))
	}
\end{equation}
That $\pi_1 (\det) \colon \pi_1(\sfU(m)) \to \pi_1(\sfU(1))$ is an
isomorphism is the crucial fact which implies that the diagonal map of
taking the determinant line bundle is bijective on isomorphism classes
of rank~$m$ vector bundles over $S^2$.
It follows that $c_1 \colon \pi_0\,\HVBdl_{m\, \sim}(S^2) \to \rmH^2(S^2,\RZ)$ is an isomorphism as required, so that $\pi_2(\sfB \det) \colon \pi_2(\sfB\sfU(m)) \to \pi_2(\sfB\sfU(1))$ is an isomorphism as well.
Consequently $\sfB \det \colon \sfB\sfU(m) \to \sfB\sfU(1)$ is a $4$-equivalence.

The final step is then provided by~\cite[Theorem 6.31]{Switzer:AT}, which implies that for any CW-complex $X$, the map $\sfB \det \circ (-) \colon [X,\sfB\sfU(m)] \to [X,\sfB\sfU(1)]$ is surjective if $\dim(X) = 4$ and bijective if $\dim(X) < 4$.
Using the horizontal isomorphisms
in~\eqref{eq:pi_2_for_Bloch_trivialisation}, this is equivalent to
$c_1 \colon \pi_0\,\HVBdl_{m\, \sim}(X) \to \rmH^2(X,\RZ)$ being
surjective for $\dim(X) = 4$ and an isomorphism of pointed sets for
$\dim(X) < 4$. These homotopy theoretic statements make sense in the
smooth setting as well, since continuous maps between smooth manifolds
can always be approximated by smooth
maps~\cite[Theorem~2.6]{Hirsch:Differential_Topology}. This completes the proof of
Proposition~\ref{st:1st_Chern_class_knows_everything_in_dleq3}.

\section{Equivariant de Rham cohomology}
\label{app:HdR and group actions}

Let $V$ be a $\bbK$-vector space, where $\bbK$ is either $\FR$ or $\FC$.

\begin{definition}
Let $\sfG$ be a Lie group acting smoothly from the right on a manifold $X$ via $(x,g) \mapsto R_g\, x \in X$ for $x \in X$ and $g \in \sfG$.
Let $\varrho \colon \sfG \to \End(V)$ be a representation of $\sfG$.
\begin{myenumerate}
	\item The space of \emph{$(\sfG, \varrho)$-equivariant
            $p$-forms} on $X$ is
	\begin{equation}
		\Omega^p(X,V)^{(\sfG, \varrho)} \coloneqq \big\{
                \omega \in \Omega^p(X,V)\, \big|\, R_g^*\,\omega =
                \varrho(g)^{-1} \, \omega \ \ \text{for all} \ g \in
                \sfG \big\}\ .
	\end{equation}
	
	\item The space of \emph{closed $(\sfG, \varrho)$-equivariant
            $p$-forms} on $X$ is
	\begin{equation}
		\Omega^p_\cl(X,V)^{(\sfG, \varrho)} \coloneqq
                \Omega^p(X,V)^{(\sfG, \varrho)} \cap
                \Omega^p_\cl(X,V)\ .
	\end{equation}
	
	\item The \emph{$p$-th $(\sfG, \varrho)$-equivariant de Rham
            cohomology group} of $X$ is
	\begin{equation}
		\rmH^p_{\dR\, (\sfG, \varrho)}(X, V) \coloneqq
                \frac{\Omega^p_\cl(X,V)^{(\sfG, \varrho)}}{\dd
                  \Omega^{p-1}(X,V)^{(\sfG, \varrho)}}\ .
	\end{equation}
\end{myenumerate}
\end{definition}

\begin{example}
\label{eg:Z_2 equivar and Real HdR}
The most common case encountered in this paper is $(\sfG, \varrho, V)
= (\ZZ, \varrho_\pm, \FR)$, where $\varrho_+$ is the trivial
representation of $\ZZ$ on $\FR$ while $\varrho_-$ is the
representation by multiplication; that is, $\varrho_-(1) = 1$ and $\varrho_-(-1) = -1$.
We introduce the short-hand notation
\begin{equation}
\begin{aligned}
	\Omega^p(X,\FR)^{(\ZZ, \varrho_\pm)} &\eqqcolon \Omega^p_\pm(X)\ ,
	\\[4pt]
	\Omega^p_\cl(X,\FR)^{(\ZZ, \varrho_\pm)} &\eqqcolon
        \Omega^p_{\cl\, \pm}(X)\ ,
	\\[4pt]
	\rmH^p_{\dR\, (\ZZ, \varrho_\pm)}(X, \FR) &\eqqcolon
        \rmH^p_{\dR\, \pm}(X)\ ,
\end{aligned}
\end{equation}
which is used in the main text of the paper.
\qen
\end{example}

If $\sfG$ is compact, we can establish a stronger version
of~\cite[Theorem 2.1]{CE--Cohomology_of_Lie_groups_and_Lie_algebras}
given by

\begin{proposition}
\label{st:averaging and splitting equivariant HdR}
Let $\sfG$ be a compact Lie group and $\varrho \colon \sfG \to \End(V)$ a real representation of $\sfG$.
There is a short exact sequence of complexes of $\FR$-vector spaces
\begin{equation}
\label{eq:SES for averaging morphism}
\begin{tikzcd}[column sep=1.25cm]
	0 \ar[r] & \ker\big(\av_{(\sfG, \varrho)}^\bullet\big) \ar[r]
        & \Omega^\bullet(X,V) \ar[r, "\av_{(\sfG, \varrho)}^\bullet"]
        & \Omega^\bullet(X,V)^{(G, \varrho)} \ar[r] & 0\ ,
\end{tikzcd}
\end{equation}
which splits and thus induces a canonical isomorphism
\begin{equation}
\label{eq:HdRXViso}
	\rmH^\bullet_\dR(X,V) \cong \rmH^\bullet_{\dR\,
          \av_{(\sfG,\varrho)}}(X,V) \oplus \rmH^\bullet_{\dR\,
          (\sfG,\varrho)}(X,V)\ ,
\end{equation}
where we have set
\begin{equation}
	\rmH^\bullet_{\dR\, \av_{(\sfG,\varrho)}}(X,V) \coloneqq
        \frac{\ker\big(\av^\bullet_{(\sfG, \varrho)}\big) \cap
          \Omega^\bullet_\cl(X,V)}{\dd \ker\big(\av^\bullet_{(\sfG,
            \varrho)}\big)}\ .
\end{equation}
\end{proposition}

\begin{proof}
Let $\nu_{\sfG}$ denote the normalised Haar measure on $\sfG$.
Define the averaging morphism on $p$-forms by
\begin{equation}
	\av^p_{(\sfG, \varrho)} \colon \Omega^p(X,V) \longrightarrow
        \Omega^p(X,V)^{(G, \varrho)}\ ,
	\quad
	\omega \longmapsto \int_\sfG\, \varrho(g) \big(R_g^*\, \omega\big)\,
        \dd \nu_{\sfG}\ .
\end{equation}
Since pullbacks commute with the exterior derivative, the collection $\av^\bullet_{(\sfG, \varrho)}$ assembles into a morphism of cochain complexes $\av^\bullet_{(\sfG, \varrho)} \colon \Omega^\bullet(X,V) \to \Omega^\bullet(X,V)^{(G, \varrho)}$.
This yields the short exact sequence~\eqref{eq:SES for averaging morphism}.

The short exact sequence~\eqref{eq:SES for averaging morphism} is split by the inclusion morphism
$\jmath^\bullet_{(\sfG, \varrho)} \colon \Omega^\bullet(X,V)^{(G, \varrho)}
\hookrightarrow \Omega^\bullet(X,V)$ which forgets the
$(G,\varrho)$-equivariant structure.
Consequently there are canonical isomorphisms
\begin{equation}
	\Omega^\bullet(X,V) \cong \ker\big(\av^\bullet_{(\sfG,
          \varrho)}\big) \oplus \Omega^\bullet(X,V)^{(G, \varrho)} \ .
\end{equation}
Applying the Zig-Zag Lemma to~\eqref{eq:SES for averaging morphism},
or by simply observing that this decomposition is compatible with the
exterior derivative, then yields canonical decompositions \eqref{eq:HdRXViso}.
\end{proof}

\begin{example}
\label{eg:top-degree coho and Z2 actions}
In the situation of Example~\ref{eg:Z_2 equivar and Real HdR}, Proposition~\ref{st:averaging and splitting equivariant HdR} yields canonical decompositions
\begin{equation}
	\rmH^\bullet_\dR(X) \cong \rmH^\bullet_{\dR\, +}(X) \oplus
        \rmH^\bullet_{\dR\, -}(X)\ .
\end{equation}
Consider an involutive manifold $(X, \tau)$ with $\dim(X) = d$ and $X$ connected, closed and orientable.
In this case, integration over $X$ yields an isomorphism
\begin{equation}
	\textint_X \colon \rmH^d_\dR(X) \longrightarrow \FR\ .
\end{equation}
Combining this with $\int_X \circ \tau^* = \ori(\tau) \int_X$, we derive
\begin{equation}
	\rmH^d_{\dR\, \ori(\tau)}(X) \cong \FR
	\qquad \text{and} \qquad
	\rmH^d_{\dR\, -\ori(\tau)}(X) \cong 0\ .
\end{equation}
As a concrete example, this applies to the $d$-dimensional torus $X = \TT^d$ with $\tau$ the inversion as considered in Section~\ref{sect:pure_time_reversal_in_TEM}.
Here we have $\ori(\tau) = (-1)^d$, and therefore we derive
\begin{equation}
\begin{aligned}
	d\ \text{even}\ &: \quad
	\rmH^d_{\dR\, +}(\TT^d) \cong \FR
	\qquad \text{and} \qquad
	\rmH^d_{\dR\, -}(\TT^d) \cong 0 \ ,
	\\[4pt]
	d\ \text{odd}\ &: \quad
	\rmH^d_{\dR\, +}(\TT^d) \cong 0
	\qquad \text{and} \qquad
	\rmH^d_{\dR\, -}(\TT^d) \cong \FR\ .
\end{aligned}
\end{equation}
Thus in these cases we can compute Real and equivariant de Rham cohomology groups quite easily.
\qen
\end{example}

\end{appendix}

\newpage

\bibliographystyle{amsplainurl}

\providecommand{\bysame}{\leavevmode\hbox to3em{\hrulefill}\thinspace}
\providecommand{\MR}{\relax\ifhmode\unskip\space\fi MR }
\providecommand{\MRhref}[2]{%
  \href{http://www.ams.org/mathscinet-getitem?mr=#1}{#2}
}
\providecommand{\href}[2]{#2}

\end{document}